\documentclass{CSML}

\def\dOi{12(4:6)2016}
\lmcsheading%
{\dOi}
{1--51}
{}
{}
{Aug.~31, 2015}
{Dec.~28, 2016}
{}

\ACMCCS{[{\bf Theory of computation}]: Models of computation;
  Semantics and reasoning; [{\bf Information systems}]:  World Wide
  Web---Web services---Service discovery and interfaces} 

\let\keywords\relax
\usepackage{pslatex}
\usepackage{amsmath}
\usepackage{amssymb}
\usepackage{hyperref}
\usepackage{graphicx}
\usepackage{url}
\usepackage{stmaryrd}
\usepackage{wasysym}
\usepackage{epsfig,color,moreverb,multirow}
\usepackage[noend]{algpseudocode}
\usepackage{mathrsfs}

\usepackage{dsfont}

\usepackage{tikz}
\usetikzlibrary{arrows,automata,decorations.pathmorphing,snakes,positioning}
\usepackage[latin1]{inputenc}
\usepackage{verbatim}
\usepackage{eurosym}
 
\usepackage{caption}
\usepackage{thm-restate}

\urldef{\mailsa}\path|basile@di.unipi.it|    
\newcommand{\keywords}[1]{\par\addvspace\baselineskip 
\noindent\keywordname\enspace\ignorespaces#1}

\usepackage{amsmath}

\newcommand{\ithel}[2]{{#1}_{({#2})}}
\newcommand{\blk}{\mbox{\tiny$\boxempty$}}
\newcommand{\Oset}{\mathbb{O}}
\newcommand{\Rset}{\mathbb{R}}

\newcommand{\matchrel}{\bowtie}

\newcommand{\conctrans}[1]{\OPconctrans}

\DeclareMathOperator{\OPconctrans}{\blackdiamond}
\newcommand{\blackdiamond}{\diamondsuit}

\renewcommand{\epsilon}{\varepsilon}

\newcommand{\hole}[0]{\bullet}

\newcommand{\irule}[2]{\frac{\textstyle\rule[-1.3ex]{0cm}{3ex}#1}%
{\textstyle\rule[-.5ex]{0cm}{3ex}#2}}

\newcommand{\hbra}{
\hbox to 1 \textwidth{\vrule width0.3mm height 1.8mm depth-0.3mm
                    \leaders\hrule height1.8mm depth-1.5mm\hfill
                    \vrule width0.3mm height 1.8mm depth-0.3mm}}
\newcommand{\hket}{
\hbox to 1 \textwidth{\vrule width0.3mm height1.5mm
                    \leaders\hrule height0.3mm\hfill
                    \vrule width0.3mm height1.5mm}}

\newcommand{\qst}{\;\colon\;} 

\newcommand{\A}{\bold{A}}

\newcommand{\Abot}{\bold{A^\bot}}

\newcommand{\La}{\bold{L}}

\newcommand{\illmix}{$ILL^{mix}$}

\begin{document}

\title[Automata for Specifying and Orchestrating Service Contracts]
      {Automata for Specifying and Orchestrating Service Contracts}

\thanks{This work has been partially supported by the projects  \emph{Security Horizons}, funded by MIUR, and PRA\_2016\_64 \emph{Through the fog}, funded by the University of Pisa.}

%
%
\author[D.~Basile]{Davide Basile}
\author[P.~Degano]{Pierpaolo Degano}
\author[G.-L.~Ferrari]{Gian-Luigi Ferrari}
\address{Dipartimento di Informatica, Universit\`{a} di Pisa, Italy}
\email{\{basile, degano, giangi\}@di.unipi.it}


%
%


\begin{abstract}
  An approach to the formal description of service contracts is
  presented in terms of automata.  We focus on the basic property of
  guaranteeing that in the multi-party composition of principals each
  of them gets his requests satisfied, so that the overall composition
  reaches its goal.  Depending on whether requests are satisfied
  synchronously or asynchronously, we construct an orchestrator that
  at static time either yields composed services enjoying the required
  properties or detects the principals responsible for possible
  violations.  To do that in the asynchronous case we resort to Linear
  Programming techniques.  We also relate our automata with two
  logically based methods for specifying contracts.

\end{abstract}
\maketitle


\section{Introduction}  

%

Modern software applications are not \emph{stand-alone} entities and are embedded in a dynamic distributed environment where new functionalities are added or deleted in a relatively short period of time.  \emph{Service Oriented Computing}~\cite{soc} is a paradigm for designing distributed applications where applications are built by combining several \emph{fine-grained} and \emph{loosely-coupled} distributed components, called \emph{services}. Services can be  combined to accomplish a certain computational task or to form a more complex service.
A service exposes both the functionalities it provides and the parameters it requires. Clients exploit service public information to discover and bind the services that better fit their requirements.

Service coordination is a fundamental mechanism of the service-oriented approach because it dictates 
how the involved services are compositionally assembled together.  
Service coordination policies differ on the interaction supports that are adopted to pass information among services. 
At design time, a main task of software engineers is therefore to express the assumptions that shape these policies and that will drive the construction of a correct service coordination.
\emph{Orchestration} and \emph{choreography}  are  the  standard solutions  to coordinate distributed services.
In an orchestrated approach, services
coordinate with each other by interacting with a distinguished service, the \emph{orchestrator}, which at run-time regulates how the computation evolves.
In a choreographed approach, the distributed services autonomously
execute and interact with each other without a central coordinator.
Here, we concentrate on orchestration, whereas we injected some aspects of our proposal within the choreographed approach in~\cite{basile2014,basile2015}.
    
We argue that the design of correct service coordination policies is naturally supported by relying on the notion of \emph{service contract} 
which specifies what a service is going to guarantee and offer (hereafter an \emph{offer}) and what in turn it expects and requires (hereafter a \emph{request}). 
The coordination policy  has therefore to define the duties and responsibilities 
for each of the different services involved in the coordination through the \emph{overall contract agreement}.
Obviously, this arrangement is based on the contracts of the involved services, 
and ensures that all requests are properly served when all the duties are properly kept.
The coordinator then organises the service coordination policy and proposes 
the resulting overall contract agreement to all the parties.
This process is called \emph{contract composition}. 

The main contribution of this paper is twofold. 
First, we propose a rigorous formal technique for describing and composing contracts, suitable to be automated. 
Second, we develop techniques capable of determining when a  contract composition is correct and leads to the design of a correct service orchestration.
More in detail, we introduce an automata-based model for contracts called \emph{contract automata}, that are a special kind of finite state automata, endowed with two operations for composing them.
A contract automaton may represent a single service or a composition of several services, hereafter called \emph{principals}.
The traces accepted by a contract automaton show the possible interactions among the principals, by recording which offers and requests are performed, and by which  principals in the composition. 
This provides the basis to define criteria that guarantee a composed service to well behave with respect to the overall service contract.

We equip our model with formal notions in language-theoretic terms aiming at characterising when contracts are honoured within a service composition. 
We first consider properties of a single trace.
We say that a trace is in \emph{agreement} when all the requests made are synchronously matched, i.e. satisfied by corresponding offers. 
The second property, \emph{weak agreement}, is more liberal, in that requests can be asynchronously matched, and an offer can be delivered even before a corresponding request, and vice-versa.
Then we say that a contract automaton is \emph{safe} (\emph{weakly safe}, respectively) when \emph{all} its traces are in agreement (weak agreement, respectively).

The notions of safety presented above may appear too strict since they require that all the words belonging to the language recognised by a contract automaton must satisfy agreement or weak agreement. 
We thus introduce a more flexible notion that characterises when a service composition may be successful, i.e. at least one among all the possible traces enjoys one of the properties above. 
We say that a contract automaton \emph{admits} (weak) agreement when such a trace exists. 

When a contract automaton admits (weak) agreement, but it is not (weakly) safe, we define those principals in a contract that are (weakly) \emph{liable}, i.e.\ those responsible for leading a contract composition into a failure.
Note that the orchestration of contracts imposes further constraints on each principal: some of the interactions dictated by its service contract may break the overall composition and thus the orchestrator will ban them.

For checking when a contract automaton enjoys the properties sketched above, we propose two formal verification techniques that have been also implemented~\cite{cat}.%
\footnote{Available at \url{https://github.com/davidebasile/workspace}}
The first one amounts to build the so-called controllers in Control Theory~\cite{Cassandras2006}. 
We show that controllers are powerful enough to synthesise a correct orchestrator enforcing agreement and to detect the liable principals.
In order to check weak agreement and detect weak liability we resort to Linear Programming techniques borrowed from Operational Research~\cite{Hemmecke10}, namely optimisation of network flows. 
The intuitive idea is that service coordination is rendered as an optimal  flow itinerary of offers and requests in a network, automatically constructed from the contract automaton.

Finally, we establish correspondence results between (weak) agreement and provability of formulae in two fragments of different intuitionistic logics, that have been used for modelling contracts. 
The first one, Propositional Contract Logic~\cite{BartolettiZ10}, has a special connective to deal with circularity between offers and requests, arising when a principal requires, say $a$, before offering $b$ to another principal who in turn first requires $b$ and then offers $a$; note that weak agreement holds for this kind of circularity.
The second fragment, Intuitionist Linear Logic with Mix~\cite{benton1995mixed} is a linear logic capable of modelling the exchange of resources with the possibility of recording debts, that arise when the request of a principal is satisfied and  not yet paid back.

\subsection*{Plan of the paper.}
In Section~\ref{sect:model} we introduce contract automata and two operators of composition.
Section~\ref{sect:agreement} discusses the properties of agreement and safety. 
The techniques for checking and enforcing them are also presented here, along with the notion of liability.
Weak agreement and weak liability are defined in Section~\ref{sect:weak-agreement}, along with a technique to check them.
 In Section~\ref{sect:logic} we present correspondence results with fragments of Propositional Contract Logic  and Intuitionistic Linear Logic with Mix.
A case study is proposed in Section~\ref{sect:casestudy}.
Finally, related work is in Section~\ref{sect:conclusion} and the concluding remarks are in Section~\ref{sect:concludingremarks}.
All the proofs of our results, and a few auxiliary definitions can be found in the appendix.
Portions of Sections~\ref{sect:model}, \ref{sect:agreement}, and~\ref{sect:weak-agreement} appeared in a preliminary form in~\cite{BasileDF14}.

%
\section{The Model}
\label{sect:model}

This section formally introduces the notion of contract automata, that are finite state automata with a partitioned alphabet.
%
A contract automaton represents the behaviour of a set of principals (possibly a singleton)
capable of performing some \emph{actions}; more precisely, the actions of contract
automata allow them to \lq\lq make\rq\rq\ requests, \lq\lq advertise\rq\rq\ offers
or \lq\lq matching\rq\rq\ a pair of ``complementary'' request/offer.
The number of principals in a contract automaton is called \emph{rank}, and we use a vectorial representation to record the action performed by each principal in a transition of a contract automaton, as well as its state as the vector of the states of its principals.

Let $\Sigma = \Rset  \cup \Oset \cup \{ \blk \}$ be the alphabet of \emph{basic actions}, made of
\emph{requests} $\Rset=\{a,b,c,\ldots\}$ and  \emph{offers} $\Oset=\{\overline{a}, \overline{b}, \overline{c}, \ldots \}$  where $\Rset \cap \Oset = \emptyset$, 
and $\blk \not \in \Rset \cup \Oset$ is a distinguished element representing the \emph{idle} move. 
We define the involution $co(\hole): \Sigma \mapsto \Sigma$ such that $ co(\Rset)=\Oset, \  co(\Oset)=\Rset, \  co(\blk)=\blk $.

Let $\vec{v}=(a_1,\tiny{...},a_{n})$ be a vector of \emph{rank} $n\geq 1$, in symbols $r_v$, and let $\ithel{\vec{v}} i$ denote the i-th element with $1 \leq i \leq r_v$.
We write $\vec{v}_1\vec{v}_2\ldots\vec{v}_m$ for the concatenation of $m$ vectors $\vec{v}_i$,   while $|\vec v| = n$ is the rank (length) of $\vec v$ and
 ${\vec v}^n$ is the vector obtained by $n$  concatenations of $\vec v$. 
%

The alphabet of a contract automaton consists of vectors, each element of which intuitively records the activity, i.e.\ the occurrence of a basic action of a single principal in the contract.
In a vector $\vec{v}$ there is either a single offer or a single request, or a single pair of request-offer that matches, i.e.\ there exists
exactly  $i, j$ such that   $\ithel{\vec{v}} i$ is an offer and  $\ithel{\vec{v}} j$ is the complementary request or vice-versa; all the other elements of the vector contain the symbol $\blk$, meaning that the corresponding principals stay idle. In the following let $\blk^{m}$  denote a vector of rank $m$, all elements of which are $\blk$.
Formally:
\begin{defi}[Actions]~\label{def:actions}
Given a vector $\vec{a} \in \Sigma^n$, if 
\begin{itemize}
\item $\vec{a} = \blk^{n_1} \alpha \blk^{n_2}, n_1,n_2 \geq 0$, then $\vec{a}$ is a \emph{request (action)  on }$\alpha$ if $\alpha \in \Rset$, and is an \emph{offer (action)  on }$\alpha$  if $\alpha \in \Oset$
\item $\vec{a}=\blk^{n_1}\alpha\blk^{n_2}co(\alpha)\blk^{n_3}, n_1,n_2,n_3 \geq 0$, then $\vec{a}$ is a \emph{match (action) on }$\alpha$, where $\alpha \in \Rset \cup \Oset$.
\end{itemize}
\noindent 

Two actions  $ \vec a$ and $\vec b$ are \emph{complementary}, in symbols $ \vec a \matchrel \vec b$ if and only if 
the following conditions hold:  
 (i)  $\exists \alpha \in \Rset \cup \Oset \qst \vec a$ is either a request or an offer on $\alpha$;
   (ii) $\vec a \text{ is an
      offer on } \alpha \implies \vec b \text{ is a request on } co(\alpha)$ and
 (iii) $\vec a \text{ is a
      request on } \alpha \implies \vec b \text{ is an offer on } co(\alpha)$.  
%

  %
\end{defi}

We now extract from an action the request or offer made by a principal, and the matching of a request and an offer, and then we lift this procedure to a sequence of actions, i.e.\ to a trace of a contract automaton that intuitively corresponds to an execution of a service composition.

\begin{defi}[Observable]
Let $w=\vec{a_1}\ldots \vec{a_n}$ be a sequence of actions, and let $\epsilon$ be the empty one, then its \emph{observable} is given by the partial function $Obs(w) \in (\Rset \cup \Oset \cup \{ \tau\})^\ast$ where:
\begin{itemize}
\item[]
$Obs(\epsilon)= \epsilon $

\item[]
$
Obs(\vec{a}\,w') =
 \left\{ 
\begin{array}{ll} 
 \ithel{\vec{a}} i\, Obs(w')& \mbox{if\ \ } \vec{a}\text{ is an offer/request and }\ithel{\vec{a}} i \neq \blk \\ 
 \tau\, Obs(w')  & \text{if\ \ }\vec{a}\text{ is a match}
\end{array} 
\right.
$
\end{itemize}
\end{defi} 


We now define contract automata, the actions and the states of which are actually vectors of basic actions and of states of principals, respectively.

\begin{defi}[Contract Automata]
\label{def:contract}
Assume as given a finite set of states $\mathfrak{Q}=\{q_1,q_2, \ldots \}$. 
Then a \emph{contract automaton} $\mathcal{A}$, CA for short, of rank $n$ is a tuple $\langle Q, \vec{q_0}, A^{r}, A^{o}, T, F \rangle$, where
\begin{itemize}
\item $Q=Q_1 \times \ldots \times Q_n \subseteq \mathfrak{Q}^n$
\item  $\vec{q_0} \in Q$ is the initial state
\item $A^{r}\subseteq \Rset, A^{o} \subseteq \Oset$ are finite sets (of requests and offers, respectively)
\item $F \subseteq Q$ is the set of final states 
\item  $T \subseteq Q \times A  \times Q$ is the set of transitions, where $A \subseteq ( A^{r} \cup A^{o} \cup \{\blk\})^n$ and if \\ $ (\vec{q},\vec{a},\vec{q'}) \in T$ then both the following conditions hold:
\begin{itemize}
\item $\vec{a}$ is either a request or an offer or a match
\item $\forall i \in 1\ldots n$. if $\ithel{\vec{a}} i=\blk$ then it must be $\ithel{\vec{q}} i=\ithel{\vec{q'}} i$
\end{itemize}
\end{itemize}
A \emph{principal} contract automaton (or simply \emph{principal}) has rank 1 and it is such that $A^r \cap co(A^o)= \emptyset$.\\
A step  $( w,\vec{q}) \rightarrow (w',\vec{q'})$ occurs if and only if $w=\vec{a}w', w' \in A^*$ and $(\vec{q},\vec{a},\vec{q'}) \in T$.
\\
The language of $\mathcal{A}$ is $\mathscr{L}(\mathcal{A})=\{w  \mid (w,\vec{q_0}) \rightarrow^* (\epsilon, \vec{q}), \vec{q} \in F\}$ where $\rightarrow^*$ is the reflexive, transitive closure of the transition relation $\rightarrow$.

\end{defi}

Note that for principals we have the restriction $A^r \cap co(A^o)= \emptyset$. Indeed, a principal
who offers what he requires makes little sense.

\begin{exa}
Figure~\ref{fig:automa_rank} shows three contract automata.
The automaton $\mathcal{A}_1$ may be understood as producing a certain number of resources through one or more offers $\overline{res}$ and it terminates with the request of receiving a signal $sig$. 
The contract  $\mathcal{A}_2$ starts by sending the signal $\overline{sig}$ and then it collects the resources produced by $\mathcal{A}_1$. 
The contract $\mathcal{A}_3$ represents the contract automaton where $\mathcal{A}_1$ and $\mathcal{A}_2$ interact as discussed below. 
Both $\mathcal{A}_1$ and $\mathcal{A}_2$ have rank 1 while $\mathcal{A}_3$ has rank 2. 
\end{exa}

Contract automata can be composed, by making the cartesian product of their states and of the labels of the joined transitions, with the additional possibility of labels recording matching request-offer.
This is the case for the action $(sig,\overline{sig})$ of the contract automaton $\mathcal{A}_3$ in Figure~\ref{fig:automa_rank}.


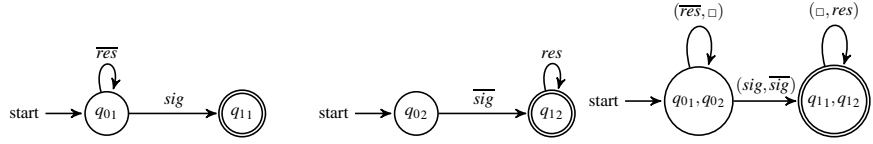
\begin{figure}[tb]
\center
\begin{tikzpicture}[->,>=stealth',shorten >=1pt,auto,node distance=3.0cm,
                    semithick, every node/.style={scale=0.6}]
  \tikzstyle{every state}=[fill=white,draw=black,text=black]

  \node[initial,state] (A)                    {${q_0}_1$};
  \node[state,accepting]         (B) [right of=A] {${q_1}_1$};

  \path (A)			edge             node{$sig$} (B)
 				     edge [loop above]             node {$\overline{res}$} (A);
\end{tikzpicture}\hspace{15pt}
\begin{tikzpicture}[->,>=stealth',shorten >=1pt,auto,node distance=3.0cm,
                    semithick, every node/.style={scale=0.6}]
  \tikzstyle{every state}=[fill=white,draw=black,text=black]

  \node[initial,state] (A)                    {${q_0}_2$};
  \node[state,accepting]         (B) [right of=A] {${q_1}_2$};

  \path (A)			edge             node{$\overline{sig}$} (B)
 		 (B)		     edge [loop above]             node {$res$} (A);
\end{tikzpicture}
\begin{tikzpicture}[->,>=stealth',shorten >=1pt,auto,node distance=3.0cm,
                    semithick, every node/.style={scale=0.6}]
  \tikzstyle{every state}=[fill=white,draw=black,text=black]

  \node[initial,state] (A)                    {${q_0}_1,{q_0}_2$};
  \node[state,accepting]         (B) [right of=A] {${q_1}_1,{q_1}_2$};

  \path (A)			edge             node{$(sig,\overline{sig})$} (B)
 				     edge [loop above]             node {$(\overline{res},\blk)$} (A)
		(B) edge [loop above]              node {$(\blk,res)$} (B);
\end{tikzpicture}
\caption{Three contract automata: from left $\mathcal{A}_1,\mathcal{A}_2$, and $\mathcal{A}_3$ (composition of 
$\mathcal{A}_1$ and $\mathcal{A}_2$)}
\label{fig:automa_rank}
\end{figure}

Below, we introduce two different operators for composing contract automata. Both products interleave all the transitions of their operands. 
We only force a synchronisation to happen when two contract automata are ready on their respective request/offer action. These operators represent two different policies of orchestration.
The first operator is called simply \emph{product} and it considers the case when a service $S$ joins a group of services already clustered as a single orchestrated service $S'$.
In the product of $S$ and $S'$, the first can only accept the still available offers (requests, respectively) of $S'$ and vice-versa.
In other words, $S$ cannot interact with the principals of the orchestration $S'$, but only with it as a whole component.
This is obtained in Definition~\ref{def:prod} through the relation $\bowtie$ (see Definition~\ref{def:actions}), which is only defined for actions that are not matches.
This is not the case with the second operation of composition, called \emph{a-product}: it puts instead all the principals of $S$ at the same level of those of $S'$.
Any matching request-offer of either contracts can be split, and the offers and requests, that become available again, can be re-combined with complementary actions of $S$, and vice-versa.
The a-product turns out to satisfactorily model coordination policies in dynamically changing environments, because the
a-product supports a form of \emph{dynamic orchestration}, that adjusts the workflow of messages when new principals join the contract.

We now introduce our first operation of composition; recall that we implicitly assume the alphabet of a contract automaton of rank $m$ to be $A \subseteq ( A^{r} \cup A^{o} \cup \{\blk\})^m$.
Note that the first case of the definition of $T$ below is for the matching of actions of two component automata, while the other considers the action of a single component.

\begin{defi}[Product]
\label{def:prod}
Let  $\mathcal{A}_i=\langle Q_i,\vec{q_0}_i,A^{r}_i, A^{o}_i, T_i, F_i \rangle, i \in 1 \ldots n$ be contract automata of rank $r_i$. The \emph{product} $\bigotimes_{i \in 1 \ldots n} \mathcal{A}_i$ is the contract automaton $\langle Q, \vec{q_0}, A^{r}, A^{o}, T, F \rangle$ of rank $m= \sum_{i \in 1 \ldots n} r_i$, where:
\begin{itemize}
\item $Q = Q_1 \times ... \times Q_n, \quad \text{where } \vec{q_0}= \vec{q_0}_1  \ldots  \vec{q_0}_n$
\item $A^{r} =\bigcup_{i \in 1 \cdots n} A^{r}_i, \quad A^{o}=\bigcup_{i \in 1 \cdots n} A^{o}_i$ 
\item $F=\{\vec{q}_1 \ldots \vec{q}_n\mid \vec{q}_1 \ldots \vec{q}_n \in Q, \vec{q}_i \in F_i, i \in 1 \ldots n \}$
 \item $T$ is the least subset of $Q \times A \times
     Q$ s.t. $(\vec{q},\vec c, \vec{q}') \in T$ iff,
    when $\vec q = \vec q_1 \ldots \vec q_n \in Q$,
    \begin{itemize}
    \item either there are $1 \leq i < j \leq n$ s.t.
      $(\vec{q}_i,\vec a_i,\vec{q}'_i) \in T_i$, $(\vec q_j, \vec
      a_j,\vec{q}'_j) \in T_j$, $\vec a_i \matchrel \vec a_j$ and \\
      $\left\{\begin{array}{l}
          \vec c = \blk^u \vec a_i \blk^v \vec a_j \blk^z $ with $
          u = r_1 + \ldots + r_{i-1}, \ v = r_{i+1} + \ldots + r_{j-1},
           |\vec{c}|=m
          \\ and \\
          \vec{q}' = \vec q_1 \ldots \vec q_{i-1} \,\ \vec{q}'_i\,\ \vec q_{i+1}
          \ldots \ \vec q_{j-1} \,\ \vec{q}'_j \,\ \vec q_{j+1} \ldots \vec q_n
        \end{array}\right.
      $
    \item or  there is  $1 \leq i \leq n$  s.t. 
	$(\vec{q}_i,\vec a_i,\vec{q}'_i) \in T_i$ and \\ $\vec c = \blk^u \vec a_i\blk^v$ with $u = r_1 + \ldots +
      r_{i-1}$, $v = r_{i+1} + \ldots + r_n$, and \\
       $\vec{q}' = \vec q_1
      \ldots \vec q_{i-1}\,\; \vec{q}'_i\,\ \vec q_{i+1} \ldots \vec q_n$ and  \\ 
	$\forall  j \neq i, 1 \leq j \leq n, (\vec q_j, \vec a_j,\vec{q}'_j) \in T_j$ it does
      not hold that $\vec a_i \matchrel \vec a_j$. 
    \end{itemize} 
\end{itemize}
\end{defi}\medskip

\noindent There is a simple way of retrieving the principals involved in a composition of contract automata obtained through the product introduced above: just introduce projections $ \prod^i$ as done below.
For example, for the contract automata in Figure~\ref{fig:automa_rank},  we have 
$\mathcal{A}_1 = \prod^1 (\mathcal{A}_3)$ and $\mathcal{A}_2 = \prod^2 (\mathcal{A}_3)$.

\begin{defi}[Projection]
\label{def:proj}
Let $ \mathcal{A} =\langle Q, \vec{q_0}, A^{r}, A^{o}, T, F \rangle$ be a contract automaton of rank $n$, then the
\emph{projection} on the i-th principal is 
 $\prod^i( \mathcal{A})=\langle \prod^i(Q), \ithel{\vec{q_0}} i, \prod^i(A^{r}), \prod^i(A^{o}), \prod^i(T), \prod^i(F) \rangle$ where $i \in 1 \ldots n$ and:
\bigskip
\\
$ 
\prod^i(Q) = \{\ithel{\vec{q}} i \mid \vec{q} \in Q\} 
\quad
\prod^i(F)=\{\ithel{\vec{q}} i \mid \vec{q} \in F\} 
\quad
\prod^i(A^{r})=\{ a \mid a \in A^{r}, (q,a,q') \in \prod^i(T)\}
\medskip
\\
\prod^i(T)=\{(\ithel{\vec{q}} i, \ithel{\vec{a}} i,\ithel{\vec{q'}} i) \mid (\vec{q},\vec{a},\vec{q'}) \in T \wedge 
\ithel{\vec{a}} i \neq \blk\}  
\quad
\prod^i(A^{o})=\{ \overline{a} \mid \overline{a} \in A^{o}, (q,\overline{a},q') \in \prod^i(T)\}
$

\end{defi}
%

The following proposition states that decomposition is the inverse of product, and its proof is immediate.

\begin{prop}[Product Decomposition]\label{pro:product_decomposition}
Let $\mathcal{A}_1, \ldots, \mathcal{A}_n$ be a set of principal contract automata, then $\prod^i(\bigotimes_{j \in 1 \ldots n} \mathcal{A}_j) = \mathcal{A}_i$.
\end{prop}

Our second operation of composition first extracts from its operands the principals they are composed of, and then reassembles them.

\begin{defi}[a-Product]\label{def:aproduct}
Let  $ \mathcal{A}_1,\mathcal{A}_2$ be two contract automata of rank $n$ and $m$, respectively, and let $I=\{\prod^i(\mathcal{A}_1) \mid 0 < i \leq n\} \cup \{ \prod^j(\mathcal{A}_2) \mid 0 < j \leq m \}$.  
Then the \emph{a-product} of $\mathcal{A}_1$ and $\mathcal{A}_2$ is $ \mathcal{A}_1 \boxtimes \mathcal{A}_{2} = \bigotimes_{\mathcal{A}_i \in I} \mathcal{A}_i$.

\end{defi}
Note that if $\mathcal{A}, \mathcal{A}'$ are principal contract automata, then $\mathcal{A} \otimes \mathcal{A}' =  \mathcal{A} \boxtimes \mathcal{A}'$. From now onwards we assume that every contract automaton $\mathcal{A}$ of rank $r_\mathcal{A} > 1$ is composed by principal contract automata using the operations of product and a-product.
E.g.\ in Figure~\ref{fig:automa_rank}, we have that 
$\mathcal{A}_3 = \mathcal{A}_1 \otimes \mathcal{A}_2 = \mathcal{A}_1 \boxtimes \mathcal{A}_{2}$. Finally, both compositions are commutative, up to the expected rearrangement of the vectors of actions, and $\boxtimes$ is also associative, while $\otimes$ is not, as shown by the following example.

\begin{exa}\label{ex:non-assoc}
In Figure~\ref{fig:alicebob} Mary (the automaton in the central position) offers a toy that both Bill (at left) and John (at right) request. 
In the product $(Bill \otimes Mary) \otimes John$ the toy is assigned to Bill who first enters into the composition with Mary, no matter if John performs the same move. 
Instead, in the product $Bill \otimes (Mary \otimes John)$ the toy is assigned to John. 
In the last row we have the a-product of the three automata that represents a dynamic re-orchestration: no matter of who is first composed with Mary, the toy will be non-deterministically assigned to either principal.
\end{exa}

%
%
%
\begin{figure}[tb]
\center
\begin{tikzpicture}[->,>=stealth',shorten >=1pt,auto,node distance=2.0cm,
                    semithick, every node/.style={scale=0.5}]
  \tikzstyle{every state}=[fill=white,draw=black,text=black]

  \node[initial,state] (A)                    {${q_0}_1$};
  \node[state,accepting]         (B) [right of=A] {${q_1}_1$};

  \path (A)			edge             node{${toy}$} (B);

\end{tikzpicture}
\begin{tikzpicture}[->,>=stealth',shorten >=1pt,auto,node distance=2.0cm,
                    semithick, every node/.style={scale=0.5}]
  \tikzstyle{every state}=[fill=white,draw=black,text=black]

  \node[initial,state] (A)                    {${q_0}_2$};
  \node[state,accepting]         (B) [right of=A] {${q_1}_2$};

  \path (A)			edge             node{$\overline{toy}$} (B);
\end{tikzpicture}
\begin{tikzpicture}[->,>=stealth',shorten >=1pt,auto,node distance=2.0cm,
                    semithick, every node/.style={scale=0.5}]
  \tikzstyle{every state}=[fill=white,draw=black,text=black]

  \node[initial,state] (A)                    {${q_0}_3$};
  \node[state,accepting]         (B) [right of=A] {${q_1}_3$};

  \path (A)			edge             node{$toy$} (B);
\end{tikzpicture} \\ \vspace{5pt}
\begin{tikzpicture}[->,>=stealth',shorten >=1pt,auto,node distance=4.0cm,
                    semithick, every node/.style={scale=0.4}]
  \tikzstyle{every state}=[fill=white,draw=black,text=black]

  \node[initial,state] (A)                               {${q_0}_1,{q_0}_2,{q_0}_3$};
	\node[state] (B)[right of=A]                    {${q_1}_1,{q_1}_2,{q_0}_3$};
	\node[state] (C)[below of=A]                   {${q_0}_1,{q_0}_2,{q_1}_3$};
  \node[state,accepting] (D) [right of=C]        {${q_1}_1,{q_1}_2,{q_1}_3$};

  \path (A)			edge             node{$({toy},\overline{toy},\blk)$} (B)
		   	   			edge             node{$(\blk,\blk,toy)$} (C)
	      (B)             edge             node{$(\blk,\blk,toy)$}(D)
           (C)             edge             node{$(toy,\overline{toy},\blk)$}(D);
\end{tikzpicture} 
\begin{tikzpicture}[->,>=stealth',shorten >=1pt,auto,node distance=4.0cm,
                    semithick, every node/.style={scale=0.4}]
  \tikzstyle{every state}=[fill=white,draw=black,text=black]

  \node[initial,state] (A)                               {${q_0}_1,{q_0}_2,{q_0}_3$};
	\node[state] (B)[right of=A]                    {${q_0}_1,{q_1}_2,{q_1}_3$};
	\node[state] (C)[below of=A]                   {${q_1}_1,{q_0}_2,{q_0}_3$};
  \node[state,accepting] (D) [right of=C]        {${q_1}_1,{q_1}_2,{q_1}_3$};

  \path (A)			edge             node{$(\blk,\overline{toy},toy)$} (B)
		   	   			edge             node{$(toy,\blk,\blk)$} (C)
	      (B)             edge             node{$(toy,\blk,\blk)$}(D)
           (C)             edge             node{$(\blk, \overline{toy},toy)$}(D);
\end{tikzpicture} 
\begin{tikzpicture}[->,>=stealth',shorten >=1pt,auto,node distance=4.0cm,
                    semithick, every node/.style={scale=0.4}]
  \tikzstyle{every state}=[fill=white,draw=black,text=black]

  \node[initial,state] (A)                               {${q_0}_1,{q_0}_2,{q_0}_3$};
	\node[state] (B)[right of=A]                    {${q_1}_1,{q_1}_2,{q_0}_3$};
	\node[state] (C)[below of=A]                   {${q_0}_1,{q_1}_2,{q_1}_3$};
  \node[state,accepting] (D) [right of=C]        {${q_1}_1,{q_1}_2,{q_1}_3$};

  \path (A)			edge             node{$({toy},\overline{toy},\blk)$} (B)
		   	   			edge             node{$(\blk, \overline{toy},toy)$} (C)
	      (B)             edge             node{$(\blk,\blk,toy)$}(D)
           (C)             edge             node{$(toy, \blk,\blk)$}(D);
\end{tikzpicture}
\caption{From left to right and top-down: the principal contract automata of Bill, Mary and John, the contract automata $(Bill \otimes Mary) \otimes John$, $Bill \otimes (Mary \otimes John)$ and $Bill \boxtimes Mary \boxtimes John$.}
\label{fig:alicebob}
\end{figure}
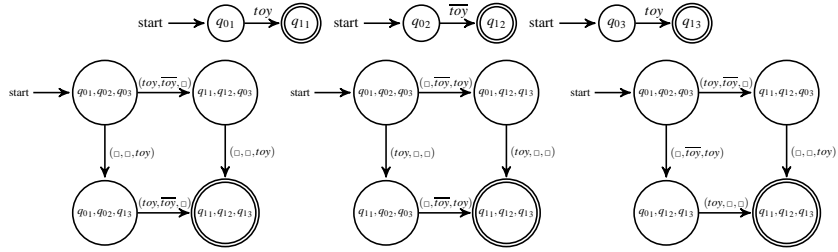
\begin{restatable}[]{prop}{proassociative}
\label{pro:associative}
The following properties hold:

-- $\exists \mathcal{A}_1, \mathcal{A}_2, \mathcal{A}_3.(\mathcal{A}_1 \otimes \mathcal{A}_2) \otimes  \mathcal{A}_3 \neq \mathcal{A}_1 \otimes (\mathcal{A}_2 \otimes  \mathcal{A}_3)$

-- $\forall \mathcal{A}_1, \mathcal{A}_2, \mathcal{A}_3. (\mathcal{A}_1 \boxtimes \mathcal{A}_2) \boxtimes \mathcal{A}_3 = \mathcal{A}_1 \boxtimes (\mathcal{A}_2 \boxtimes  \mathcal{A}_3)
$
\end{restatable}

\section{Enforcing Agreement}
\label{sect:agreement}
It is common to say that some contracts are in agreement when all the requests they make  have been fulfilled by corresponding offers~\cite{Cast2009ACM,Cast2009CONCUR,Bordeaux2005, 
LaneveP2015,
Barbanera2015,
BDLL15,
deNicolaHennessy,
Bernardi2014,
Bartoletti2015, 
BartSci2013}.
In terms of contract automata, this is rendered in two different ways, the first of which is introduced below and resembles the notion of compliance introduced in~\cite{Cast2009ACM,Cast2009CONCUR}.
We say that two or more contract automata are in \emph{agreement} when the final states of their product are reachable from the initial state by traces only made of matches and offer actions. 
Our goal is to enforce the behaviour of principals so that they only follow the traces of the automaton which lead to agreement. 
Additionally, it is easy to track every action performed by each principal, because we use vectors of actions as the elements of the alphabet of contract automata.
It is equally easy finding who is liable in a bad interaction, i.e.\ the principals who perform a transition 
leaving a state from which agreement is possible, reaching a state where instead agreement is no longer possible.

We now introduce the notion of \emph{agreement} as a property of the language recognised by a contract automaton.

\begin{defi}[Agreement]
A trace accepted by a contract automaton is in \emph{agreement} if it belongs to the set 
\[\mathfrak{A}=\{w  \in ( {\Sigma^{n}})^* \mid Obs(w) \in (\Oset \cup \{\tau\})^*,n>1\}\]
\end{defi}

\noindent
Note that, if an action observable in $w$ is a request, i.e.\ it belongs to $\Rset$, then $w$ is not in agreement. 
 Intuitively, a trace is in agreement if it only contains offer and match actions, i.e.\ if no requests are left unsatisfied.

\begin{exa}
The automaton $\mathcal{A}_3$ in Figure~\ref{fig:automa_rank} has a trace in agreement:
$Obs((\overline{res},\blk)(sig,\overline{sig})) = \overline{res}\,\tau \in \mathfrak{A}$, and one not in agreement:  $Obs( (sig,\overline{sig}) (\blk, res)) = \tau\, res \not \in \mathfrak{A}$.
\end{exa}

A contract automaton is safe when all the traces of its language are in agreement, and admits agreement when at least one 
of its traces is in agreement.
Formally:

\begin{defi}[Safety]
A contract automaton $\mathcal{A}$  is \emph{safe} if $\mathscr{L}(\mathcal{A})\subseteq \mathfrak{A}$, otherwise it is \emph{unsafe}. 

Additionally, if $\mathscr{L}(\mathcal{A}) \cap \mathfrak{A} \neq \emptyset$ then $\mathcal{A}$ \emph{admits agreement}.
\end{defi}

\begin{exa}
The contract automaton $\mathcal{A}_3$ of Figure~\ref{fig:automa_rank} is unsafe, but it admits agreement since $\mathscr{L}(\mathcal{A}_3) \cap \mathfrak{A} =   (\overline{res},\blk)^*(sig,\overline{sig})$. 
Consider now the contract automata $Bill$ and $Mary$ in Figure~\ref{fig:alicebob};
their product $Bill \otimes Mary$ is safe because $\mathscr{L}(Bill \otimes Mary)= (toy,\overline{toy}) \subset \mathfrak{A}$.
\end{exa}
Note that the set $\mathfrak{A}$ can be seen as a safety property in the default-accept approach~\cite{schneider2000enforceable}, where the set of bad prefixes of $\mathfrak A$ contains those traces ending with a trailing request,  i.e. $\{w\vec{a} \mid w \in \mathfrak{A}, Obs(\vec{a}) \in \Rset  \}$.
One could then consider a definition of product that disallows the occurrence of transitions labelled by requests only.
However, this choice would not prevent a product of contracts to reach a deadlock.
In addition, compositionality would have been compromised, as shown in the following example.

\begin{exa}\label{ex:selling}
In what follows, we feel free to present contract automata through a sort of extended regular expressions. 
Consider a simple selling scenario involving two parties $Ann$  and $Bart$. 

Bart starts by notifying Ann that he is ready to start the negotiation, and waits from Ann to select a pen or a book.
In case  Ann selects the pen, he may decide
to withdraw and restart the negotiation again, or to accept the payment. 
As soon as  Ann selects the book, then Bart cannot
withdraw any longer, and waits for the payment.
The contract of $Bart$ is:
\[
Bart = (\overline{init}.pen.\overline{cancel})^*.
  (\overline{init}.book.pay + \overline{init}.pen.pay)
\]
The contract of Ann is dual to Bart's. 
Ann waits to receive a notification from Bart when ready to negotiate.
Then Ann decides what to buy.
If she chooses the pen, she may proceed with the payment unless a withdrawal from Bart is received.
In this case, Ann can repeatedly try to get the pen, until she succeeds and pays for it, or
buys the book but omits to pay it (violating the contract $Ann \otimes Bart$ resulting from the orchestration, see below).

The contract of $Ann$ is:
\[
Ann =(init.\overline{pen}.cancel)^*.(init.\overline{pen}.\overline{pay} + 
                                                            init.\overline{book})
\]
The contract $\mathcal A = Ann \otimes Bart$ is in  Figure~\ref{fig:controller} top. 
%
%
Assume now to change the product $\otimes$ so to disallow transitions labelled by requests.
The composition of $Ann$ and $Bart$ is in Figure~\ref{fig:controller}, bottom right part, and contains the malformed trace in which $Bart$ does not reach a final state: 
\[
(init,\overline{init})(\overline{book},book)
\]
%
In addition, if a third principal $Carol=\overline{pay}$ were involved, willing to pay for everybody, the following trace in agreement would not be accepted
\[
(init,\overline{init},\blk)(\overline{book},book,\blk)(\blk,pay,\overline{pay})
\]
because Bart's request was discarded by the wrongly amended composition operator.
So, compositionality would be lost.

\end{exa}

To avoid the two unpleasant situations of deadlock and lack of compositionality, we introduce below a technique for driving a safe composition of contracts, in the style of the Supervisory Control for Discrete Event Systems~\cite{Cassandras2006}.

A discrete event system is a finite state automaton, where \emph{accepting} states represent the successful termination of a task, while \emph{forbidden} states should never be traversed in ``good'' computations.
Generally, the purpose of supervisory control theory is to synthesise a controller that enforces good computations.
%
To do so, this theory distinguishes between
\emph{controllable} events (those the controller can disable) and \emph{uncontrollable} events (those always enabled), besides partitioning events into \emph{observable} and \emph{unobservable} (obviously uncontrollable).
If all events are observable then a most  permissive controller exists that never blocks a good computation~\cite{Cassandras2006}.

The purpose of contracts is to declare all the activities of a principal in terms of requests and offers. Therefore all the actions 
of a (composed) contract are controllable and observable.
Clearly, the behaviours that we want to enforce upon a given contract automaton $\mathcal{A}$ are exactly the traces in agreement, and so  we assume that a request leads to a forbidden state.
A most permissive controller exists for contract automata and is defined below.

\begin{defi}[Controller]
Let $\mathcal{A}$ and $\mathcal{K}$ be contract automata, we call $\mathcal{K}$ \emph{controller} of $\mathcal{A}$ if and only if
$\mathscr{L}(\mathcal{K}) \subseteq \mathfrak{A} \cap \mathscr{L}(\mathcal{A})$.\\
A controller $\mathcal{K}$ of $\mathcal{A}$ is the \emph{most permissive controller (mpc)} if and only if  for all $\mathcal{K}'$ controller of $\mathcal{A}$ it is $\mathscr{L}(\mathcal{K}') \subseteq \mathscr{L}(\mathcal{K})$.
\end{defi}

Since the most permissive controller eliminates the traces not in agreement, the following holds.

\begin{restatable}[]{prop}{propmpca}\label{prop:mpc_A}
Let $\mathcal{K}$ be the mpc of the contract automaton $\mathcal{A}$, then $\mathscr{L}(\mathcal{K})=\mathfrak{A} \cap \mathscr{L}(\mathcal{A})$.
 \end{restatable}

In order to effectively build the most permissive controller, we introduce below the notion of hanged state, i.e.\ a state from which no final state can be reached.

\begin{defi}[Hanged state]
Let  $\mathcal{A}=\langle Q, \vec{q_0}, A^{r}, A^{o}, T, F \rangle$ be a contract automaton, then $\vec{q} \in Q$ is \emph{hanged}, and belongs to the set $Hanged(\mathcal{A})$, if for all
$\vec q_f \in F, \nexists w.(w,\vec{q}) \rightarrow^*(\epsilon,\vec{q_f})$.
\end{defi}

\begin{defi}[Mpc construction]\label{def:mpc}
Let  $\mathcal{A}=\langle Q, \vec{q_0}, A^{r}, A^{o}, T, F \rangle$ be a contract automaton,\linebreak $\mathcal{K}_1=\langle Q, \vec{q_0}, A^{r}, A^{o},  T \setminus (\{t \in T \mid t$ is a request transition $\}, F \rangle$ and define
\[
\mathcal{K_A}=\langle Q \setminus Hanged(\mathcal{K}_1), \vec{q_0}, A^{r}, A^{o}, T_{\mathcal{K}_1} \setminus \{(\vec{q},a,\vec{q'}) \mid \{\vec{q},\vec{q'}\}\cap Hanged(\mathcal{K}_1) \neq \emptyset \}, F \rangle
\]

\end{defi}
\begin{restatable}[Mpc]{prop}{procontroller}
\label{pro:controller}
The controller $\mathcal K_A$ of Definition~\ref{def:mpc} is the most permissive controller of the contract automaton $\mathcal A$.
\end{restatable}
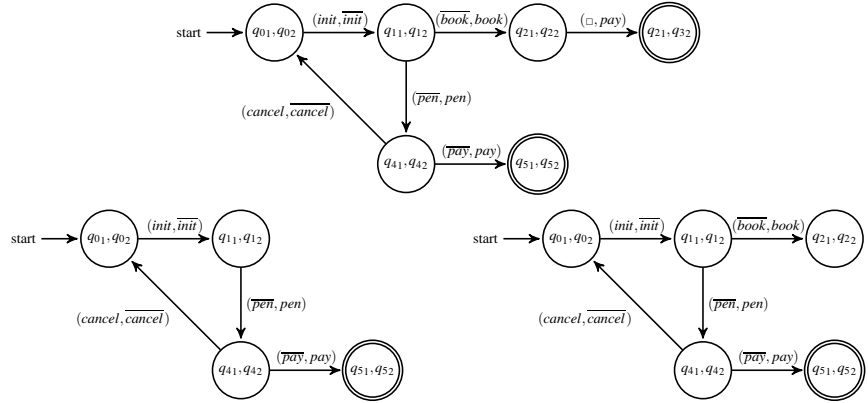
\begin{figure}[tb]
\center
\begin{tikzpicture}[->,>=stealth',shorten >=1pt,auto,node distance=3.5cm,
                    semithick, every node/.style={scale=0.5}]
  \tikzstyle{every state}=[fill=white,draw=black,text=black]

  \node[initial,state] (A)                    {${q_0}_1,{q_0}_2$};
  \node[state]         (B) [right of=A] {${q_1}_1,{q_1}_2$};
  \node[state]	     (C) [right of=B] {${q_2}_1,{q_2}_2$};
  \node[state]		(D)  [below of=B]{${q_4}_1,{q_4}_2$};
  \node[state,accepting]		(F)   [right of=D]{${q_5}_1,{q_5}_2$}; 
  \node[state,accepting]		(E)	  [right  of=C]{${q_2}_1,{q_3}_2$};

  \path (A)			edge             node{$(init,\overline{init})$} (B)
		(B) 			edge             node {$(\overline{book},book)$} (C)
						edge			node {$(\overline{pen},pen)$}(D)
		(C)				edge			node {$(\blk,pay)$}(E)
		(D)				edge			node{$(cancel,\overline{cancel})$}(A)
						edge			node{$(\overline{pay},pay)$}(F);
\end{tikzpicture}

\medskip

\begin{tikzpicture}[->,>=stealth',shorten >=1pt,auto,node distance=3.5cm,
                    semithick, every node/.style={scale=0.5}]
  \tikzstyle{every state}=[fill=white,draw=black,text=black]

  \node[initial,state] (A)                    {${q_0}_1,{q_0}_2$};
  \node[state]         (B) [right of=A] {${q_1}_1,{q_1}_2$};
  \node[state]		(D)  [below of=B]{${q_4}_1,{q_4}_2$};
  \node[state,accepting]		(F)   [right of=D]{${q_5}_1,{q_5}_2$};		

  \path (A)			edge             node{$(init,\overline{init})$} (B)
		(B) 			edge             node {$(\overline{pen},pen)$} (D)
		(D)				edge			node{$(cancel,\overline{cancel})$}(A)
						edge			node{$(\overline{pay},pay)$}(F);
\end{tikzpicture}
\qquad
\begin{tikzpicture}[->,>=stealth',shorten >=1pt,auto,node distance=3.5cm,
                    semithick, every node/.style={scale=0.5}]
  \tikzstyle{every state}=[fill=white,draw=black,text=black]

  \node[initial,state] (A)                    {${q_0}_1,{q_0}_2$};
  \node[state]         (B) [right of=A] {${q_1}_1,{q_1}_2$};
  \node[state]	     (C) [right of=B] {${q_2}_1,{q_2}_2$};
  \node[state]		(D)  [below of=B]{${q_4}_1,{q_4}_2$};
  \node[state,accepting]		(F)   [right of=D]{${q_5}_1,{q_5}_2$};

  \path (A)			edge             node{$(init,\overline{init})$} (B)
		(B) 			edge             node {$(\overline{book},book)$} (C)
						edge			node {$(\overline{pen},pen)$}(D)
		(D)				edge			node{$(cancel,\overline{cancel})$}(A)
						edge			node{$(\overline{pay},pay)$}(F);
\end{tikzpicture}
%
%
\caption{The contract automata of Example~\ref{ex:selling}: top the contract automaton $\mathcal{A}$; bottom left its most permissive controller $\mathcal K_\mathcal{A}$, bottom right an automaton obtained with an inaccurate filtering composition.}
\label{fig:controller}
\end{figure}

\begin{exa}

Consider again Example~\ref{ex:selling}.  For obtaining the most permissive controller
 we first compute the auxiliary set $\mathcal{K}_1$ that does not contain the transition 
$(({q_2}_1,{q_2}_2),(\blk,pay), ({q_2}_1,{q_3}_2))$ because it represents a request 
from Bart which is not fulfilled by Ann.
As a consequence, some states are hanged:
\[
Hanged(\mathcal{K}_1) =\{({q_2}_1,{q_2}_2)\}
\]
By removing them, we eventually obtain  $\mathcal K_\mathcal{A}$, the most permissive controller 
of $\mathcal{A}$ depicted in Figure~\ref{fig:controller}, bottom left part.
%
\end{exa}



The following proposition rephrases the notions of safe, unsafe and admits agreement on automata
in terms of their most permissive controllers.
\begin{restatable}[]{prop}{propksafety}\label{prop:Ksafety}
Let $\mathcal{A}$ be a contract automaton and let $\mathcal K_{\mathcal{A}}$ be its mpc, the following hold:
\begin{itemize}
\item if $\mathscr{L}(\mathcal K_{\mathcal{A}})=\mathscr{L}(\mathcal{A})$ then $\mathcal{A}$ is safe, otherwise if $\mathscr{L}(\mathcal K_{\mathcal{A}})\subset \mathscr{L}(\mathcal{A})$ then $\mathcal{A}$ is unsafe; 
\item  if $\mathscr{L}(\mathcal K_{\mathcal{A}})\neq\emptyset$, then $\mathcal{A}$ admits agreement.
\end{itemize}
\end{restatable}

We introduce now an original notion of \emph{liability}, that characterises those principals potentially responsible of the divergence from the behaviour in agreement. 
The liable principals are those who perform the first transition in a run, that is not possible in the most permissive controller. 
As noticed above, after this step is done, a successful state cannot be reached any longer, and so the principals who performed it will be blamed.
Note in passing that hanged states play a crucial role here: just removing the request transitions from $\mathcal{A}$ would result in a contract automaton language equivalent to the mpc, but detecting liable principals would be much more intricate.

\begin{defi}[Liability]
\label{def:culpability}
Let ${\mathcal A}$ be a contract automaton and $\mathcal K_\mathcal{A}$ be its mpc of Definition~\ref{def:mpc}; 
let \hbox{$(v\vec{a}w, \vec{q_0}) \rightarrow^*(\vec{a}w, \vec{q})$} be a run of both automata and let $\vec q$ be such that $(\vec{a}w, \vec{q}) \rightarrow (w,\vec{q'})$ is possible in $\mathcal A$ but not in $\mathcal K_{\mathcal{A}}$.
The principals  $\Pi^i(\mathcal{A})$ such that $\ithel{\vec{a}} i  \neq \blk, i \in 1 \ldots r_{\mathcal{A}}$ are \emph{liable for $\vec{a}$} and belong to $Liable(\mathcal{A},v\vec{a}w)$.
Then, the set of \emph{liable} principals in $\mathcal{A}$ is 
$Liable(\mathcal{A})=\{i \mid \exists w.i \in Liable(\mathcal{A},w)\}$.
\end{defi}


\begin{exa}
In Figure \ref{fig:controller}, bottom left, we have $Liable(\mathcal{A}) = \{1,2\}$, hence both Ann and Bart are possibly liable, because the match transition with label $(\overline{book},book)$ can be performed, that leads to a violation of the agreement.
\end{exa}

The following proposition is immediate.

\begin{restatable}{prop}{prosafeliable}\label{pro:safe_liable}
A contract automaton $\mathcal{A}$ is safe if and only if $Liable(\mathcal{A})=\emptyset$.
\end{restatable}

%

Note that the set $Liable(\mathcal{A})$ can be rewritten as follows
\[
\{i \mid  (\vec{q},\vec{a},\vec{q'}) \in T_{\mathcal{A}}, \ithel{\vec{a}} i  \neq \blk, \vec{q} \in Q_{\mathcal K_{\mathcal{A}}}, \vec{q'} \not \in  Q_{\mathcal K_{\mathcal{A}}}\}
\]
so making its calculation straightforward, as well as checking the safety of $\mathcal{A}$.


Some properties of $\otimes$ and $\boxtimes$ follow, that enable us to predict under which conditions a composition is safe without actually computing it.

We first introduce the notions of collaborative and competitive contracts. 
Intuitively, two contracts are \emph{collaborative} if some requests of one meet the offers of the other, and are \emph{competitive} if both can satisfy the same request.
An example follows.

\begin{exa}\label{ex:liable}
Consider the contract automata $Bill, Mary, John$  in Figure~\ref{fig:alicebob}. 
In Figure~\ref{fig:bobeve} the contract automaton $Bill \otimes John$ is displayed.
The two contract automata $Mary$ and $Bill \otimes John$ 
are collaborative and not competitive, indeed the offer $\overline {toy}$ of $Mary$ is matched
in $Bill \otimes John$, and no other principals interfere with this offer.
Moreover, let $\mathcal{A}_1= \overline{apple}+cake \otimes apple + \overline{cake}$ and $\mathcal{A}_2=\overline{apple}$. 
The pair $\mathcal{A}_1,\mathcal{A}_2$ is competitive since $\mathcal{A}_2$ interferes with $\mathcal{A}_1$ on the $\overline{apple}$ offer.
\end{exa}

\begin{defi}[Competitive, Collaborative]
The pair of CA $\mathcal{A}_1=\langle Q_1, \vec {q_0}_1, A^{r}_1, A^{o}_1, T_1, F_1 \rangle$ and 
$\mathcal{A}_2=\langle Q_2,\vec  {q_0}_2, A^{r}_2, A^{o}_2, T_2, F_2 \rangle$ are
\begin{itemize} 
\item 
\emph{competitive} if $A^o_1 \cap A^{o}_2 \cap co(A^r_1 \cup A^r_2) \neq \emptyset$

\item
\emph{collaborative} if $(A^o_1 \cap co(A^r_2)) \cup (co(A^r_1) \cap A^o_2) \neq \emptyset$.
\end{itemize}
\end{defi}

Note that \emph{competitive} and \emph{collaborative} are not mutually exclusive, as stated in the first and second item of Theorem~\ref{the:composition} below. 
Moreover if two contract automata are \emph{non-competitive} then all their match actions are preserved in their composition, indeed we have $\mathcal{A}_1 \boxtimes \mathcal{A}_2 = \mathcal{A}_1 \otimes \mathcal{A}_2$.

\begin{figure}[tb]
\center
\begin{tikzpicture}[->,>=stealth',shorten >=1pt,auto,node distance=4.0cm,
                    semithick, every node/.style={scale=0.4}]
  \tikzstyle{every state}=[fill=white,draw=black,text=black]

  \node[initial,state] (A)                               {${q_0}_2,{q_0}_3$};
	\node[state] (B)[right of=A]                    {${q_1}_2,{q_0}_3$};
	\node[state] (C)[below of=A]                   {${q_0}_2,{q_1}_3$};
  \node[state,accepting] (D) [right of=C]        {${q_1}_2,{q_1}_3$};

  \path (A)			edge             node{$(toy,\blk)$} (B)
		   	   			edge             node{$(\blk,toy)$} (C)
	      (B)             edge             node{$(\blk,toy)$}(D)
           (C)             edge             node{$(toy,\blk)$}(D);
\end{tikzpicture}
\caption{The contract automaton $Bill \otimes John$ of Example~\ref{ex:liable}}.
\label{fig:bobeve}
\end{figure}
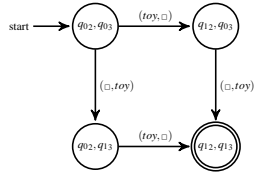
The next theorem says that the composition of safe and non-competitive contracts prevents all principals from harmful interactions, unlike the case of safe competitive contracts.
In other words, when $\mathcal{A}_1$ and $\mathcal{A}_2$ are safe, no principals will be found liable in $\mathcal{A}_1 \otimes \mathcal{A}_2$ 
(i.e.\ $Liable(\mathcal{A}_1 \otimes \mathcal{A}_2) = \emptyset$), and the same happens for 
$\mathcal{A}_1 \boxtimes \mathcal{A}_2$ if the two are also non-competitive 
(i.e.\ $Liable(\mathcal{A}_1 \boxtimes \mathcal{A}_2) = \emptyset$).

\begin{restatable}[]{thm}{thecomposition}
\label{the:composition} 
If two contract automata $\mathcal{A}_1$ and $\mathcal{A}_2$ are
\begin{enumerate}
\item  \label{lem:composition} competitive then they are collaborative,
\item \label{lem:collcomp} collaborative and safe, then they are competitive,
\item  \label{lem:safe} safe then $\mathcal{A}_1 \otimes \mathcal{A}_2\text{ is safe}, \mathcal{A}_1 \boxtimes \mathcal{A}_2\text{ admits agreement}$,
\item  \label{lem:unsafe} non-collaborative, and one or both unsafe, then $\mathcal{A}_1 \otimes \mathcal{A}_2,\mathcal{A}_1 \boxtimes  \mathcal{A}_2$ are unsafe,
\item \label{lem:noncompetitive_safe} safe and non-competitive, then $\mathcal{A}_1 \boxtimes \mathcal{A}_2$ is safe.
\end{enumerate}

\end{restatable}

Note that in item  \ref{lem:safe} of Theorem~\ref{the:composition} it can be that $\mathcal A_1 \boxtimes \mathcal A_2$  is not
\emph{safe}. Moreover consider the contract automata $\mathcal{A}_1$  and $\mathcal{A}_2 $ of Example~\ref{ex:liable}.
We have that  $\mathcal{A}_1\boxtimes\mathcal{A}_2 $ is unsafe because the trace $(\blk, apple, \overline{apple})(cake,\blk,\blk)$ belongs to $\mathscr{L}(\mathcal{A}_1\boxtimes\mathcal{A}_2)$.

%

\section{Weak Agreement}
\label{sect:weak-agreement}
%
As said in the introduction, we will now consider a more liberal notion of agreement, where an offer can be asynchronously fulfilled by a matching request, even though either of them occur beforehand.
In other words, some actions can be taken on credit, assuming that in the future the obligations will be honoured. 
According to this notion, called here \emph{weak agreement}, computations well behave when all the requests are matched by offers, in spite of lack of synchronous agreement, in the sense of Section~\ref{sect:agreement}.
This may lead to a circularity, as shown by the example below, because, e.g.\ one principal first requires something from the other and then is willing to fulfil the request of the other principal, who in turn behaves in the same way.
This is a common scenario in contract composition, and variants of weak agreement have been studied using many different formal techniques, among which Process Algebras, Petri Nets, non-classical Logics, Event Structures~\cite{BartolettiZ09,BartPOST2013,BartFSEN2013,BartSci2013}.

\begin{exa}
Suppose Alice and Bob want to share a bike and an airplane, but neither trusts the other.
Before providing their offers they first ask for the complementary requests. 
As regular expressions:  $Alice =bike.\overline{airplane}$ and 
$Bob = airplane.\overline{bike}$. 
The language of their composition  is:
\[\mathscr{L}(Alice \otimes Bob)=\{ (\blk, airplane)(bike, \overline{bike})(\overline{airplane}, \blk),
(bike, \blk)(\overline{airplane}, airplane)(\blk, \overline{bike}) \}\,.\] 
In both possible traces the contracts fail in exchanging the bike or the airplane synchronously, hence $\mathscr{L}(Alice \otimes Bob) \cap \mathfrak{A} = \emptyset$ and the composition does not admit agreement. 
\label{ex:biketoy}
\end{exa}
The circularity in the requests/offers is solved by weakening the notion of agreement, allowing a request to be performed on credit and making sure that in the future a complementary offer will occur, giving rise to a trace in weak agreement.
We now formally define weak agreement.

\begin{defi}[Weak Agreement]
A trace accepted by a contract automaton of rank $n > 1$ is in \emph{weak agreement} if it belongs to
$\mathfrak{W}= \{w \in  ({\Sigma^{n}})^* \mid w=\vec a_1 \ldots \vec a_m, \exists \text{ a function }$ 
\mbox{$f: [1..m] \rightarrow [1..m]$}  total and injective on the (indexes of the) request actions of $w$, and such that $f(i)=j$ only if $\vec a_i \matchrel \vec a_j\}$.
\end{defi}

Needless to say, a trace in agreement is also in weak agreement, so $\mathfrak{A}$ is a proper subset of $\mathfrak{W}$, as shown below.

\begin{exa}
Consider $\mathcal{A}_3$ in Figure \ref{fig:automa_rank}, whose trace $(\overline{res},\blk)(sig,\overline{sig})(\blk,res)$ is in $\mathfrak{W}$ but not in $\mathfrak{A}$ 
(all $f$ such that $f(3)=1$ certify the membership) , while $(\overline{res},\blk)(sig,\overline{sig})(\blk,res)(\blk,res) \not \in \mathfrak{W}$. 
\end{exa}


\begin{defi}[Weak Safety]
Let $\mathcal{A}$  be a contract automaton. Then
\begin{itemize}
\item   if $\mathscr{L}(\mathcal{A})\subseteq \mathfrak{W}$ then $\mathcal A$ is \emph{weakly safe}, otherwise is  \emph{weakly unsafe};
\item if $\mathscr{L}(\mathcal{A}) \cap \mathfrak{W} \neq \emptyset$ then $\mathcal{A}$ \emph{admits weak agreement}.
\end{itemize}
\end{defi}

\begin{exa}
In Example~\ref{ex:biketoy} we have $\mathscr{L}(Alice \otimes Bob) \subset \mathfrak{W}$, hence the composition of $Alice$ and $Bob$ is weakly safe. Indeed every $f$ such that $f(1)=3$ certifies the membership for both traces.
\end{exa}

The following theorem states the conditions under which weak agreement is preserved by our operations of contract composition.

\begin{restatable}{thm}{lemweaksafe}
\label{lem:weak-safe}\label{lem:weak-unsafe}\label{lem:noncompetitive_weak-safe}
Let $\mathcal{A}_1, \mathcal{A}_2$ be  two contract automata, then if $\mathcal{A}_1, \mathcal{A}_2$ are
\begin{enumerate}
\item  weakly safe then $\mathcal{A}_1 \otimes \mathcal{A}_2$ is weakly safe, $\mathcal{A}_1 \boxtimes \mathcal{A}_2$ admits weak agreement
\item   non-collaborative and one or both unsafe, then $\mathcal{A}_1 \otimes \mathcal{A}_2,\mathcal{A}_1 \boxtimes \mathcal{A}_2$ are weakly unsafe 
\item  safe and non-competitive, then  $\mathcal{A}_1 \boxtimes \mathcal{A}_2$ is weakly safe.
\end{enumerate}
\end{restatable}

The example below shows that weak agreement is not a context-free notion, in language theoretical sense; rather we will prove it context-sensitive.
Therefore, we cannot define a most permissive controller for weak agreement in terms of contract automata, because they are finite state automata.

\begin{exa}
\label{exa:wagreement}
Let $\mathcal{A}_4$, $\mathcal{A}_5$ and $\mathcal{A}_4\otimes\mathcal{A}_5$ be the contract automata  in Figure \ref{fig:automa_noncontextfree}, then we have that $L=\mathfrak{W} \cap \mathscr{L}(\mathcal{A}_4\otimes\mathcal{A}_5)\neq \emptyset$ is not context-free. 
Consider the following regular language 
\[L'=\{(\overline{a},\blk)^*(\overline{b},\blk)^*(sig,\overline{sig})(\blk,a)^*(\blk,b)^* \}\]
We have that 
\[
L \cap L' = \{ (\overline a, \blk)^{n_1} (\overline b, \blk)^{m_1} (sig, \overline{sig}) (\blk,a)^{n_2}(\blk,b)^{m_2} 
\mid n_1\geq n_2 \geq 0, m_1 \geq m_2 \geq 0 \}
\]
is not context-free (by pumping lemma), and since $L'$ is regular, $L$ is not context-free.
\end{exa}
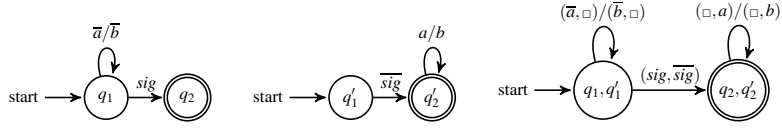
\begin{figure}[tb]
\center
\begin{tikzpicture}[->,>=stealth',shorten >=1pt,auto,node distance=1.8cm,
                    semithick, every node/.style={scale=0.6}]
  \tikzstyle{every state}=[fill=white,draw=black,text=black]

  \node[initial,state] (A)                    {$q_1$};
  \node[state,accepting]         (B) [right of=A] {$q_2$};

  \path (A) edge [loop above]              node {$\overline{a}/\overline{b}$} (A)
			edge             node{$sig$} (B);
\end{tikzpicture} \quad
\begin{tikzpicture}[->,>=stealth',shorten >=1pt,auto,node distance=1.8cm,
                    semithick, every node/.style={scale=0.6}]
  \tikzstyle{every state}=[fill=white,draw=black,text=black]

  \node[initial,state] (A)                    {$q_1'$};
  \node[state,accepting]         (B) [right of=A] {$q_2'$};

  \path (A)			edge             node{$\overline{sig}$} (B)
		(B) edge [loop above]              node {$a/b$} (B);
\end{tikzpicture} \quad
\begin{tikzpicture}[->,>=stealth',shorten >=1pt,auto,node distance=3.0cm,
                    semithick, every node/.style={scale=0.6}]
  \tikzstyle{every state}=[fill=white,draw=black,text=black]

  \node[initial,state] (A)                    {$q_1,q_1'$};
  \node[state,accepting]         (B) [right of=A] {$q_2,q_2'$};

  \path (A)			edge             node{$(sig,\overline{sig})$} (B)
 				     edge [loop above]             node {$(\overline{a},\blk)/(\overline{b},\blk)$} (A)
		(B) edge [loop above]              node {$(\blk,a)/(\blk,b)$} (B);
\end{tikzpicture}
\caption{From left to right the contract automata of Example~\ref{exa:wagreement}: $\mathcal{A}_4,\mathcal{A}_5$, and $\mathcal{A}_4 \otimes \mathcal{A}_5$.}
\label{fig:automa_noncontextfree}
\end{figure}
\begin{restatable}{thm}{thecontextsensitive}
\label{the:contextsensitive}
$\mathfrak{W}$ is a context-sensitive language, but not context-free.
Word decision can be done in $O(n^2)$ time and $O(n)$ space.
\end{restatable}
In general, it is undecidable checking whether a regular language $L$ is included in a
 context-sensitive one, as well as checking emptiness of the intersection of a regular
 language with a context-sensitive one. 
However in our case these two problems are decidable: we will introduce an effective procedure to
check whether a contract automaton $\mathcal{A}$ is weakly safe, or whether it admits weak agreement.
The technique we propose amounts to find optimal solutions to network flow problems~\cite{Hemmecke10}, and will be used  also for detecting  weak liability.

As an additional comment, note that the membership problem is polynomial in time for mildly context-sensitive languages~\cite{Joshi90theconvergence}, but it is PSPACE-complete  for arbitrary ones.
In the first case, checking membership can be done in polynomial time through
 \emph{two way deterministic pushdown automata}~\cite{GrayInf67},
that  have a read-only input tape readable backwards and forwards. 
It turns out that $\mathfrak{W}$ is mildly context-sensitive, and checking whether $w \in \mathfrak{W}$ can be intuitively  done by repeating what follows for all the actions occurring in $w$.
Select an action $\alpha$; scroll the input; and push all the requests on $\alpha$ on the stack;
scroll again the input and pop a request, if any, when a corresponding offer is found.
If at the end the stack is empty the trace $w$ is in $\mathfrak{W}$.

%

Before presenting our decision procedure we fix some useful notation.
Assume as given a contract automaton $\mathcal{A}$, with a single final state $\vec{q}_f\neq \vec{q}_0$.
If this is not the case, one simply adds artificial dummy transitions from all the original final states to the new single final state. 
Clearly, if the modified contract automaton admits weak agreement, also the original one does --- and the two will have the same liable principals.
We assume that all states are reachable from $\vec q_0$ and so is $\vec q_f$ from each of them.
In addition, we enumerate the requests of $\mathcal{A}$, i.e.\  
$A^{r} = \{a^i \mid i \in I_l = \{1,2, \ldots, l\}\}$, as well as its transitions  $T=\{t_1, \ldots, t_n \}$.
Also, let $FS(\vec{q})= \{(\vec{q},\vec{a},\vec{q}')\mid (\vec{q},\vec{a},\vec{q}')\in T\}$
be the \emph{forward star} of a state $\vec{q}$, and let $\ BS(\vec{q})= \{(\vec{q}',\vec{a},\vec{q}) \mid (\vec{q}',\vec{a},\vec{q}) \in T\}$  be its \emph{backward star}. 
For each transition $t_i$ we introduce the \emph{flow variables} $x_{t_i} \in \mathds{N}$, and 
$z_{t_i}^{\vec{q}}\in \mathds{R}$ where $\vec{q}\in Q, \vec q \neq \vec q_0$.

We are ready to define the set $F_{\vec{s},\vec{d}}$ of \emph{flow constraints}, an element of which $\vec{x}=(x_{t_1},\ldots , x_{t_n}) \in F_{\vec{s},\vec{d}}$ defines traces from the source state $\vec{s}$ to the target state $\vec{d}$.
The intuition is that each variable $x_{t_i}$ represents how many times the transition $t_i$ is traversed in the traces defined by $\vec x$.
Hereafter, we will abbreviate $F_{\vec{q_0},\vec{q_f}}$ as $F_x$, and we identify a transition through its source and target states.

An example follows.

\begin{exa}\label{ex:flow}
Figure \ref{fig:weaksafe} (top right) shows a simple service of booking, which is the composition of 
a client and a hotel contracts.

The contract of the client requires to book a room ($r$), including breakfast ($b$) and a transport service, by car ($c$) or taxi ($t$); finally it sends a signal of termination ($\overline e$).
The contract of the client is then:
\[
C=r.b.(c+t).\overline e
\]
The hotel offers a room, breakfast and taxi. Its contract is:
\[
H=\overline r.\overline t.\overline b.e
\]
Four traces accepted by the automaton  $H \otimes C$  are: 
\[
w_1 = (\overline r, r )( \blk , b)(\overline t, t)(\overline b, \blk)(e, \overline e)
\]
\[
w_2 = (\overline r, r )( \blk , b)(\blk,c)(\overline t,\blk)(\overline b,\blk)(e, \overline e)
\]
\[
w_3 = (\overline r, r )(\overline t, \blk )(\overline b, b)(\blk, t)(e, \overline e)
\]
\[
w_4 = (\overline r, r )(\overline t, \blk )(\overline b, b)(\blk, c)(e, \overline e)
\]
We now detail the flows associated with each trace giving the set of variables with value 1, all the others having value 0, because there are no loops.
The associated flows are: 
\[
w_1:
\{x_{\vec q_0 , \vec q_1},
x_{\vec q_1 , \vec q_2},
x_{\vec q_2 , \vec q_3},
x_{\vec q_3 , \vec q_4},
x_{\vec q_4 , \vec q_{10}}\}
\]
\[
w_2: 
\{x_{\vec q_0 , \vec q_1},
x_{\vec q_1 , \vec q_2},
x_{\vec q_2 , \vec q_5},
x_{\vec q_5, \vec q_6},
x_{\vec q_6, \vec q_9},
x_{\vec q_9 , \vec q_{10}} \}
\]
\[
w_3: 
\{x_{\vec q_0 , \vec q_1},
x_{\vec q_1 , \vec q_7},
x_{\vec q_7 , \vec q_8},
x_{\vec q_8 , \vec q_4},
x_{\vec q_4 , \vec q_{10}} \}
\]
\[ 
w_4: 
\{x_{\vec q_0 , \vec q_1},
x_{\vec q_1 , \vec q_7},
x_{\vec q_7 , \vec q_8},
x_{\vec q_8 , \vec q_9},
x_{\vec q_9 , \vec q_{10}} \}
\]
\end{exa}
\begin{figure}[tb]
\center
%
%
%
%
\begin{tikzpicture}[->,>=stealth',shorten >=1pt,auto,node distance=1.8cm,
                    semithick, every node/.style={scale=0.7}]
  \tikzstyle{every state}=[fill=white,draw=black,text=black]

  \node[initial,state] (A)            {$\vec{q_0}$};
  \node[state] (B) [right of=A]							{$\vec q_1$};
  \node[state] (C) [right of=B]							{$\vec q_2$};
  \node[state, accepting] (D) [right of=C]							{$\vec q_3$};
  \node[state] (E) [above of=B]							{$\vec q_4$};
  \node[state] (F) [right of=E]							{$\vec q_5$};
 
  \path (A)		edge		node{$(b, \overline b)$} (B)
   		(A)		edge		node{$( \overline a , a)$} (E)
		(B)		edge		node[right]{$(b, \overline b)$} (E)
		(B)		edge		node{$( \overline a, a)$} (C)
		(C)		edge		node{$( \blk, c)$} (D)
		(E)		edge[loop above]		node{$( \overline c, \blk)$} (E)
		(E)		edge		node{$(b, \overline b)$} (F)
		(F)		edge		node{$(\blk,a )$} (D);
		
\end{tikzpicture} 
\begin{tikzpicture}[->,>=stealth',shorten >=1pt,auto,node distance=1.8cm,
                    semithick, every node/.style={scale=0.7}]
  \tikzstyle{every state}=[fill=white,draw=black,text=black]

  \node[initial,state] (A)              {$\vec{q_0}$};
  \node[state] (B)  [below of=A]                  {$\vec{q_1}$};
  \node[state] (C)  [right of=A]                  {$\vec{q_2}$};
  \node[state] (D)  [right of=C]      			 {$\vec{q_3}$};
  \node[state] (E)  [right of=D]					{$\vec{q_4}$};
	\node[state] (F)  [below of=C]					{$\vec{q_5}$};
	\node[state] (G)  [right of=F]		{$\vec{q_6}$};
	\node[state] (H)  [below of=F]					{$\vec{q_7}$};
\node[state](L)[right of=G]			{$\vec{q_9}$};
	\node[state] (I)  [below of=L]			{$\vec{q_8}$};
\node[state,accepting] (M) [right of=E]		{$\vec q_{10}$};

  \path (A)			edge             node[left]{$(\overline{r},r)$} (B)
  		(B)			edge             node[above]{$(\blk,b)$} (C)
  		(B)			edge        	    node[left]{$(\overline{t},\blk)$} (H)
  		(C)			edge             node{$(\overline{t},t)$} (D)
  		(C)			edge             node[right]{$(\blk,c)$} (F)
  		(D)			edge             node{$(\overline{b},\blk)$} (E)
		(E)			edge		node{$(e,\overline e)$}(M)
  		(F)			edge             node{$(\overline{t},\blk)$} (G)
		(G)			edge             node{$(\overline{b},\blk)$} (L)
  		(H)			edge             node{$(\overline{b},b)$} (I)
 		(I)			edge             node{$(\blk,c)$} (L)
 		(I)			edge[bend right] node[right]{$(\blk,t)$} (E)
		(L)			edge		node[above]{$(e,\overline e)$}(M);
		
\end{tikzpicture} 
%
%
\medskip
\begin{tikzpicture}[->,>=stealth',shorten >=1pt,auto,node distance=1.8cm,
                    semithick, every node/.style={scale=0.7}]
  \tikzstyle{every state}=[fill=white,draw=black,text=black]

  \node[initial,state] (A)              {${q_0}$};
  \node[state] (B)  [right of=A]                  {${q_1}$};
  \node[state,accepting] (C)  [right of=B]                  {${q_2}$};
  \path (A)			edge             node{$a$} (B)
  		(B)			edge             node{$b$} (C)
  		(B)			edge[bend left]        	    node{$c$} (A)
  		(C)			edge[bend right]             node[above]{$d$} (A);
\end{tikzpicture} 
\caption{Top left: the product of two contract automata of Examples~\ref{ex:flow} and~\ref{ex:flow2path}; top right the booking service of Example~\ref{ex:flow}; bottom: the principal contract automaton whose flow constraints generate many traces, as discussed at the end of Example~\ref{ex:flow}. }
\label{fig:weaksafe}
\label{fig:weakagreement}
\end{figure}
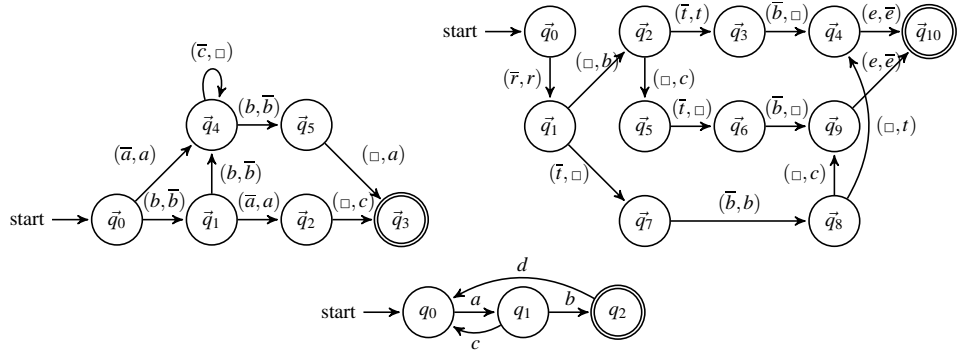
\noindent Note that a flow $\vec x$ may represent many traces that have the same balance of requests/offers for each action occurring therein.
For example, in the contract automaton at the bottom of Figure~\ref{fig:weaksafe}, the same flow 
$x_{ q_0, q_1}=3, x_{ q_1,  q_2}=2, x_{ q_2,  q_0}= x_{ q_1,  q_0}=1$
represents both $w_1 = acabdab$ and $w_2=abdacab$.
%

The following auxiliary definition introduces a notation for flow constraints.
It is beneficial in the statements of Theorems~\ref{the:flow_weak_safe},~\ref{the:flow_weak_agreement} and~\ref{the:weakly_liable_flow1} below.
\begin{defi}
Given a source state $\vec s$ and a destination state $\vec d$, the set of \emph{flow constraints} $F_{\vec{s},\vec{d}}$  from $\vec s$ to $\vec d$  is defined as:
%
%
%
\begin{align}
F_{\vec{s},\vec{d}} = \{& (x_{t_1},\ldots , x_{t_n}) \mid  \forall \vec{q}:
(\sum_{t_i \in BS(\vec{q})} x_{t_i} - \sum_{t_i \in FS(\vec{q})} x_{t_i}) = \left\{ 
\begin{array}{ll} 
-1 &  \mbox{if } \vec{q} = \vec{s}\\ 
0  & \mbox{if } \vec{q} \neq \vec{s}, \vec{d} \\
1 & \mbox{if } \vec{q} = \vec{d}
\end{array} 
\right. \  \notag \\ 
& \forall \vec{q}\neq \vec{s}, t_i. \quad 0 \leq z_{t_i}^{\vec{q}} \leq x_{t_i}, \notag \\
& \forall \vec{q}\neq \vec{s}, \, \forall \vec{q'}:\  
(\sum_{t_i \in BS(\vec{q'})} z_{t_i}^{\vec{q}} - \sum_{t_i \in FS(\vec{q'})} z_{t_i}^{\vec{q}}\ ) =  \left\{ 
\begin{array}{ll} 
- p^{\vec q}  &  \mbox{if } \vec{q'}=\vec{s}\\ 
0  & \mbox{if } \vec{q'} \neq \vec{s}, \vec{q} \\
p^{\vec q}  & \mbox{if } \vec{q'} = \vec{q}
\end{array}
\right. 
  \notag \\
 &  \qquad \text{where } \quad p^{\vec q} = \left\{
\begin{array}{ll} 
1  &  \mbox{if }\sum_{t_i \in FS(\vec{q})} x_{t_i} > 0\\ 
0  &   otherwise \\
\end{array}
\right. \qquad\qquad \} \notag
\end{align}
\end{defi}

In the definition above, the variables $z_{t_i}^{\vec{q}}$ represent $|Q|-1$ auxiliary flows and make sure that a flow $\vec x$ represents valid runs only, i.e. they guarantee that there are no disconnected cycles with a positive flow.
A more detailed discussion is in Example \ref{ex:flow2path} below.
Note that the values of $z_{t_i}^{\vec{q}}$ are \emph{not} integers, and so we are defining Mixed Integer Linear Programming problems that have efficient solutions~\cite{Hemmecke10}.
 
We eventually define a set of variables $a^i_{t_j}$ for each action and each transition, that take the value -1 for requests, 1 for offers, and 0 otherwise; they help counting the difference between offers and requests of an action in a flow (recall that $I_l$ contains the indexes of the requests).
%
\[\forall t_j =(\vec{q},\vec{a},\vec{q}')\in T, \forall i \in I_l: \quad a^i_{t_j} = \left\{ 
\begin{array}{ll} 
1 &  \mbox{if } Obs(\vec{a})= \overline{a^i}\\ 
-1  & \mbox{if } Obs(\vec{a})=a^i \\
0 &  otherwise 
\end{array} 
\right. \]

\begin{exa}\label{ex:flow2path}

Figure \ref{fig:weaksafe} (top left) depicts the contract $A \otimes B$, where 
\[
A=\overline a. \overline {c}^*.b + b.(b.\overline{c}^*.b+\overline a)
\qquad \qquad
B = a.\overline b.a + \overline b.(\overline b.\overline b.a+ a.c)
\]
To check whether there exists a run recognising a trace $w$ with less or equal requests than offers (for each action) we solve $\sum_{t_j} a^i_{t_j} x_{t_j} \geq 0$, for $\vec{x} \in F_x$.


We illustrate how the auxiliary variables $z_{t_i}^{\vec{q}}$ ensure that the considered solutions represent valid runs. 
Consider 
the following assignment to $\vec{x}$:
$x_{\vec q_ 0, \vec q_1}= x_{\vec q_1 , \vec q_2}=x_{\vec q_2 , \vec q_3}=1, x_{\vec q_4 , \vec q_4}\geq 1$, and null everywhere else.
It does not represent valid runs, because the transition $(\vec q_4, (\overline c,\blk), \vec q_4)$ cannot be fired in a run that only takes transitions with non-null values in $\vec{x}$. 
However, the constraints on the flow $\vec{x}$ are satisfied (e.g.\ we have $\sum_{t_j \in FS(\vec{q_4})} x_{t_j} =\sum_{t_j \in BS(\vec{q_4})} x_{t_j}$).
Now the constraints on the auxiliary $z_{t_i}^{\vec{q}}$ play their role, checking if a node is reachable from the initial state on a run defined by $\vec x$.
The assignment above is not valid since for $z^{\vec q_4}$ we have :
\[
0 \leq z^{\vec q_4}_{(\vec q_0, \vec q_4)} \leq x_{(\vec q_0, \vec q_4)}=0
\]
\[
0 \leq z^{\vec q_4}_{(\vec q_1, \vec q_4)} \leq x_{(\vec q_1, \vec q_4)}=0
\]
\[
0 \leq z^{\vec q_4}_{(\vec q_4, \vec q_5)} \leq x_{(\vec q_4, \vec q_5)}=0
\]
Hence 
$\sum_{t_j \in BS(\vec{q_4})} z_{t_j}^{\vec{q_4}}= z^{\vec q_4}_{(\vec q_4,\vec q_4)}, 
\sum_{t_j \in FS(\vec{q_4})} z_{t_j}^{\vec{q_4}}= z^{\vec q_4}_{(\vec q_4,\vec q_4)}$ 
and we have: 
\[
\sum_{t_j \in BS(\vec{q_4})} z_{t_j}^{\vec{q_4}} -  \sum_{t_j \in FS(\vec{q_4})} z_{t_j}^{\vec{q_4}}= 0 \neq 1  = p^{\vec q_4}
\]

Finally, note in passing that there are no valid flows $\vec x \in F_x$ for this problem.

More importantly, note that the auxiliary variables $z_{t_i}^{\vec{q}}$ are not required to have integer values, which is immaterial for checking that those solutions represent valid runs, but makes finding them much easier.

\end{exa}

The main results of this section follow.


\begin{restatable}{thm}{theflowweaksafe}\label{the:flow_weak_safe}\ 
Let $\vec{v}$ be a binary vector.
Then a contract automaton $\mathcal{A}$ is \emph{weakly safe} if and only if $\textsf{min } \gamma \geq0$ where:
\[ \sum_{i \in I_l} v_i\sum_{t_j \in T} a^i_{t_j} x_{t_j} \leq \gamma  \quad \sum_{i \in I_l} v_i=1 \quad \forall i \in I_l.\ v_i \in \{ 0,1\} \quad (x_{t_1} \ldots x_{t_n} ) \in F_x 
\quad \gamma \in \mathds{R}
\]
\end{restatable}

The minimum value of $\gamma$ selects the trace and the action $a$ for which the difference between the number of offers and requests is the minimal achievable from $\mathcal{A}$. 
If this difference is non-negative, there will always be enough offers matching the requests, and so $\mathcal{A}$ will never generate a trace not in $\mathfrak{W}$.
In other words, $\mathcal{A}$ is \emph{weakly safe}, otherwise it is not.

\begin{exa}\label{ex:weaksafe}
Consider again Example~\ref{ex:flow} and let $a^1 = r$, $a^2 = b$, $a^3 = t$, $a^4 = c$, $a^5=e$. 
\\
If $v_1 = 1$, for each flow $\vec{x} \in F_x$, we have that 
$\sum_{t_j} a^1_{t_j} x_{t_j} = 0$ (for $i \neq 1$, we have $v_i = 0$). 
This means that the request of a room is always satisfied.
Similarly for breakfast and the termination signal $e$.
If $v_3 = 1$, for the flow representing the traces $w_1,w_3$ we have 
$\sum_{t_j} a^3_{t_j} x_{t_j} = 0$,
while for the flow representing the traces $w_2,w_4$ the result is 1.
The requests are satisfied also in this case.
Instead, when $v_4 = 1$, for the flow representing the traces $w_1,w_4$ we have 
$\sum_{t_j} a^4_{t_j} x_{t_j} = 0$, but for the flow representing $w_2,w_3$, the result is $-1$.
Hence $\textsf{min } \gamma = -1$, and the contract automaton $H \otimes C$ is not \emph{weakly safe}, indeed we have $w_2,w_3 \not \in \mathfrak{W}$.
\end{exa}

In a similar way, we can check if a contract automaton offers a trace in weak agreement.

\begin{restatable}{thm}{theflowweakagreement}\label{the:flow_weak_agreement}
The contract automaton $\mathcal{A}$ admits weak agreement if and only if $\textsf{max } \gamma \geq 0$ where
\[
\forall i \in I_l. \sum_{t_j \in T} a^i_{t_j}\; x_{t_j} \geq \gamma  \quad  (x_{t_1} \ldots x_{t_n} )  \in F_x  \quad
\gamma \in \mathds{R}
\]
\end{restatable}

The maximum value of $\gamma$ in Theorem \ref{the:flow_weak_agreement} selects the trace $w$ that maximises the least difference between offers and requests of an action in $w$. If this value is non-negative, then there exists a trace $w$ such that for all the actions in it, the number of requests is less or equal than the number of offers.
In this case, $\mathcal{A}$ admits weak agreement; otherwise it does not.

\begin{exa}\label{ex:weakagreement}
In Example \ref{ex:flow}, $\textsf{max }\gamma=-1$ for the flows representing the traces $w_2,w_3$ and \mbox{$\textsf{max }\gamma=0$} for those of the traces $w_1,w_4$, that will be part of the solution and are indeed in weak agreement.
Consequently, $H \otimes C$ admits weak agreement.
%
%
\end{exa}

We now define the \emph{weakly liable} principals: those who perform the first transition $t$ of a run such that after $t$ it is not possible any more to obtain a trace in $\mathfrak{W}$, i.e.\  leading to traces $w \in \mathscr{L}(\mathcal A) \setminus \mathfrak W$ that cannot be extended to $ww' \in \mathscr{L}(\mathcal{A})\cap \mathfrak W$.

\begin{defi}
\label{def:weaklyliable}
Let $\mathcal{A}$ be a contract automaton and let $w=w_1\vec{a}w_2$ such that $w \in \mathscr{L}(\mathcal{A})\setminus \mathfrak W$, $\forall w'.ww' \not\in \mathscr{L}(\mathcal{A})\cap \mathfrak W, \forall w_3.w_1\vec{a}w_3 \not \in \mathscr{L}(\mathcal{A})\cap \mathfrak W$ and $\exists w_4.w_1w_4 \in \mathscr{L}(\mathcal{A})\cap \mathfrak W$.

The principals $\Pi^i(\mathcal{A})$ such that $\vec{a}_{(i)} \neq \blk$ are \emph{weakly liable} and form the set  
$W\!Liable(\mathcal{A},w_1\vec{a})$.

Let $W\!Liable(\mathcal{A})=\{i \mid \exists w$ such that $i \in W\!Liable(\mathcal{A},w)\}$ be the set of all \emph{potentially weakly liable} principals in $\mathcal{A}$.
\end{defi}

For computing the set $W\!Liable(\mathcal A)$ we optimise a network flow problem for a transition $\overline t$ to check if there exists a trace $w$ in which $\overline t$ reveals some weakly liable principals. 
By solving this problem for all transitions we obtain the set $W\!Liable(\mathcal A)$.


\begin{restatable}{thm}{theweaklyliableflow}\label{the:weakly_liable_flow1}
%
The principal $\Pi^i(\mathcal{A})$ of a contract automaton $\mathcal A$ is \emph{weakly liable}
if and only if 
there exists a transition $\overline{t}=(\vec{q_s},\vec{a},\vec{q_d}) \in T_{\mathcal A}$, $\vec{a}_{(i)} \neq \blk$ such that $\gamma_{\overline t} <0$, where
\[
\gamma_{\overline t}=\textsf{min }\{ g(\vec{x}) \mid \vec{x} \in F_{\vec{q_0},\vec{q_s}}, 
 \  \vec{y} \in F_{\vec{q_s},\vec{q_f}},
 \ \forall i \in I_l. \sum_{t_j \in T} a^i_{t_j}(x_{t_j} + y_{t_j})\geq 0  \} \notag \\
\]
\[
g(\vec{x})=\textsf{max } \{ \gamma \mid  \vec{u} \in F_{\vec{q_d},\vec{q_f}} ,  
\ \forall i \in I_l. \sum_{t_j \in T} a^i_{t_j}(x_{t_j} + u_{t_j}) + a^i_{\overline{t}} \geq \gamma, \gamma \in \mathds{R} \} \notag \\
\]
\end{restatable}

\begin{figure}[tb]
\center
\begin{tikzpicture}[->,>=stealth',shorten >=1pt,auto,node distance=2.3cm,
                    semithick, every node/.style={scale=0.6}]
  \tikzstyle{every state}=[fill=white,draw=black,text=black]
  \node[initial,state] 				 (A)                    {$\vec{q_0}$};
 \node[state] 				         (B) [right=0.7cm and 4cm of A]   {$\vec{q_s}$};
  
  \node[state] 				         (E) [below right=0.5cm of B]   {$\vec{q_d}$};
  \node[state,accepting]			 (F) [right=0.7cm and 4cm of B] {$\vec{q_f}$};

  \path (A)  edge [decorate, decoration={snake}]       				   node{$\vec x$} (B)
		  (B)  	edge [decorate, decoration={snake}]            	     	   node {$\vec y$} (F)
				edge         				       node {$\overline t$} (E)
		 (E) edge[decorate, decoration={snake}]  node {$\vec u$} (F);
\end{tikzpicture}
\caption{The three flows computed by Theorem~\ref{the:weakly_liable_flow1}}
\label{fig:weakliable}
\end{figure}
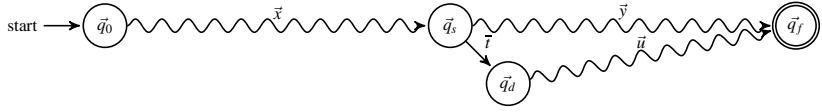
Figure~\ref{fig:weakliable} might help to understand how the flows $\vec x,\vec y$ (and $\vec u$) and the transition $\overline t$ are composed to obtain a path from the initial to the final state.
Intuitively, the flow defined above can be seen as split into three parts: the flow $\vec x$ from $\vec{q_0}$ to $\vec{q}_s$, the flow $\vec y$ from $\vec{q_s}$ to $\vec{q_f}$, and the flow $\vec u$ from $\vec{q_d}$ to $\vec{q_f}$, computed through the function $g$. 

This function takes as input the flow $\vec x$ and selects a flow $\vec u$ such that, by concatenating $\vec x$ and $\vec u$ through $\overline t$, we obtain a trace $w$ where the least difference between offers and requests is maximised for an action in $w$. 
Using the same argument of Theorem~\ref{the:flow_weak_agreement}, if the value computed is negative, then there  exists no flow $\vec u$ that composed with $\vec x$ selects traces in weak agreement.

Finally $\gamma_{\overline t}$ yields the minimal result of $g(\vec x)$, provided that there exists a flow $\vec y$, that combined with $\vec x$ represents only traces in weak agreement. 
If $\gamma_{\overline t} < 0$ then the transition $\overline{t}$ identifies some \emph{weakly liable} principals. 
Indeed the flow $\vec x$ represents the traces $w$ such that 
(1) $\exists w_1$, represented by $\vec y$, with $ww_1 \in \mathscr{L}(\mathcal A) \cap \mathfrak W$ and 
(2) $\forall w_2$, represented by $\vec u$, with $w\vec{a}w_2 \in \mathscr{L}(\mathcal A) \setminus \mathfrak W$.
Note that if a flow $\vec x$ reveals some weakly liable principals, the minimisation carried on by $\gamma_{\overline t}$ guarantees that the relevant transition $\overline t$ is found. 
Finding the weakly liable principals is a hard task, and belongs to the family of bilevel problems~\cite{bilevel}. 
Basically, these problems contain two optimisation problems, one embedded in the other,
and finding optimal solutions to them is still a hot research topic.

\begin{exa}\label{ex:weakliable}
In Figure \ref{fig:weaksafe} (top right), the transitions $(\vec{q_2},(\blk,c),\vec{q_5})$ and $(\vec q_8, (\blk,c), \vec q_9)$ reveal the second principal (i.e. $C$) \emph{weakly liable}. 
Indeed the trace $(\overline r,r)(\blk,b)$ ending in $\vec q_2$ can be extended to one in weak agreement, while $(\overline r,r)(\blk,b)(\blk,c)$ cannot. 
Also the trace $(\overline r,r)(\overline t,\blk)(\overline b,b)$ can be extended to one in weak agreement while $(\overline r,r)(\overline t,\blk)(\overline b,b)(\blk,c)$ cannot.

For the transition $(\vec{q_2},(\blk,c),\vec{q_5})$ we have the trace $(\overline r,r)(\blk,b)$ for the flow $\vec x$ and $(\overline t,t)(\overline b,\blk)(e,\overline e)$ for the flow $\vec y$, and we have $\forall i \in I_l. \sum_{t_j \in T} a^i_{t_j}(x_{t_j} + y_{t_j})\geq 0$.  
Note that if we select as flow $\vec y$ the trace $(\blk,c)(\overline t,\blk)(\overline b,\blk)(e,\overline e)$ then the constraints  $\forall i \in I_l. \sum_{t_j \in T} a^i_{t_j}(x_{t_j} + y_{t_j})\geq 0$ are not satisfied for the action $a^4=c$ (recall Example~\ref{ex:weaksafe}).
For the flow $\vec u$ the only possible trace is  $(\overline t,\blk)(\overline b,\blk)(e,\overline e)$, and $\textsf{max } \gamma=-1=\gamma_{(\vec{q_2},(\blk,c),\vec{q_5})}$ since $\sum_{{t_j} \in T} a^4_{t_j} (x_{t_j} + u_{t_j}) +(-1)=-1$.

For the transition $(\vec q_8, (\blk,c), \vec q_9)$ the flow $\vec x$ selects the trace $(\overline r, r)(\overline t,\blk)(\overline b,b)$, the flow $\vec y$ selects the trace $(\blk,t)(e,\overline e)$, since the other possible trace, that is $(\blk,c)(e,\overline e)$, does not respect the constraints for the action $a^4$ (i.e. $c$). Finally, for the flow $\vec u$ we have the trace $(e,\overline e)$, and as the previous case $\textsf{max } \gamma=-1 = \gamma_{(\vec q_8, (\blk,c), \vec q_9)}$.
\end{exa}

\section{Automata and Logics for Contracts}
\label{sect:logic}
%
%

Recently, the problem of expressing contracts and of controlling that the principals in a composition fulfil their duties has been studied in Intuitionistic Logic, where a clause is interpreted as a principal in a contract, in turn rendered as the conjunction of several clauses.
Actually, the literature only considers fragments of Horn logics because they have an immediate interpretation in terms of contracts.
More in detail, these Horn  fragments avoid contradiction clauses, as well as formulae with a single Horn clause. 
These two cases are not relevant because their interpretation as contracts makes little sense, e.g. a contract requires at least two parties.
It turns out that these theories can be interpreted as contract automata, without much effort. 

The first logic we consider is Propositional Contract Logic (PCL)~\cite{BartolettiZ10} able to deal with circular obligations.
Its distinguishing feature is a new implication, called \emph{contractual implication}, that permits to assume as true the conclusions even before the premises have been proved, provided that they will be in the future.
Roughly, a contract is rendered as a Horn clause, and a composition is a conjunction of contracts.
When a composition is provable, then all the contracts are fulfilled, i.e.\ all the requests (represented as premises of implications) are entailed.

In the next sub-sections, we translate a fragment of the Horn formulae of Propositional Contract Logic into contract automata, and we prove that a formula is provable if and only if the corresponding contract automaton admits agreement.  

We then study the connection between contract automata and the Intuitionistic Linear
Logic with Mix ($ILL^{mix}$)\cite{benton1995mixed}. 
This logic is used for modelling exchange of resources between partners with the possibility of recording debts (requests satisfied by a principal offer but not yet paid
back by honouring one of its requests),
and has been recently given a model in terms of Petri Nets~\cite{DebitNets}.
In this logic one can represent the depletion of resources, in our case of offers, that also here can be put forward before a request occurs.
Again, we translate a fragment of Horn formulae as contract automata, and we prove that a theorem there corresponds to an automaton that admits agreement. 

Our constructions have been inspired by analogous ones~\cite{DebitNets}; ours however offer a more flexible form of compositionality. Indeed, for checking if two separate formulas are provable, it suffices to check if the composition of the two corresponding automata is still in agreement. 
If the two automata are separately shown to be safe, then their composition is in agreement
due to Theorem~\ref{the:composition}.
With Debit Petri Nets~\cite{DebitNets} instead, one needs to recompute the whole translation for the composed formulas, while here we propose a modular approach.


\subsection{Propositional Contract Logic}
\label{sect:1NPCL}

The usual example for showing the need of circular obligations is Example~\ref{ex:biketoy}.
In the Horn fragment of PCL we use, called H-PCL, the contracts of Alice and Bob make use of the new contractual implication 
$F \twoheadrightarrow F'$, whose intuition is that the formula $F'$ is deducible, provided that later on in the proof also $F$ will be deduced.

According to this intuition and elaborating over Example~\ref{ex:biketoy}, Alice's contract (\emph{I offer you my aeroplane provided that in the future you will lend me your bike}) and Bob's (\emph{I offer you my bike provided that in the future you will lend me your aeroplane}) are rendered as
$
bike \twoheadrightarrow airplane$, 
$
airplane \twoheadrightarrow bike 
$, respectively.
Their composition is obtained by joining the two, and one represents that both Alice and Bob are proved to obtain the toy they request by 
\[
((bike \twoheadrightarrow airplane) \wedge (airplane \twoheadrightarrow bike)) \vdash (bike \wedge airplane)
\]
In words, the composition of the two contracts entails \emph{all} the requests (\emph{bike} by Alice and \emph{airplane} by Bob).
%
%
We now formally introduce the fragment of H-PCL~\cite{BartFSEN2013,BartCPZ15} that has a neat interpretation in contract automata, under the assumption that a principal cannot offer and require the same.

\begin{defi}[H-PCL]\label{def:hornpcl}
Assume a denumerable set of atomic formulae $Atoms=\{a,b,c, \ldots\}$ indexed by $i \in I, j \in J$ where $I$ and $J$ are finite set of indexes; then the \emph{H-PCL formulae} $p, p', \dots$ and the clauses $\alpha, \alpha_i, ... $ are generated by the following BNF grammar
\smallskip

\begin{tabular}{ll}
$ p  ::= \bigwedge_{i \in I}\alpha_i \qquad $  & 
$\alpha ::=\bigwedge_{j \in J}a_j  \hspace{3pt} |\hspace{3pt} ( \bigwedge_{j \in J}a_j) \rightarrow b  \hspace{3pt} |\hspace{3pt} ( \bigwedge_{j \in J}a_j) \twoheadrightarrow b $ \\
 & $\text{where } |I| \geq 2, |J| \geq 1, i \neq j \text{ implies } a_i \neq a_j , \text{ and } \forall j\in J. \, a_j \neq b  $\\
\end{tabular}
\smallskip 

\noindent
Also, let $\lambda(p)$ be the conjunction of all atoms in $p$. 
\end{defi}

\begin{figure}[tb]
\[
\irule{\Gamma \vdash q}{ \Gamma \vdash p \twoheadrightarrow q} Zero 
\quad\qquad\qquad
\irule{\Gamma, p \twoheadrightarrow q,r \vdash p \quad \Gamma, p \twoheadrightarrow q,p \vdash q}{\Gamma, p \twoheadrightarrow q \vdash r}  Fix
\]
\[
\irule{\Gamma, p \twoheadrightarrow q,a \vdash p \quad  \Gamma, p \twoheadrightarrow q,q \vdash b}{\Gamma, p \twoheadrightarrow q \vdash a \twoheadrightarrow b}  PrePost
\]
\caption{The three rules of PCL for the contractual implication.}
\label{fig:gentzen}
\end{figure}
In Figure~\ref{fig:gentzen} we recall the three rules of the sequent calculus for the contractual implication~\cite{BartolettiZ10,BartolettiZ09}; the others are the standard ones of the Intuitionistic Logic and are in the appendix,
 Figure~\ref{fig:fullgentzen}.

As anticipated, in H-PCL  all requests of principals are satisfied  if and only if the conjunction $p$ of the contracts of all principals entails all the atoms mentioned.

\begin{defi}
The formula $p$ represents a composition whose principals respect all their obligations if and only if
$p \vdash \lambda(p)$.
\end{defi} 


Below, we define the translation from H-PCL formulae to contract automata.
A simple inspection of the rules below suffices to verify that the obtained automata are deterministic.

\begin{defi}[From H-PCL to CA]
\label{def:translationPCL}
A H-PCL formula, with sets of indexes $I$ and $J$ as in Definition~\ref{def:hornpcl}, is translated into a contract automaton by  the following rules, where 
$\mathcal{P} = \{ q \cup \{ \ast \} \mid q \in 2^J \}$: 
\smallskip

\noindent
$\llbracket \bigwedge_{i \in I}\alpha_i \rrbracket = \Large{\boxtimes}_{i \in I} \llbracket \alpha_i \rrbracket$\\ 

\noindent
$\llbracket \bigwedge_{j \in J} a_j \rrbracket = \langle \{ \{\ast \} \},  \{\ast \}, \emptyset, 
\{ \overline{a_j} \mid j \in J\}, \{ ( \{\ast \}, \overline{a_j}, \{\ast \}) \mid  \overline{a_j} \in A^o\} ,  \{ \{\ast \} \} \rangle$ 

\noindent
\begin{flalign}
\llbracket ( \bigwedge_{j \in J}a_j) \rightarrow b  \rrbracket =  & 
\langle   \mathcal{P} ,
	 J \cup \{ \ast \},\{ a_j \mid j \in J\},\{\overline{b}\}, && \notag \\
& \{  (J' \cup \{j\}, a_j, J') \mid J' \cup \{j\} \in   \mathcal{P}, j \in J\} \cup \{ (\{\ast \},\overline{b},\{\ast \})\} , \{  \{\ast \} \}  \rangle  && \notag
\end{flalign}

\noindent
\begin{flalign}
\llbracket (\bigwedge_{j \in J}a_j) \twoheadrightarrow b \rrbracket =  &
\langle \mathcal{P} , J \cup \{ \ast \},\{ a_j \mid j \in J\},\{\overline{b}\}, && \notag \\
 & \{ ( J' \cup \{j\}, a_j, J') \mid  J' \cup \{j\} \in  \mathcal{P}, j \in J \} \cup  
\{ (q, \overline{b}, q) \mid q  \in   \mathcal{P}\}, \{  \{  \ast \} \}   \rangle && \notag 
\end{flalign}
\end{defi}

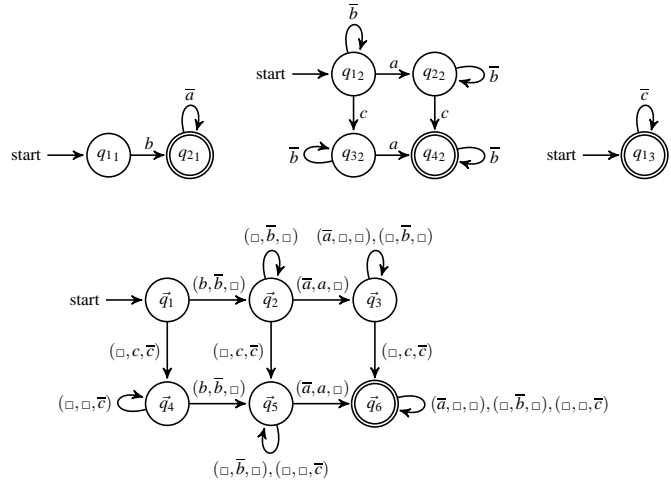
\begin{figure}[tb]
\center
\begin{tikzpicture}[->,>=stealth',shorten >=1pt,auto,node distance=1.8cm,
                    semithick, every node/.style={scale=0.6}]
  \tikzstyle{every state}=[fill=white,draw=black,text=black]

  \node[initial,state] (A)                    {${q_1}_1$};
  \node[state,accepting]         (B) [right of=A] {${q_2}_1$};

  \path (A) edge             					node{$b$} (B)
		  (B) edge [loop above]              node {$\overline{a}$} (B);
\end{tikzpicture} \quad
\begin{tikzpicture}[->,>=stealth',shorten >=1pt,auto,node distance=1.8cm,
                    semithick, every node/.style={scale=0.6}]
  \tikzstyle{every state}=[fill=white,draw=black,text=black]

  \node[initial,state] 				 (A)                  {${q_1}_2$};
  \node[state] 				         (B) [right of=A] {${q_2}_2$};
  \node[state]						 (C) [below of=A] {${q_3}_2$};
  \node[state,accepting]         (D) [right of=C] {${q_4}_2$};

  \path (A) edge             					node{$a$} (B)
			   edge             					node{$c$} (C)
			   edge  [loop above] 			node{$\overline b$} (A)
		  (B) edge   [loop right]         node {$\overline{b}$} (B)
				edge         				    node {$c$} (D)
		  (C) edge                            node {$a$} (D)
				edge [loop left]              node {$\overline b$} (C)
		  (D) edge [loop right]             node {$\overline b$} (D);
\end{tikzpicture} \quad
\begin{tikzpicture}[->,>=stealth',shorten >=1pt,auto,node distance=1.8cm,
                    semithick, every node/.style={scale=0.6}]
  \tikzstyle{every state}=[fill=white,draw=black,text=black]

  \node[initial,state,accepting] (A)                    {${q_1}_3$};

  \path (A)edge [loop above]              node {$\overline{c}$} (A);
\end{tikzpicture}  \\ \vspace{.5cm}
\begin{tikzpicture}[->,>=stealth',shorten >=1pt,auto,node distance=2.3cm,
                    semithick, every node/.style={scale=0.6}]
  \tikzstyle{every state}=[fill=white,draw=black,text=black]

  \node[initial,state] 				 (A)                    {$\vec{q_1}$};
  \node[state] 				         (B) [right of=A]   {$\vec{q_2}$};
  \node[state]						 (C) [right of=B]   {$\vec{q_3}$};
  \node[state]         				 (D) [below of=A] {$\vec{q_4}$};	
  \node[state] 				         (E) [right of=D]   {$\vec{q_5}$};
  \node[state,accepting]			 (F) [right of=E] {$\vec{q_6}$};

  \path (A)  edge             				   node{$(b,\overline b,\blk)$} (B)
			    edge             				   node[left]{$(\blk,c,\overline{c})$} (D)
		  (B)  edge   [loop above]           node {$(\blk, \overline b, \blk)$} (B)
				edge         				       node {$(\overline{a},a,\blk)$} (C)
				edge         				       node[left] {$(\blk,c,\overline{c})$} (E)
		  (C)  edge                             node {$(\blk,c,\overline{c})$} (F)
				edge [loop above]               node {$(\overline a, \blk, \blk), (\blk, \overline b, \blk)$} (C)
		  (D)  edge [loop left]             node {$(\blk, \blk, \overline c)$} (D)
				edge                             node{$(b,\overline b,\blk)$} (E)
		  (E)  edge[loop below]	               node {$ (\blk, \overline b, \blk), (\blk, \blk, \overline c)$} (E)
				edge              	     	   node {$(\overline{a},a,\blk)$} (F)
		   (F)  edge [loop right]             node {$(\overline a, \blk, \blk),(\blk, \overline b, \blk), (\blk, \blk, \overline c)$} (F);
\end{tikzpicture} \quad
\caption{The contract automata of Examples~\ref{ex:PCL} and~\ref{ex:hpcltranslation}, top from left to right:  $\llbracket Alice \rrbracket, \llbracket Bob \rrbracket, \llbracket Charlie \rrbracket$;
			  bottom: $\mathcal{K}_{\llbracket Alice \rrbracket \otimes \llbracket Bob \rrbracket \otimes \llbracket Charlie \rrbracket}$.
				}
\label{fig:pclautomata}
\end{figure}

As expected, a Horn formula is translated as the product of the automata raising from its components $\alpha_i$. 
In turn, a conjunction of atoms yields an automaton with a single state and loops driven by offers in bijection with the atoms.
Each state stores the number of requests that are waiting to fire, and $\{*\}$ stands for no requests.
A (standard) implication shuffles all the requests corresponding to the premises $a_j$ and then has the single offer corresponding to the conclusion $b$.
A contractual implication is similar, except that the offer ($\overline b$ in the definition) can occur at \emph{any} position in the shuffle, and from there onwards it will be always available. 
Note that there is no control on the number of times an offer can be taken, as H-PCL is not a linear logic.

\begin{exa}
Consider again Example~\ref{ex:biketoy}, and let us modify it to better illustrate some peculiarities of H-PCL.
Assume then that there are three kids: Alice, Bob and Charlie, who want to share some toys of theirs: a bike $b$, an aeroplane $a$ and a car $c$. 
The contract of Alice says 
``I will lend you my aeroplane provided that you lend me your bike''. 
The contract of Bob says 
``I will lend you my bike provided that in the future you will lend me your aeroplane and your car''. 
The contract of Charlie says ``I will lend you my car''.
The contract of Alice is expressed by the classical implication $b \rightarrow a$.
The contract of Bob is $(a \wedge c) \twoheadrightarrow b$, while the contract of Charlie is simply $c$.
The three contracts reach an agreement: the conjunction of the 
formulae representing the contracts entails all its atoms, that is 
$(b \rightarrow  a) \wedge ( (a \wedge c) \twoheadrightarrow b) \wedge c \vdash  a \wedge c \wedge b$.  

Figure~\ref{fig:pclautomata} shows the translation of $Alice \wedge Bob \wedge Charlie$, according to Definition~\ref{def:translationPCL}.
It is immediate verifying that the automaton is safe, since all its traces are in agreement.
\label{ex:PCL}
\end{exa}

The following proposition helps to understand the main result of this section.

\begin{restatable}{prop}{propuno}
\label{prop:1}
Given a H-PCL formula $p$ and the automaton $\llbracket p \rrbracket=\langle Q, q_0, A^{r}, A^{o}, T, F \rangle$:
\begin{enumerate}
\item $F = \{\vec q =\langle \{\ast \}, \ldots ,\{\ast \} \rangle\}$, and all $(\vec q, \vec a, \vec{q'})$ are such that 
$\vec {q'} = \vec q$ and ${\vec a}$ is an offer;

\item every state $\vec{q}=\langle J_1, \ldots, J_n \rangle$ has as many request or match outgoing transitions as the request actions prescribed by $\bigcup_{i \in 1 \ldots n} J_i$;
\item $\llbracket p \rrbracket$ is deterministic.
\end{enumerate}
\end{restatable}

As said above, when seen in terms of composed contracts, the formula $p \vdash \lambda(p)$ expresses that all the requests made by principals in $p$ must be fulfilled sooner or later.
We now show that the contract automaton $\llbracket p \rrbracket$ admits agreement if and only if $p \vdash \lambda(p)$ is provable.

\begin{restatable}{thm}{theNPCLagreement}\label{the:1NPCLagreement}
Given a H-PCL formula $p$ we have 
 $p \vdash \lambda(p) $ if and only if  $\llbracket p \rrbracket$ admits agreement.
\label{the:PCL}
\end{restatable}

We have constructively proved that a formula $p$ fulfils all its obligations if and only if the corresponding automaton $\llbracket p \rrbracket$ admits agreement.  
Interestingly, a contractual implication $a \twoheadrightarrow b$ corresponds to a contract automaton that is enabled to fire the conclusion $b$ at each state; while for the standard implication $c \rightarrow d$ the conclusion is available only after the premise $c$ has been satisfied.

\begin{exa}
Consider Example~\ref{ex:PCL}. The conjunction of all the formulas entails its atoms, indeed
the corresponding translation into contract automata displayed in Figure~\ref{fig:pclautomata}
admits agreement.
\end{exa}

Needless to say, the provability of $p \vdash \lambda(p) $ implies that $\llbracket p \rrbracket$ admits weak agreement.
However, the implication is in one direction only, as shown by the following example.

\begin{exa}
Consider the H-PCL formula $p=(b \rightarrow a) \wedge (a \rightarrow b)$. 
We have that $\llbracket p \rrbracket$ does not admit agreement and
$p \not \vdash \lambda(p)$. Nevertheless $\llbracket p \rrbracket$
  admits weak agreement. 
For example,  
$(b,-)(\overline a,a)(-,\overline b) \in \mathscr{L}(\llbracket p \rrbracket)$ is a trace in weak agreement.
\end{exa}

%

As a matter of fact, weak agreement implies provability when a formula $p$ contains no (standard) implications, as stated below.

%
%
%
%

\begin{restatable}{thm}{thepclweakautomata}\label{the:pclweakautomata}
Let $p$ be a H-PCL formula with no occurrence of standard implications $\rightarrow$, then \mbox{$p \vdash \lambda(p)$} if and only if $\llbracket p \rrbracket$ admits weak agreement.
\end{restatable}

This result helps to gain insights on the relation between the contractual implication $\twoheadrightarrow$ and the property of weak agreement. Indeed, checking weak agreement on a contract automaton $\llbracket p \rrbracket$ is equivalent to prove that the formula $p$ fulfils all its obligations (i.e. $p \vdash \lambda(p)$) \emph{only if}  $p$ contains no standard implication $\rightarrow$.

\subsection{Intuitional Linear Logic with Mix}
%
In this sub-section, we will interpret a fragment of the Intuitionistic Linear
Logic with Mix (\illmix)~ \cite{benton1995mixed} in terms of contract automata.
Originally, this logic has been used for modelling exchange of resources between partners with the possibility of recording debts, through the so-called \emph{negative atoms}.
Below, we slightly modify Example~\ref{ex:PCL} to better illustrate some features of \illmix.
%
%
%
%
\begin{exa}
Alice, Bob and Charlie want to share their bike, aeroplane and car, according to the same contracts declared in Example~\ref{ex:PCL}.
In \illmix\ the contract of Alice is expressed by the linear implication $b \multimap a$;
the contract of Bob is $a^\bot \otimes c^\bot \otimes b$ ($\otimes$ is the tensor product of Linear Logic); 
the contract of Charlie is the offer $c$.
The intuition is that a positive atom, e.g.\ $c$ in the contract of Charlie, represents a resource that can be used; similarly for the $b$ of Bob.
Instead, the negative atoms ($a^\bot$ and $c^\bot$ of Bob) represent missing resources that however can be taken on credit to be honoured later on.
The implication of Alice says that the resource $a$ is produced by consuming $b$, provided $b$ is available.
(There are some restrictions on the occurrences of negative atoms made precise below).
The composition (via tensor product) of the three contracts is successful, in that all resources are exchanged and all debts honoured. 
Indeed, it is possible to prove that all the negative atoms, i.e.\ all the requests, will be eventually satisfied.
In this case we have that all the resources are consumed, and that the following sequent is provable: $Alice\otimes Bob \otimes Charlie \vdash $.  
\label{ex:illmix}
\end{exa}

\begin{figure}[tb]
\center
\begin{tikzpicture}[->,>=stealth',shorten >=1pt,auto,node distance=1.8cm,
                    semithick, every node/.style={scale=0.6}]
  \tikzstyle{every state}=[fill=white,draw=black,text=black]

  \node[initial,state] (A)                    {${q_1}_1$};
  \node[state]        			 (B) [right of=A] {${q_2}_1$};
  \node[state,accepting]         (C) [right of=B] {${q_3}_1$};
  \path (A) edge             					node{$b$} (B)
		  (B) edge            node {$\overline{a}$} (C);
\end{tikzpicture} \quad
\begin{tikzpicture}[->,>=stealth',shorten >=1pt,auto,node distance=1.8cm,
                    semithick, every node/.style={scale=0.6}]
  \tikzstyle{every state}=[fill=white,draw=black,text=black]

  \node[initial,state] 				 (A)                  {${q_1}_2$};
  \node[state] 				         (B) [right of=A] {${q_2}_2$};
  \node[state]						 (C) [below of=A] {${q_3}_2$};
  \node[state]         (D) [right of=C] 			{${q_4}_2$};
   \node[state] 				 (E) [right of=B]                 {${q_5}_2$};
  \node[state,accepting] 				         (F) [right of=E] {${q_6}_2$};
  \node[state]						 (G) [above of=E] {${q_7}_2$};
  \node[state]         (H) [right of=G] {${q_8}_2$};

  \path (A) edge             					node{$a$} (B)
			   edge             					node{$c$} (C)
			   edge 			node{\rotatebox{45}{$\overline b$}} (G)
		  (B) edge           node[below]{\hspace{-55pt}\rotatebox{45}{$\overline b$}} (H)
				edge         				    node {$c$} (D)
		  (C) edge                            node [below]{$a$} (D)
				edge               node {\rotatebox{45}{$\overline b$}}(E)
		  (D) edge              node{\rotatebox{45}{$\overline b$}} (F)
		  (G) edge             					node{$a$} (H)
			   edge             					node{$c$} (E)
		  (H)	edge         				    node {$c$} (F)
		  (E) edge                            node {$a$} (F);
\end{tikzpicture} \\ \vspace{5mm}
\begin{tikzpicture}[->,>=stealth',shorten >=1pt,auto,node distance=1.8cm,
                    semithick, every node/.style={scale=0.6}]
  \tikzstyle{every state}=[fill=white,draw=black,text=black]

  \node[initial,state] (A)                    {${q_1}_3$};
  \node[state,accepting](B)[right of=A]					   {${q_2}_3$};

  \path (A)edge             node {$\overline{c}$} (B);
\end{tikzpicture}  \quad
\begin{tikzpicture}[->,>=stealth',shorten >=1pt,auto,node distance=2.3cm,
                    semithick, every node/.style={scale=0.6}]
  \tikzstyle{every state}=[fill=white,draw=black,text=black]

  \node[initial,state] 				 (A)                    {$\vec{q_1}$};
  \node[state] 				         (B) [right of=A]   {$\vec{q_2}$};
  \node[state]						 (C) [right of=B]   {$\vec{q_3}$};
  \node[state]         				 (D) [below of=A] {$\vec{q_4}$};	
  \node[state] 				         (E) [right of=D]   {$\vec{q_5}$};
  \node[state,accepting]			 (F) [right of=E] {$\vec{q_6}$};

  \path (A)  edge             				   node{$(b,\overline b,\blk)$} (B)
			    edge             				   node[left]{$(\blk,c,\overline{c})$} (D)
		  (B)  
				edge         				       node {$(\overline{a},a,\blk)$} (C)
				edge         				       node[left] {$(\blk,c,\overline{c})$} (E)
		  (C)  edge                             node {$(\blk,c,\overline{c})$} (F)
				
		  (D)  
				edge                             node{$(b,\overline b,\blk)$} (E)
		  (E) 
				edge              	     	   node {$(\overline{a},a,\blk)$} (F);
\end{tikzpicture} \quad
	\caption{The contract automata of Example~\ref{ex:illmix}. 
	Top from left to right: $\llbracket Alice \rrbracket, \llbracket Bob \rrbracket$. 
 Bottom from left to right: $\llbracket Charlie \rrbracket,\mathcal{K}_{\llbracket Alice  \rrbracket \boxtimes \llbracket Bob \rrbracket \boxtimes \llbracket  Charlie \rrbracket}$.
	}
\label{fig:exampleLinear}
\end{figure}
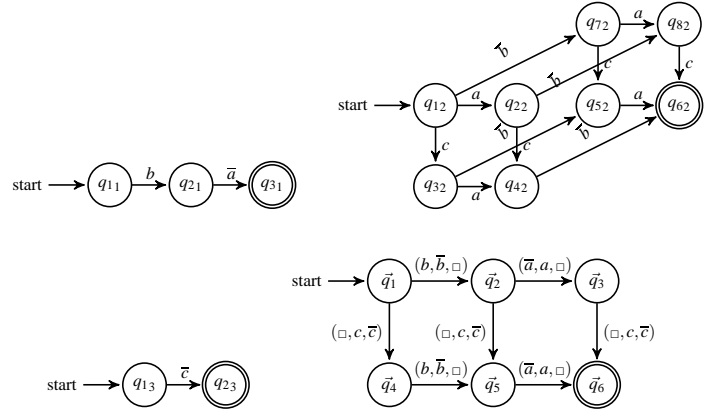

We now recall the basics of \illmix.
Let $\A,\Abot$ be respectively the set of \emph{positive} and \emph{negative atoms}, ranged over by 
$a,b,c,\ldots \in \A$ and by $a^\bot,b^\bot,c^\bot,\cdots \in \Abot$. Let $\La = \A \cup \Abot$ be the set
of \emph{literals}, and assume $Y \subseteq \A, X \subseteq \La$, where $X$ does not contain any atom $a$ and its negation $a^\bot$, according to Definition~\ref{def:contract} 
(recall that a principal automaton is such that $A^r \cap co(A^o) = \emptyset$).
A \emph{positive} tensor product is a tensor product of positive atoms. 

As said, we only consider a fragment of Horn $ILL^{mix}$ called H-\illmix, defined below.  
It only has tensor products and \emph{Horn implications}:
$\bigotimes_{b \in Y} b \multimap \bigotimes_{a \in X} a$. 
Note that the premises of the Horn implications
are always positive tensor products, and the conclusions are tensor products of literals, possibly negative.

Since the treatment for non-linear implications of \illmix \ is similar to that presented in Section~\ref{sect:1NPCL}, we feel free to only deal below with linear implications and tensor products of literals.

\begin{defi}[H-$ILL^{mix}$]\label{def:hornillmix}
The Horn formulae $p, p_i, ...$ and the clauses $\alpha, \alpha_i, ... $ of  \emph{H-$ILL^{mix}$} are defined by
\[
	p ::= \bigotimes_{i \in I} \ \alpha_i \qquad 
	\alpha ::= \bigotimes_{a \in X} a \ \vert \   \bigotimes_{b \in Y} b \multimap \bigotimes_{a \in X} a 
\]
where $|I|\geq 2$; $|X|,|Y| \geq 1$; $\{ a,  a^\bot\} \not \subseteq X$; and $b \in Y$ implies $b \not \in X$.
\end{defi}

The subset of the rules of the sequent calculus of \illmix\ relevant to our treatment is in Figure~\ref{fig:sequent_illmix}, where $A,B$ stand for a Horn formula $p$ or clause $\alpha$, while  $\gamma$ may also be empty (note that 
in rule $(Neg L)$, $A = a$ and so $A^\bot= a^\bot$);
$\Gamma$ and $\Gamma'$ stand for multi-sets containing  Horn formulae or clauses; and
$\Gamma, \Gamma'$ is the multi-set union of $\Gamma$ and $\Gamma'$, assuming $\Gamma, \emptyset =\Gamma$.
The complete set of rules for \illmix\ is in~\cite{benton1995mixed}, and can be found in the appendix.

\begin{figure}[tb]
\[
\irule{}{A \vdash A}{  \ \ Ax}
\qquad
\irule{\Gamma \vdash \ \ \Gamma' \vdash \gamma}{ \Gamma, \Gamma' \vdash \gamma}{ \ \ Mix}
\qquad 
\irule{\Gamma \vdash A}{\Gamma,A^\bot \vdash}{ \ \ Neg L }
\] \\
\[
\irule{\Gamma, A, B \vdash \gamma}{\Gamma, A \otimes B \vdash \gamma}{ \ \ \otimes L }
\qquad
\irule{\Gamma \vdash A \quad \Gamma' \vdash B}{\Gamma, \Gamma' \vdash A \otimes B}{ \ \ \otimes R }
\]			 \\
\[
\irule{\Gamma \vdash A \quad \Gamma',B \vdash \gamma}{ \Gamma, \Gamma', A \multimap B \vdash \gamma}{ \ \ \multimap  L }
\qquad
\irule{\Gamma, A \vdash B}{ \Gamma \vdash A \multimap B}{ \ \ \multimap R }
\]
\caption{A subset of the rules of the sequent calculus of $ILL^{mix}$.}
\label{fig:sequent_illmix}
\end{figure}

The following auxiliary definition of the concatenation of two automata helps to translate a H-\illmix\ formula.

\begin{defi}[Concatenation of CA]
Given two principal contract automata \\ $\mathcal A^1 = \langle  Q^1, q_0^1, A^{r^1}, A^{o^1}, T^1,F^1\rangle$ 
and $\mathcal A^2 = \langle Q^2, q_0^2, A^{r^2}, A^{o^2}, T^2,F^2\rangle$,
 their \emph{concatenation} is
\[
\begin{array}{ll}
\mathcal A^1 \cdot \mathcal A^2 = & \langle  Q^1 \cup  Q^2, q_0^1, A^{r1} \cup A^{r2}, A^{o1} \cup A^{o2}, \\
& (T^1  \setminus \{( q,  a,  q') \in T^1 \mid q' \in F^1  \} ) \cup T^2  \\
&  \qquad \quad \   \cup \{(q, a,  q_0^2) \mid (q, a, q') \in T^1, q' \in F^1  \} , F^2 \rangle
\end{array}
\]
\end{defi}

Concatenation is almost standard, with the proviso that we replace every transition of  $\mathcal A^1$ leading to a final state with a transition with the same label leading to the initial state of $\mathcal A^2$.
Note also that loops can be ignored, because the automata obtained by the translation in Definition~\ref{def:translation} below have no cycles.

Similarly to what has been done in the previous sub-section, a tensor product is rendered as all the possible orders in which the automaton can fire (the actions corresponding to) its literals.
If the literal is a positive atom, then it becomes an offer, while it originates a request if the atom is negative.
A linear implication is rendered as the concatenation of the automaton coming from the premise, and that of the conclusion, with the following proviso.
In the premise all the atoms are positive, but they are \emph{all} rendered as \emph{requests} (i.e. as negative atoms), and shuffled.
The states are in correspondence with the atoms still to be fired and $\{*\}$ stands for the (final) state 
where all atoms have been fired.

\begin{defi}[Translation of H-\illmix]
\label{def:translation} 
Given a set of atoms $X$, let $P = \{ q \cup \{ \ast \} \mid q \in 2^X \}$ with typical element $Z$.
The translation of a H-\illmix\ formula $p$  into a contract automata $\llbracket p \rrbracket$ 
is inductively defined by the following rules:\\ \\
$
\llbracket \bigotimes_{i \in I} \alpha_i \rrbracket =  \boxtimes_{i \in I} \llbracket \alpha_i \rrbracket 
$\\
\begin{flalign}
\llbracket \bigotimes_{a \in X} a  \rrbracket = &\langle  P,
	X \cup \{\ast\}, 
	 \{ a \mid a^\bot \in X \cap \Abot\},\{\overline a \mid a \in X \cap \A\}, && \notag  \\
 & \{  ( Z \cup \{ a^\bot\}, a, Z) \mid  Z \cup \{a^\bot\} \in P, a^\bot \in X\} \cup && \notag  \\
 & \{  ( Z \cup \{a\}, \overline a, Z) \mid  Z \cup \{a\} \in P, a \in X\}, && \notag  \\
 & \{  \{ \ast \} \} \rangle&&\notag
\end{flalign}
\\
$\llbracket \bigotimes_{b \in Y} b \multimap \bigotimes_{a \in X} a \rrbracket =
   \llbracket \bigotimes_{b \in Y} b^\bot  \rrbracket \cdot \llbracket \bigotimes_{a \in X} a  \rrbracket 
$\bigskip

Moreover, we homomorphically translate multi-sets of Horn formulae and clauses as follows:
\[
\llbracket p, \Gamma \rrbracket = \llbracket p \rrbracket \boxtimes \llbracket \Gamma \rrbracket \qquad \qquad
\llbracket \alpha, \Gamma \rrbracket = \llbracket \alpha \rrbracket \boxtimes \llbracket \Gamma \rrbracket \\
\]
\end{defi}

The automata obtained by translating the formulae representing the contracts of Alice, Bob and Charlie in Example~\ref{ex:illmix} are in Figure~\ref{fig:exampleLinear}.

\begin{defi}
A sequent $\Gamma \vdash Z$ is \emph{honoured} if and only if it is provable and $Z$ is a positive tensor product or empty.
\end{defi}

Intuitively, honoured sequents can be proved and additionally they have no negative atoms, i.e.\ no debts.
The main result of this section is that a sequent $\Gamma \vdash Z$ is honoured
if and only if the corresponding contract automaton $\llbracket \Gamma \rrbracket$ admits agreement.
An important outcome is the possibility of expressing each H-\illmix\ formula as a contract automaton $\mathcal A$, so to use our verification techniques. 
 It is then possible to compose several H-\illmix\ formulae through the composition operators of contract automata, exploiting   compositionality and the related results (for example Theorem~\ref{the:composition}) for efficiently checking the provability of formulae in H-\illmix. In the statement below and in the proofs in the appendix, we say that  $\llbracket \Gamma \rrbracket$ admits agreement on $Z$ whenever there exists a trace in $\mathscr{L} (\llbracket \Gamma \rrbracket)$ only made of match actions and offers in correspondence with the literals in $Z$.

\begin{restatable}{thm}{theautomaillmix}\label{the:automaillmix}
Given a multi-set of Horn formulae $\Gamma$, we have that
\[
\Gamma \vdash Z \text{ is an honoured sequent if and only if } 
\llbracket \Gamma \rrbracket \text{ admits agreement on } Z
\]
\end{restatable}

Through this result we have linked the problem of verifying the correctness of a composition of services 
to the generation of a deduction tree that proves a H-\illmix\ formula.
Moreover, we have shown that the possibility of recording debts in H-\illmix\  solves circularity issues arising from a composition of services.

\section{An example}
\label{sect:casestudy}

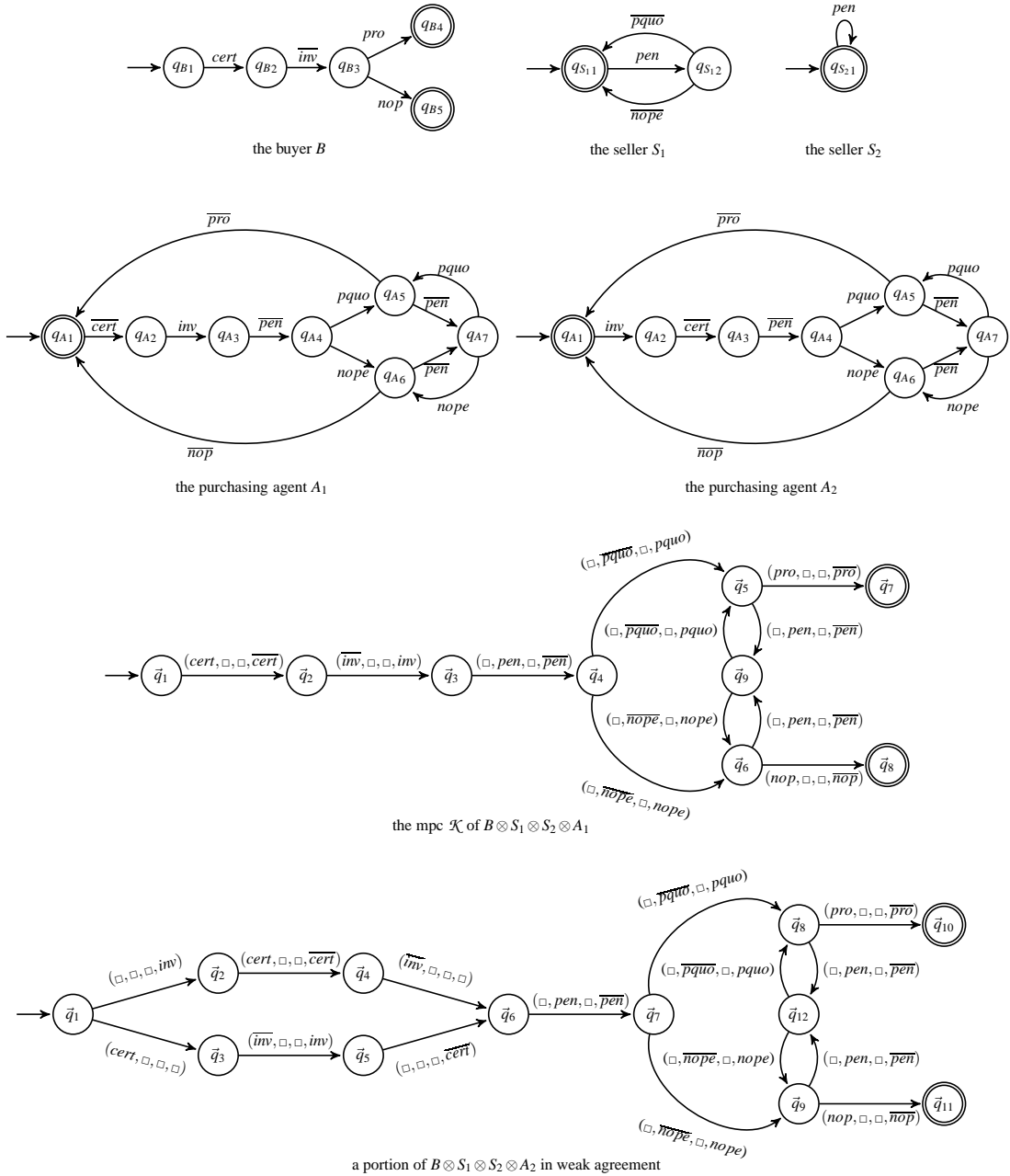
\begin{figure}[tb]
\center
\begin{tikzpicture}[->,>=stealth',shorten >=1pt,auto,node distance=2.0cm,
                    semithick, every node/.style={scale=0.6},initial text={}]
  \tikzstyle{every state}=[fill=white,draw=black,text=black]

  \node[initial,state] (A)                    {${q_B}_1$};
  \node[state] (B)[right of=A]                    {${q_B}_2$};
  \node[state] (C)[right of=B]                    {${q_B}_3$};
  \node[state,accepting]         (D) [right of=C, above=0.3cm] {${q_B}_4$};
  \node[state,accepting]         (E) [right of=C, below=0.3cm] {${q_B}_5$};

  \path (A)			edge             node{$cert$} (B)
	     (B)			edge             node{$\overline{inv}$} (C)
 	     (C)		     edge             node{$pro$} (D)
 	     (C)		     edge             node[below=0.1cm] {$nop$} (E);
\node [below=1cm,  scale=1] at (B)
        {
        \hspace{1cm} the buyer $B$
        };
\end{tikzpicture}\hspace{15pt}
\begin{tikzpicture}[->,>=stealth',shorten >=1pt,auto,node distance=3.0cm,
                    semithick, every node/.style={scale=0.6},initial text={}]
  \tikzstyle{every state}=[fill=white,draw=black,text=black]

  \node[initial,state,accepting] (A)                    {${q_{S_1}}_1$};
  \node[state] (B)[right of=A]                    {${q_{S_1}}_2$};

  \path (A)			edge             node{$pen$} (B)
 	     (B)		     edge   [bend left=45]          node[below]{$\overline{nope}$} (A)
 	     (B)		     edge   [bend right=45]      node[above] {$\overline{pquo}$} (A);
\node [below=1cm,  scale=1] at (A)
        {
        \hspace{2cm} the seller $S_1$
        };
\end{tikzpicture}\hspace{15pt}
\begin{tikzpicture}[->,>=stealth',shorten >=1pt,auto,node distance=2.0cm,
                    semithick, every node/.style={scale=0.6},initial text={}]
  \tikzstyle{every state}=[fill=white,draw=black,text=black]

  \node[initial,state,accepting] (A)                    {${q_{S_2}}_1$};

  \path (A)			edge[loop above]             node{$pen$} (A);
\node [below=1cm,  scale=1] at (A)
        {
        the seller $S_2$
        };
\end{tikzpicture}\\
\vspace{0.5cm}
%
%
\begin{tikzpicture}[->,>=stealth',shorten >=1pt,auto,node distance=2.0cm,
                    semithick, every node/.style={scale=0.6},initial text={}]
  \tikzstyle{every state}=[fill=white,draw=black,text=black]

  \node[initial,state,accepting] (A)                    {${q_A}_1$};
  \node[state] (B)[right of=A]                    {${q_A}_2$};
  \node[state] (C)[right of=B]                    {${q_A}_3$};
 \node[state] (D)   [right of=C]                			 {${q_A}_4$};
  \node[state] (E)[right of=D, above=0.3cm]                    {${q_A}_5$};
\node[state](G)[right of=D, below=0.3cm]				{${q_A}_6$};
\node[state](F)[right of=E, below=0.3cm]				{${q_A}_7$};

  \path (A)			edge             node{$\overline {cert}$} (B)
	     (B)			edge             node{$inv$} (C)
 	     (C)		     edge             node{$\overline{pen}$} (D)
	     (D)		     edge             node[above=0.1cm]{$pquo$} (E)
 	     (D)		     edge        		node[below=0.1cm] {$nope$} (G)
		(G)			edge   [bend left=50]          node{$\overline{nop}$} (A)		
		(E)			edge   [bend right=50]          node[above]{$\overline{pro}$} (A)
		(E)			edge          node[above]{$\overline{pen}$} (F)
	     (F)		     edge   [bend right=65]          node[above]{$pquo$} (E)
 	     (F)		     edge   [bend left=65]      node[below=0.1cm] {$nope$} (G)
		(G)			edge          node[below]{$\overline{pen}$} (F);
\node [below=2cm,  scale=1] at (C)
        {
        \hspace{1cm} the purchasing agent $A_1$
        };
\end{tikzpicture}
%
%
\begin{tikzpicture}[->,>=stealth',shorten >=1pt,auto,node distance=2.0cm,
                    semithick, every node/.style={scale=0.6},initial text={}]
  \tikzstyle{every state}=[fill=white,draw=black,text=black]

  \node[initial,state,accepting] (A)                    {${q_A}_1$};
  \node[state] (B)[right of=A]                    {${q_A}_2$};
  \node[state] (C)[right of=B]                    {${q_A}_3$};
 \node[state] (D)   [right of=C]                			 {${q_A}_4$};
  \node[state] (E)[right of=D, above=0.3cm]                    {${q_A}_5$};
\node[state](G)[right of=D, below=0.3cm]				{${q_A}_6$};
\node[state](F)[right of=E, below=0.3cm]				{${q_A}_7$};

  \path (A)			edge             node{$inv$} (B)
	     (B)			edge             node{$\overline{cert}$} (C)
 	     (C)		     edge             node{$\overline{pen}$} (D)
	     (D)		     edge             node[above=0.1cm]{$pquo$} (E)
 	     (D)		     edge        		node[below=0.1cm] {$nope$} (G)
		(G)			edge   [bend left=50]          node{$\overline{nop}$} (A)		
		(E)			edge   [bend right=50]          node[above]{$\overline{pro}$} (A)
		(E)			edge          node[above]{$\overline{pen}$} (F)
	     (F)		     edge   [bend right=65]          node[above]{$pquo$} (E)
 	     (F)		     edge   [bend left=65]      node[below=0.1cm] {$nope$} (G)
		(G)			edge          node[below]{$\overline{pen}$} (F);
\node [below=2cm,  scale=1] at (C)
        {
        \hspace{1cm} the purchasing agent $A_2$
        };
\end{tikzpicture}\\
\vspace{.3cm}
%
%
\begin{tikzpicture}[->,>=stealth',shorten >=1pt,auto,node distance=3.5cm,
                    semithick, every node/.style={scale=0.6}, initial text={}]
  \tikzstyle{every state}=[fill=white,draw=black,text=black]

  \node[initial,state] (A)                    {$\vec q_1$};
  \node[state] (B)[right of=A]                    {$\vec q_2$};
  \node[state] (C)[right of=B]                    {$\vec q_3$};
 \node[state] (D)   [right of=C]                			 {$\vec q_4$};
  \node[state] (E)[right of=D, above=1cm]                    {$\vec q_5$};
\node[state](G)[right of=D, below=1cm]				{$\vec q_6$};
\node[state,accepting](F)[right of=E]						{$\vec q_7$};
\node[state,accepting](H)[right of=G]						{$\vec q_8$};
\node[state](I)[right of=D]		{$\vec q_9$};

  \path (A)			edge             node{$(cert,\blk,\blk,\overline{cert})$} (B)
	     (B)			edge             node{$(\overline{inv}, \blk, \blk, inv)$} (C)
 	     (C)		     edge             node{$(\blk,pen,\blk,\overline{pen})$} (D)
	     (D)		     edge [bend left=70]            node[above]{\hspace{-1pt}\rotatebox{15} {($\blk,\overline{pquo},\blk,pquo$)}} (E)
 	     (D)		     edge   [bend right=70]     		node[below]{\hspace{-1pt}\rotatebox{-15}{($\blk, \overline{nope},\blk,nope$)}} (G)
		(G)			edge           node[below]{$(nop,\blk,\blk,\overline{nop})$} (H)		
		(E)			edge        node{$(pro,\blk,\blk,\overline{pro})$} (F)
		(E)			edge [bend left]      node{$(\blk,pen,\blk,\overline{pen})$} (I)
		(G)			edge[bend right]       node[right]{$(\blk,pen,\blk,\overline{pen})$} (I)
		(I)			edge [bend left]      node {($\blk,\overline{pquo},\blk,pquo$)} (E)
		(I)			edge [bend right]      node[left]{($\blk, \overline{nope},\blk,nope$)} (G);
\node [below=2cm,  scale=1] at (C)
        {
        \hspace{1.8cm} the mpc $\mathcal K$ of $B\otimes S_1 \otimes S_2 \otimes A_1$
        };
\end{tikzpicture}\\
\vspace{0.3cm}
%
%
\begin{tikzpicture}[->,>=stealth',shorten >=1pt,auto,node distance=3.5cm,
                    semithick, every node/.style={scale=0.6}, initial text={}]
  \tikzstyle{every state}=[fill=white,draw=black,text=black]

  \node[initial,state] (A)                    {$\vec q_1$};
  \node[state] (B)[right of=A,above=0.3cm]                    {$\vec q_2$};

\node[state](I)[right of=A,below=0.3cm]		{$\vec q_3$};
\node[state](L)[right of=B]						{$\vec q_{4}$};
\node[state](M)[right of=I]					{$\vec q_{5}$};
  \node[state] (C)[right of=L,below=0.3cm]                    {$\vec q_6$};
 \node[state] (D)   [right of=C]                			 {$\vec q_7$};
  \node[state] (E)[right of=D, above=1cm]                    {$\vec q_8$};
\node[state](G)[right of=D, below=1cm]				{$\vec q_9$};
\node[state,accepting](F)[right of=E]						{$\vec q_{10}$};
\node[state,accepting](H)[right of=G]						{$\vec q_{11}$};
\node[state](N)[right of=D]		{$\vec q_{12}$};

  \path (A)			edge             node[above]{\rotatebox{15}{$(\blk,\blk,\blk,inv)$}} (B)
		 (A)			edge             node[below]{\rotatebox{-15}{$(cert,\blk,\blk,\blk)$}} (I)
	     (B)			edge             node{$(cert,\blk,\blk,\overline{cert})$} (L)
 	 (L)	edge	node[above]{\rotatebox{-15}{$(\overline{inv}, \blk, \blk, \blk)$}}(C)
	(I)	edge	node{$(\overline{inv},\blk,\blk,inv)$}(M)
        (M) edge	node[below]{\rotatebox{15}{$(\blk, \blk, \blk, \overline{cert})$}}(C)
 	     (C)		     edge             node{$(\blk,pen,\blk,\overline{pen})$} (D)
	     (D)		     edge [bend left=70]            node[above]{\hspace{-1pt}\rotatebox{15} {($\blk,\overline{pquo},\blk,pquo$)}} (E)
 	     (D)		     edge   [bend right=70]     		node[below]{\hspace{-1pt}\rotatebox{-15}{($\blk, \overline{nope},\blk,nope$)}} (G)
		(G)			edge           node[below]{$(nop,\blk,\blk,\overline{nop})$} (H)		
		(E)			edge        node{$(pro,\blk,\blk,\overline{pro})$} (F)
		(E)			edge [bend left]      node{$(\blk,pen,\blk,\overline{pen})$} (N)
		(G)			edge[bend right]       node[right]{$(\blk,pen,\blk,\overline{pen})$} (N)
		(N)			edge [bend left]      node {($\blk,\overline{pquo},\blk,pquo$)} (E)
		(N)			edge [bend right]      node[left]{($\blk, \overline{nope},\blk,nope$)} (G);
\node [below=2.0cm, align=center, scale=1] at (C)
        {
          a portion of $B\otimes S_1 \otimes S_2 \otimes A_2$ in weak agreement
        }; 
\end{tikzpicture}\hspace{13pt}
\caption{The contract automata for the example}
\label{fig:casestudy}
\end{figure}

In this section we consider a well-known case study taken from~\cite{pel03}.
This is a purchasing system scenario, where a manufacturer (the buyer)
wants to build a product.
To configure it, the buyer lists in an inventory the needed components and contacts a purchasing agent.
The agent looks for suppliers of these components, and eventually sends back to the buyer its proposal, if any.
A supplier is assumed to signal whether it can fulfil a request or not; if neither may happen, the interactions between it and the purchasing agent are rolled back, so as to guarantee the transactional integrity of the overall process.
A description of the WSDL of the services, as well as 
the BPEL process from the purchasing agent's perspective are in~\cite{pel03}, where
the transactional integrity is maintained using the tags
 \texttt{<faultHandlers>} and \texttt{<scope>} of BPEL.
 
We slightly modify the original protocol, where the purchasing agent guarantees its identity to the buyer through a public-key certificate.
For brevity, here we assume to have two sellers $S_1$ and $S_2$, and two purchasing agents $A_1$ and $A_2$, that behave differently.
A service instance involves the buyer, an agent and both sellers.
The buyer $B$ requires the certificate of an agent (action $cert$), then it 
offers the inventory requirements ($\overline{inv}$). 
Finally, it terminates by receiving either a proposal ($pro$) or a negative message ($nop$), if no proposal can be formulated.
The seller $S_1$ waits for a request ($pen$) of a component from an agent. 
It then replies by offering
a quote for that part ($\overline{pquo}$), or a negative message ($\overline{nope}$) if it is unavailable, and restarts.
The second seller $S_2$ always accepts a request, but never replies.
The first agent $A_1$ offers its certificate ($\overline{cert}$), then requires the inventory list ($inv$).
It then sends a request to and waits for a reply from the sellers. 
The agent must communicate at least with one supplier before replying to the buyer, and it can 
span over all the available suppliers in the network, unknown a priori, before compiling its proposal.
Finally, it sends to the buyer a proposal ($\overline{prop}$), or the negative message ($\overline{nop}$). 
The second agent $A_2$ behaves similarly to $A_1$, except the first two actions are exchanged: before sending its certificate to $B$, it first requires the inventory list.

In Figure~\ref{fig:casestudy} from top to bottom, we display, from left to right, the automata $B, S_1$ and $S_2$; the automata  $A_1$ and $A_2$; then the most permissive controller $\mathcal K$ of  $B \otimes S_1 \otimes S_2 \otimes A_1$ (the whole composition is omitted to save space); finally a portion of $B \otimes S_1 \otimes S_2 \otimes A_2$ in weak agreement.
This example shows that through contract automata one can identify which traces reach success, and which a failure, together with those principals responsible for diverging from the behaviour in agreement, as well as to single out which failures depend on the order of actions, and which not.
Indeed, by inspecting $\mathcal K$, that of course is safe, one can notice that $A_1$ never interacts with $S_2$ because it never replies and so it is recognised liable.
As a matter of fact, the composed automaton  $B \otimes S_1 \otimes S_2 \otimes A_1$ admits agreement, but it is not safe.
Note that $\mathcal K$ blocks every communication with $S_2$, so enforcing transactional integrity, 
because $\mathcal{K}$ removes all possibilities of rollbacks from a trace not in agreement.
The composed automaton $B \otimes S_1 \otimes S_2 \otimes A_2$ admits weak agreement but not agreement (and its most permissive controller is empty), because $B$ and $A_2$ fail in exchanging the certificate and the inventory requirements, as both are stuck  waiting for the fulfilment of their requests.
However,  by abstracting away the order in which actions are performed, circularity is no longer a problem, and these requests satisfied.
Note that $S_2$ is detected to be also weakly liable.

\section{Related Work}
\label{sect:conclusion}

Contract automata are similar to I/O~\cite{LynchT89} and Interface Automata~\cite{AlfaroH01}, introduced in the field of Component Based Software Engineering.
A first difference is that principal contract automata have no internal transitions, and that our operators of composition track each principal, to find the possible liable ones.  
Also we do not allow input enabled operations and non-linear behaviour (i.e.\ broadcasting offers to every possible request), and our notion of agreement is dual to that of
compatibility in~\cite{AlfaroH01}, that requires all the \emph{offers} to be matched.

We now relate our approach to the growing body of work in the literature introduced  to describe and analyse service contracts.

\subsection*{Behavioural contracts}
In~\cite{Bordeaux2005} the behaviour of web-services is described through automata, equivalent to our principal contract automata.
However, only bi-party interactions are  considered, i.e.\ interactions between a single client and a single server, while our model deals with multi-party interactions through orchestration. 
Different notions of compliance are introduced, and one of them is close to our notion of agreement.
In~\cite{Cast2009ACM} behavioural contracts  are expressed in CCS and the interactions between services are modelled via I/O actions. 
The main focus of this work is on formalising the notion of progress of interactions. 
Two different choice operators, namely internal and external, describe how two services interact. 
The internal choice requires the other party to be able to synchronise with all the possible branches of the first, while for the 
external choice it suffices to synchronise with at least one branch. 
A client and a server are compliant
if their interactions never get stuck.
This approach is extended to a multi-party version by extending the $\pi$-calculus in~\cite{Cast2009CONCUR} with the above notions of non-deterministic choice. 
Our model represents internal/external choice as a branching of requests/offers, and it is intrinsically multi-party.
Also, we consider stronger properties than theirs: progress guarantees that a subset of contracts meets their requests, while (weak) agreement requires that all of them do, i.e.\ that each principal reaches a successful state. 
We also consider (weak) liability of principals, and conditions under which (weak) safety is preserved by composition (collaborative and competitive).
A CCS-like process calculus, called BPEL \emph{abstract activities} is used in~\cite{LaneveP2015} to represent BPEL  activities~\cite{bpel}, for connecting BPEL processes with contracts in~\cite{Cast2009ACM}.
The calculus is endowed with a notion of compliance and sub-contract relation (see below). 
Contract automata and this formalism are very close, e.g.\ both are finite state, so it would not be difficult to formally relate them.

In \cite{Padovani10} the approach of~\cite{Cast2009ACM} is extended by exploiting an orchestrator for managing the 
\emph{sub-contract} relation. A contract $\sigma_1$ is sub-contract of $\sigma_2$ if $\sigma_1$  is more deterministic or allows more interactions or  is a permutation of  the same channels of $\sigma_2$. 
However, it is not always the case that  a contract $\sigma$, compliant with $\sigma_1$, is also 
compliant with $\sigma_2$. 
A technique for synthesising an orchestrator is presented to enforce compliance of contracts
under the sub-contract relation.
This approach is further extended in~\cite{Barbanera2015}, where an orchestrator is synthesised from \emph{session contracts}, where actions in a branching can only be all inputs or outputs. 
Only bi-party contracts are considered, and synthesis is decidable even in the presence of messages never delivered to the receiver (orphan messages).
Two notions of compliance are studied: respectful and disrespectful. 
In the first, orphan messages and circularities are ruled out by the orchestrator, while in the second they are allowed.
Our notion of weak agreement is close to the orchestrator of~\cite{Padovani10,Barbanera2015} in the case of \emph{disrespectful compliance}. 

In~\cite{BDLL15} the contracts of~\cite{Cast2009ACM}  are enriched with a mechanism for recovering from a stuck computation. 
The external choices are called \emph{retractable}, and a client contract $\overline a + \overline b$ is compliant with a server $a$ since, in case the client decides to send $b$, it can retract the choice and perform the correct operation $\overline a$.
In our work, the controller for the case of agreement cuts all the paths which may lead one principal to perform a retract. 
Hence, a controlled interaction of services needs not to roll back, as the orchestrator \emph{prevents} firing of liable transitions. 
This means that, if a composition of contracts is safe then the contracts are compliant according to~\cite{BDLL15}. 
The converse does not hold. Indeed, our notion of agreement is stronger, as we force an interaction of services to reach a successful state. 

The compliance relations studied 
in~\cite{Cast2009ACM,Cast2009CONCUR,LaneveP2015,Padovani10,Barbanera2015,BDLL15} 
are mainly inspired by testing equivalence~\cite{deNicolaHennessy}: 
a CCS process (in our case the service) is tested against an observer (the client), in two different ways.
A service \emph{may-satisfy} a client if there exists a computation that ends in a successful
state, and a service \emph{must-satisfy} a client if in every maximal trace (an infinite trace or a trace that can not be prolonged) the client can terminate successfully. 
We conjecture that may-test corresponds to the notion of \emph{strong agreement} of~\cite{basile2014,basile2015} (there exists a trace only composed of matches),
while must-test implies \emph{strong safety} (all traces are in strong agreement), but not vice-versa.
For example the service $\overline a^*. \overline b $ does not must-satisfy the client $a^*.b$, but their product is strongly safe (if unfair, the service may never offer $\overline b$ to its client). 
Actually, strong safety is alike \emph{should testing} of~\cite{Rensink2007}, where 
the divergent computations are ruled out.

\subsection*{Session types and choreographies}
Session types have been introduced to reason over the behaviour of communicating processes, and are used for 
typing channel names by structured sequences of types~\cite{Dezani2009}. 
Session types can be global or local.
A \emph{global type} represents a formal specification of a choreography of services in terms of their interactions. 
The projection of a safe global type to its components yields a safe \emph{local type}, which is a term of a process algebra similar to those of~\cite{Cast2009ACM}. 
Conversely, from safe local types it is possible to synthesise a choreography as a safe global type~\cite{LangeT12,LangeTY15}. 
In~\cite{Bernardi2012} the contracts of~\cite{Cast2009ACM} are shown to be a model of first-order session types~\cite{Simon2005}.
This approach is then extended in~\cite{Bernardi2014} by introducing a notion of higher-order contracts and relating them to higher-order session types, that also handle session delegation.

Although the above approaches and ours seem unrelated, one can compare them by resorting to communicating finite states machines~\cite{BrandZ83}, that are finite state automata similar to ours,
to which local types are proved to correspond~\cite{DenielouY13}.
These automata interact through FIFO buffers, hence a principal can receive an input only if it was previously enqueued, and in this they differ from contract automata, where offers and requests can match or even fire unmatched in any order.
However, under mild conditions, the two classes of automata are equivalent~\cite{basile2014,basile2015}, so establishing a first bridge between the choreography model based on session types and our automata model of orchestration.

Many properties of communicating finite state machines, as compliance in the asynchronous case, are not decidable in general~\cite{BrandZ83}, but some become such by using FIFO queues and bags~\cite{ClementeHS14}. 
Moreover in~\cite{LangeTY15} compliance between 
communicating finite state machines is guaranteed whenever it is possible to synthesise a global choreography from them.
It would be interesting to describe compliance of~\cite{BrandZ83} in terms of flow control, as done
for weak agreement, and to study a relaxation of the linear problem  which makes the problem decidable.

In~\cite{LaneveP2015} the compliance and sub-contract relations are extended to deal with choreographies. 
Compliance is obtained by seeing a choreography as a compound service, similarly to our composed contract automata.
Since a client cannot interact with the choreography on actions already used while synchronising by other services, in order to obtain compliance the client must be \emph{non-competitive} with the other services.

\subsection*{$\lambda$-calculus, logics,  event-structures}
Services are represented in~\cite{BartolettiDFZ11,BartolettiDF09} by $\lambda$-expressions, and safety policies are imposed over their interactions. 
A type and effect system is used to compute the types of the services and their abstract behaviours, that are then model checked at static time to guarantee that the required policies are always satisfied. 
A main result shows how to construct a plan that associates requests with offers so to guarantee that no executions will violate the security requirements.  
In~\cite{BasileDF13,BasileDF14journal} these techniques have been applied to an automata based representation of the contracts of~\cite{Cast2009ACM}, recovering the same notion of progress.


Propositional Contract Logic~\cite{BartolettiZ10} and Intuitionistic Linear Logic with Mix~\cite{benton1995mixed}  have been already discussed in Section~\ref{sect:logic}.

Processes and contracts are two separate entities in~\cite{BartSci2013}, unlike ours.
In this formalism contracts are represented as formulae
or as process algebras.
A process can fulfil its duty by obeying its contract or it behaves dishonestly and becomes \emph{culpable} --- and redeems by performing later on the prescribed actions.
Also our principals can be at fault, but our notion of liability slightly differs from culpability, mainly because we do not admit the possibility of redeeming.

Contracts are represented in~\cite{Bartoletti2015} through Event Structures endowed with certain notions from Game Theory. 
An agreement property is proposed, ensuring safe interactions among participants, that is similar to ours under an eager strategy.
A principal is culpable if it has not yet fired an
enabled event, it is otherwise innocent. 
In particular a principal agrees to a contract 
if it has a positive pay-off in case all the principals are innocent, or if someone else is found culpable.
Additionally the authors study protection: a protected principal has a non-losing strategy in every possible
context, but this is not always possible.
Finally two encodings from session types to Event Structures are proposed, and compliance between bi-party
session types is shown to correspond to agreement of the corresponding event structures via an eager strategy.

\section{Concluding Remarks}
\label{sect:concludingremarks}
%

We have studied contract composition for services, focussing on orchestration.
Services are formally represented by a novel class of finite state automata, called contract automata.
They have two operators that compose services according to two different notions of orchestrations:
one when a principal joins an existing orchestration with no need of a global reconfiguration, and the other when a global adaptive re-orchestration is required. 
We have defined notions that illustrate when a composition of contracts behaves well, roughly when all the requests are fulfilled.
These properties have been formalised as agreement and safety, and have been studied both in the case when requests are satisfied synchronously and asynchronously. 
Furthermore, a notion of liability has been put forward. 
A liable principal is a service leading the contract composition into a fail state.
Key results of the paper are ways to enforce good behaviour of services.
For the synchronous versions of agreement and safety, we have applied techniques from Control Theory, while for the asynchronous versions we have taken advantage of Linear Programming techniques borrowed from Operational Research.
Using them, we efficiently find the optimal solutions of the flow in the network automatically derived from contract automata. 

We have also investigated the relationships between our contract automata and two intuitionistic logics, particularly relevant for their ability in describing the potential, but harmless and often essential circularity occurring in services.
We have considered a fragment of the Propositional Contract Logic~\cite{BartolettiZ10,BartolettiZ09} particularly suited to describe contracts, and we relate it through a translation of its formulas into contract automata.
Similarly, we have examined certain sequents of the Intuitionistic Linear Logic with Mix that naturally represent contracts in which all requests are satisfied.
Then we have proved that these sequents are provable if and only if a suitable translation of them as contract automata admits agreement.

A main advantage of our framework is that it supports the development of automatic verification tools for checking
and verifying properties of contract composition.
In particular, the formal treatment of contract composition in terms of optimal solutions of network flows paves the way of exploiting efficient optimisation algorithms. 
We have developed a prototypical verification tool~\cite{cat}, available at \url{https://github.com/davidebasile/workspace}.

\section*{Acknowledgements}
We are deeply indebted with Giancarlo Bigi for many discussions and
suggestions on the usage of optimisation techniques, with Massimo
Bartoletti for many insights on the logical aspects of our proposal,
and with the anonymous referees for their valuable comments and
remarks.


\bibliographystyle{splncs03}
\bibliography{bib}

\newpage
\section{Appendix}
\subsection{The Model}

%

\proassociative*
\begin{proof}
Example~\ref{ex:non-assoc} suffices to prove the first statement.
For the second statement one has 
 $\mathcal A = (\mathcal{A}_1 \boxtimes \mathcal{A}_2) \boxtimes \mathcal{A}_3 = 
  \bigotimes_{ \mathcal A_i \in I} \mathcal{A}_i =  \mathcal{A}_1 \boxtimes (\mathcal{A}_2 \boxtimes  \mathcal{A}_3)$ 
where 
$I = \{\Pi^i(\mathcal A) \mid i \in 1,2,3\}$.
\end{proof}

\subsection{Agreement}
\propmpca*
\begin{proof}
The existence of $\mathcal K$ is guaranteed since all actions are controllable and observable and $\mathscr L (\mathcal A)$ is regular, as well as $\mathfrak A$ \cite{Cassandras2006}.
 By contradiction assume $\mathscr{L}(\mathcal K)\subset \mathfrak{A} \cap \mathscr{L}(\mathcal{A})$, then there exists another controller $\mathcal K'_{\mathcal{A}}$ such that  $ \mathscr{L}(\mathcal K) \subset  \mathscr{L}(\mathcal K') = \mathfrak{A} \cap \mathscr{L}(\mathcal{A})$. 
\end{proof}

\procontroller*
\begin{proof}
In $\mathcal K_{\mathcal{A}}$ every request transition is removed in the first step, so it must be $\mathscr{L}(\mathcal K_{\mathcal{A}}) \subseteq \mathfrak{A} \cap \mathscr{L}(\mathcal{A})$. We will prove that  $\mathscr{L}(\mathcal K_{\mathcal{A}}) = \mathfrak{A} \cap \mathscr{L}(\mathcal{A})$, from this follows that $\mathcal K_{\mathcal{A}}$ is the most permissive controller. By contradiction assume that exists a trace $w \in \mathfrak{A} \cap \mathscr{L}(\mathcal{A}), w \not \in \mathscr L (\mathcal K_{\mathcal{A}})$. Then there exists a transition $t=(\vec{q},\vec{a},\vec{q'}) \not \in T_{\mathcal K_{\mathcal{A}}} $ in the accepting path of $w$ (i.e. the sequence of transitions used to recognise $w$). The transition $t$ is not a request since $w \in \mathfrak{A} \cap \mathscr{L}(\mathcal{A})$, and $\vec{q},\vec{q'} \not \in Hanged(\mathcal K_{\mathcal{A}})$ because the transition belongs to an accepting path. Since the only transitions removed to obtain $ \mathcal K_{\mathcal{A}}$ are requests and those involving hanged states, it follows that $t \in T_{\mathcal K_{\mathcal{A}}}$.
\end{proof}

\thecomposition*
\begin{proof}

\ref{lem:composition}) Assume by contradiction that $\mathcal{A}_1$ and $\mathcal{A}_2$ are non-collaborative, that is
\[(A^o_1 \cap co(A^r_2)) \cup (co(A^r_1) \cap A^o_2) = \emptyset\] 
Since the two automata are competitive, we  have 
\[A^o_1 \cap A^{o}_2 \cap (co(A^r_1) \cup co(A^r_2)) \neq \emptyset\] 
By the distributive law
\[(A^o_1 \cap (co(A^r_1) \cup co(A^r_2))) \cap (A^{o}_2 \cap (co(A^r_1) \cup co(A^r_2))) \neq \emptyset\]
By hypothesis the two automata are non-collaborative, hence the above term can be rewritten as
\[(A^o_1 \cap co(A^r_1)) \cap (co(A^r_2) \cap A^o_2) \neq \emptyset\]
%
%
By associative and commutative laws
\[(A^o_1 \cap co(A^r_2)) \cap (co(A^r_1) \cap A^o_2) \neq \emptyset\] 
Which implies
\[(A^o_1 \cap co(A^r_2)) \cup (co(A^r_1) \cap A^o_2) \neq \emptyset\] 
obtaining a contradiction.\\

\ref{lem:collcomp} ) 
By hypothesis the automata are collaborative:
\[(A_1^o  \cap  co(A_2^r)) \cup (A_2^o  \cap co(A_1^r)) \neq \emptyset\]
By hypothesis  $\mathcal{A}_1$ and $\mathcal{A}_2$ are safe, hence for each request there 
is a corresponding action, that is $co(A^r_i) \subseteq A^o_i$ where $i=1,2$. Then the following holds
\[
 \ A_i^o \cap co(A_i^r) = co(A_i^r) \qquad i=1,2
\]
By substitution in the previous term we obtain
\[
( A_1^o  \cap A_2^o \cap co(A_2^r))  \cup  (A_2^o  \cap A_1^o \cap co(A_1^r)) \neq \emptyset
\]
Which implies
\[
( A_1^o  \cap A_2^o \cap (co(A^r_1) \cup co(A^r_2)))  \cup  (A_2^o  \cap A_1^o \cap (co(A^r_1) \cup co(A^r_2)) ) \neq \emptyset
\]
By simplification we have
\[
( A_1^o  \cap A_2^o \cap (co(A^r_1) \cup co(A^r_2))) \neq \emptyset
\]
Hence $\mathcal{A}_1$ and $\mathcal{A}_2$ are competitive.

\ref{lem:safe})
 Note that the labels of $\mathcal{A}_1\otimes\mathcal{A}_2$ are the union of the labels of $\mathcal{A}_1$ and $\mathcal{A}_2$ (extended with idle actions for fitting the rank), hence no request transitions are added, and \hbox{$\mathcal{A}_1\otimes\mathcal{A}_2$} is \emph{safe}.
%
%
Since the traces of $\mathcal A_1 \otimes \mathcal A_2$ are a subset of $\mathcal A = \mathcal A_1 \boxtimes \mathcal A_2$,
$\mathcal A$ has at least a trace in agreement. Example~\ref{ex:liable} shows that not all the traces of $\mathcal A$ admit
agreement.
%

\ref{lem:unsafe}) Without loss of generality assume that $\mathcal{A}_1$ is unsafe, hence there exists a request $\vec{a}$, and 
traces $w,v$ such that $w\vec{a}v \in \mathscr{L}(\mathcal{A}_1)$. Since $\mathcal{A}_1$ and $\mathcal{A}_2$ are non-collaborative there will be no match between the actions of  $\mathcal{A}_1$ and $\mathcal{A}_2$, hence we have $w_1\vec{a'}v_1 \in \mathscr{L}(\mathcal{A}_1\otimes \mathcal{A}_2), w_2\vec{a'}v_2 \in \mathscr{L}(\mathcal{A}_1\boxtimes \mathcal{A}_2)$ for some $w_1,w_2,v_1,v_2$, where
$\vec a'$ is obtained from $\vec a$ by adding the idle actions to principals from
$r_{\mathcal{A}_1}+1$ to $r_{\mathcal{A}_1}+r_{\mathcal{A}_2}$.


\ref{lem:noncompetitive_safe}) The proof is similar to that of item~\ref{lem:safe}, indeed it suffices to prove that no new matches between principals in  $\mathcal{A}_1$ and $\mathcal{A}_2$ are introduced in  $\mathcal{A}_1 \boxtimes \mathcal{A}_2$. 
By item~\ref{lem:collcomp} it follows that $\mathcal{A}_1$ and $\mathcal{A}_2$ are non-collaborative:
\[(A_1^o  \cap  co(A_2^r)) \cup (A_2^o  \cap co(A_1^r)) \neq \emptyset\]
This suffices to prove that no matches will be introduced in their composition.
\end{proof}
%
%
%
\subsection{Weak Agreement}
%
\lemweaksafe*
\begin{proof}
Let $req^w_a,of^w_a$ be the number of requests and offers of an action $a \in \Rset \cup \Oset$ in a trace $w$.
\begin{enumerate}
\item For $\otimes$: we will prove that in every trace of $\mathcal{A}_1 \otimes \mathcal{A}_2$, for each action the number of requests are less than or equal to  the number of offers, and the thesis follows. 
By contradiction, assume that 
there exists a trace $w$ in $\mathcal{A}_1 \otimes \mathcal{A}_2$ and an action $a$ with $req^w_a > of^w_a$. Assume that $w$ is obtained combining two traces $w_1,w_2$ of $\mathcal{A}_1$ and $\mathcal{A}_2$, that is each principal in each automaton  performs the moves prescribed by its trace. Since both automata are weakly safe, we have $req^{w_1}_a \leq of^{w_1}_a$ and  $req^{w_2}_a \leq of^{w_2}_a$ for all actions $a$.

Independently of how many matches occur, in $w$ we still have more requests than offers: 
$req^{w1}_a + req^{w2}_a - k \leq of^{w1}_a + of^{w2}_a - k$ where $k$ are the new matches.

For $\boxtimes$ it suffices to take a trace $w$ in  $\mathcal{A}_1 \boxtimes \mathcal{A}_2$  obtained by combining two traces $w_1,w_2$ of respectively $\mathcal{A}_1$ and $\mathcal{A}_2$, where the match actions of both automata are maintained in $w$ (the matches are performed by the same principals). In this case, the trace $w$ will be present also in  $\mathcal{A}_1 \otimes \mathcal{A}_2$, hence $w \in \mathfrak{W}$.

\item
Without loss of generality assume that $\mathcal{A}_1$ is weakly unsafe, hence there exists an action $a$ and a trace $w_1$ in $\mathcal{A}_1$ such that $req^{w_1}_a > of^{w_1}_a$. Since $\mathcal{A}_1$ and $\mathcal{A}_2$ are non-collaborative,  in every trace $w$ of $\mathcal{A}_1 \otimes \mathcal{A}_2$ or  $\mathcal{A}_1 \boxtimes  \mathcal{A}_2$ obtained by shuffling  $w_1$ with an arbitrary $w_2$ in $\mathcal{A}_2$ we will have $req^{w}_a > of^{w}_a$.

\item
from Theorem \ref{the:composition} item~\ref{lem:noncompetitive_safe},  $\mathcal{A}_1 \boxtimes \mathcal{A}_2$ is safe
and since $\mathfrak A \subset \mathfrak W$ the thesis follows.\qedhere
\end{enumerate}
\end{proof}

The following proposition helps the proof of Theorem~\ref{the:contextsensitive}.

\begin{prop}\label{pro:obs_w}
Let $WA(\mathfrak{W})=\{ w \in  (\Rset \cup \Oset \cup \{ \tau \})^* \mid
 \exists f: [1 \ldots |w|] \rightarrow [1 \ldots |w| ] $
injective and such that  $ f(i)=j \text{ only if }  \ithel{w}{i}=co(\ithel{w}{j}), 
 \text{ total on the requests of $w$}  \}$. \\
Then, $Obs(w) \in WA(\mathfrak{W})$ implies $w \in \mathfrak{W}$. 
\end{prop}
\begin{proof}
Let $\sigma =Obs(w) \in WA(\mathfrak{W})$, and let $f$  be a function that certifies that $\sigma \in WA(\mathfrak{W})$, i.e.\ that all the requests in $w$ are fulfilled. Then $f$ certifies $w \in \mathfrak{W}$.
\end{proof}

\thecontextsensitive*
\begin{proof}
Example \ref{exa:wagreement} shows that the property is not context-free. 
For proving that $\mathfrak{W}$ is context-sensitive
we now outline a Linear Bounded Automata (LBA) \cite{Kuroda64} that decides whether a trace $w$ belongs to $\mathfrak{W}$,
giving us time and space complexity for the membership problem. 
Roughly, a LBA is a Turing machine with a tape, linearly bounded by the size of the input.
Since we have an infinite alphabet due to the (unbounded) rank of vector $\vec{a}$, we compute $Obs(w)$ and decide if $Obs(w) \in WA(\mathfrak{W})$. 
By Proposition~\ref{pro:obs_w} we obtain the thesis. 
Below is the scheme of the algorithm:

\begin{minipage}{1\textwidth}
\small
\begin{algorithmic}
\For{$i=0;i<length(w);i++$ } 
	\If{$w_i \in \Rset$}
		\For{$j=0;j<length(w);j++$} 
			\If{$w_j=co(w_i)$}
				\State $w_j \gets \#$
				\State break 
			\Else
				\If{$j=length(w)-1$}
					\Return{false}
				\EndIf
			\EndIf
		\EndFor
	\EndIf
\EndFor
\State \Return{true}
\end{algorithmic}
\end{minipage}\\

\noindent
The length of the tape equals the length of $w$, so the algorithm is $O(n)$ space, while it is $O(n^2)$ time, because 
of the two nested {\bf for} cycles.
%
\end{proof}

The following is an auxiliary result to the theorems below.

\begin{lem}\label{lem:flowtrace}
Let $\mathcal{A}$ be a contract automaton such that  $\vec x \in F_x$, then there exists a run $(w,\vec{q_0}) \rightarrow^* (\epsilon,\vec q_f)$ that passes through each $t_j \in T$ exactly $x_{t_j}$ times. 
\end{lem}
\begin{proof}
 We outline an algorithm that visits all the transitions $t_j$ with $x_{t_j}>0$, starting from $\vec q_f$ and proceeding backwards to $\vec q_0$. 
 
 We use auxiliary variables $\overline x_{t_j}, {t_j} \in T$, initialised to zero, for storing how many times we have passed through a transition ${t_j}$. At each iteration the algorithm selects non deterministically a transition ${\hat t}$ in the backward star of the selected node such that $x_{\hat t} - \overline x_{\hat t}>0$, and increases by one unit the variable $\overline x_{\hat t}$ for the selected $\hat t$. The next node will be the starting state of $\hat t$.
  The algorithm terminates when for all the transitions $t_j$ in the backward star we have $x_{t_j} - \overline x_{t_j}=0$.

We prove that the algorithm terminates and constructs a trace that passes through each ${t_j}$ exactly $x_{t_j}$ times, and the last transition considered leaves the initial state. 
For the first step we have $\sum_{t_j \in BS(\vec q_f)} x_{t_j} - \sum_{t_j \in FS(\vec q_f)} x_{t_j}=1$ hence there exists at least one $t_i \in BS(\vec q_f)$ such that $x_{t_i}>0$ (and $\overline x_{t_i}=0$).

Pick up one of these transitions, say $t_i$, and assign it to the iteration variable $\hat t$.
Two cases may arise, depending on the source of $\hat t$:

\begin{enumerate}
\item the source of $\hat t$ is $\vec{q} \neq \vec{q_0}$:  we have $\sum_{t_j \in BS(\vec{q})} x_{t_j}- \sum_{t_j \in FS(\vec{q})} x_{t_j}\geq 0$ and we know that $\sum_{t_j \in FS(\vec{q})} x_{t_j}> 0$, because $\hat t \in FS(\vec{q})$ and $x_{\hat t}>0$,  hence  $\sum_{t_j \in BS(\vec{q})} x_{t_j} > 0$.

We now show that there is at least one $t \in BS(\vec{q})$ such that $(x_{t} - \overline x_t)>0$. 
By contradiction, assume $\sum_{t_j \in BS(\vec{q})} x_{t_j} - \sum_{t_j \in BS(\vec{q})} \overline x_{t_j}=0$. We distinguish two cases:
\begin{itemize}
\item $\vec q = \vec q_f$: we have $\sum_{t_j \in FS(\vec{q})} \overline x_{t_j} = \sum_{t_j \in BS(\vec{q})} \overline x_{t_j}$, since at every iteration we increase by one unit the value of $\overline x_{t_i}$ for $\hat t$ and we are proceeding backwards starting from $\vec q_f$ (the flow variable of a loop belongs to both backward and forward star). 
Since  $\sum_{t_j \in BS(\vec{q})} x_{t_j} > \sum_{t_j \in FS(\vec{q})} x_{t_j}$, we have $\sum_{t_j \in FS(\vec{q})} x_{t_j} - \sum_{t_j \in FS(\vec{q})} \overline x_{t_j} < 0$. 
Contradiction, since by definition the value $ \overline x_{t_j}$ for a transition $t_j$ will never be greater then the corresponding value  $x_{t_j}$.
\item $\vec q \neq \vec q_f$: we have $\sum_{t_j \in FS(\vec{q})} \overline x_{t_j} > \sum_{t_j \in BS(\vec{q})} \overline x_{t_j}$.
Since  $\sum_{t_j \in BS(\vec{q})} x_{t_j} = \sum_{t_j \in
  FS(\vec{q})} x_{t_j}$, we also have $\sum_{t_j \in FS(\vec{q})} x_{t_j} - \sum_{t_j \in FS(\vec{q})} \overline x_{t_j} < 0$ obtaining a contradiction as above.
\end{itemize}

\noindent Then, we iterate the algorithm taking the above $t$ as $\hat t$.

\item the source of $t_i$ is $\vec{q_0}$: we have $\sum_{t_j \in BS(\vec{q_0})} x_{t_j} - \sum_{t_j \in FS(\vec{q_0})} x_{t_j}=-1$. \\
Let
$k_1 = \sum_{t_j \in FS(\vec{q_0})} x_{t_j} - \sum_{t_j \in FS(\vec{q_0})} \overline x_{t_j}$,
$k_2 = \sum_{t_j \in BS(\vec{q_0})} x_{t_j} - \sum_{t_j \in BS(\vec{q_0})} \overline x_{t_j}$, 
and note that since we are proceeding backwards starting from $\vec q_f$ it must be that 
 $\sum_{t_j \in FS(\vec{q_0})} \overline x_{t_j} = 1 + \sum_{t_j \in BS(\vec{q_0})} \overline x_{t_j}$.  
Hence, from the previous equations it must be that $k_2 - k_1 = 0$.
We have that:
\begin{itemize}
\item if $k_1=0$, 
we have $k_2=0$  and the algorithm terminates;

\item  if $k_1>0$, we have $k_2>0$ 
and the algorithm continues by selecting a transition $\hat t \in  BS(\vec{q_0})$ such that $x_{\hat t} - \overline x_{\hat t}=0$.
\end{itemize} 

\end{enumerate}

\noindent Since at every iteration we increase the value $\overline x_{\hat t}$, the constraints on $F_x$ guarantee that the algorithm will eventually terminate. 
Moreover there exists an execution of the algorithm that traverses all the possible cycles of the trace induced by $\vec x$. 
Hence we have a trace from $\vec{q_0}$ to $\vec q_f$ that passes through each transition $t_j$ visited by the algorithm exactly $x_{t_j}$ times.

It remains to prove that for all the transitions $t_j$ not visited by the algorithm we have $x_{t_j}=0$. By contradiction assume that there exists a transition $t_i=(\vec q_s, \vec a, \vec q_d)$ with $x_{t_i} - \overline x_{t_i} >0$ for all the possible executions of the algorithm. 

This is possible only if $\vec q_d$ it is not connected to $\vec q_f$ by the flow $\vec x$. Moreover in this case by the flow constraints on $\vec x$ it follows that $\vec q_s$ is not reachable from $\vec q_0$ by the flow $\vec x$, i.e.\ $t_i$ is not part of the trace induced by $\vec x$. 
Then there must exist a cycle   $C=\{ t_{c1}, \ldots, t_{cm} \}$ with $t_i \in C$ and
 disconnected from $\vec q_0$ and $\vec q_f$ with positive flow.
Let $Q_C$ be the set of nodes having ingoing or outgoing transitions in $C$.
 The constraints $\sum_{t \in BS(\vec{q})} x_{t} - \sum_{t \in FS(\vec{q})} x_{t}=0$ are satisfied for all $\vec q \in C$.


We show that $C$ will eventually violate the constraints defined by the variables $z^{\vec{q_s}}_{t_j}$. 
We have:

\[
\forall \vec{q'} \in Q:\ 
\sum_{t_j \in BS(\vec{q'})} z_{t_j}^{\vec{q_s}} - \sum_{t_j \in FS(\vec{q'})} z_{t_j}^{\vec{q_s}} = \left\{ 
\begin{array}{ll} 
- p^{\vec q_s} &  \mbox{if } \vec{q'}=\vec q_0\\ 
0  & \mbox{if } \vec{q'} \neq \vec q_0,\vec q_s \\
p^{\vec q_s} & \mbox{if } \vec{q'} = \vec q_s
\end{array}
\right. 
\]
\[
\forall t_j \in T. \ z_{t_j}^{\vec{q_s}} \in \mathds{R}, \quad 0 \leq z_{t_j}^{\vec{q_s}} \leq x_{t_j}
\]

We have  $\sum_{t_j \in FS(\vec{q_s})} x_{t_j}>0$ and  $p^{\vec q_s}=1$, hence $\sum_{t_j \in BS(\vec{q_s})}
 z_{t_j}^{\vec{q_s}} - \sum_{t_j \in FS(\vec{q_s})} z_{t_j}^{\vec{q_s}}=
 1$ and for all 
$\vec q \in Q_C, \vec q  \neq \vec{q_s}: \sum_{t_j \in BS(\vec{q})}
 z_{t_j}^{\vec{q_s}} - \sum_{t_j \in FS(\vec{q})} z_{t_j}^{\vec{q_s}}=
 0$. Note that is not possible to satisfy these constraints since 
for all $t \in C$, $x_t$ are all equal and positive and $0 \leq z_{t}^{\vec{q_s}} \leq x_{t}$.
  %
\end{proof}
%
\theflowweaksafe*
\begin{proof}
($\Rightarrow$) By contradiction assume that $\textsf{min } \gamma <0$. Hence there exists an action $a^j$ such that $v_j=1, \forall i \in I_l, i\neq j. v_i=0$ and $\gamma= \sum_{t_j \in T} a^j_{t_j} x_{t_j} < 0$. 
By Lemma \ref{lem:flowtrace} we know that $\vec x$ builds a trace recognising $w \in \mathscr{L}(\mathcal{A})$, and the number of offers for $a^j$ in $w$ are less than the corresponding number of requests since $\sum_{t_j \in T} a^j_{t_j} x_{t_j} < 0$, hence $w \not \in \mathfrak W$.

($\Leftarrow$)
By contradiction there exists $w \in  \mathscr{L}(\mathcal{A}) \setminus \mathfrak W$. 
Hence there exists an action $a^j$ that occurs in $w$ fewer times as an offer than as a request.
Let $\vec x$ be the flow induced in the obvious way by the trace $w$, counting the number of times each transition occurs in the path accepting $w$. 
We have $\sum_{t_j \in T} a^j_{t_j} x_{t_j} < 0$, hence it must be $\textsf{min }\gamma <0$.
\end{proof}

\theflowweakagreement*
\begin{proof}
($\Rightarrow$) Let $w$ be a trace in weak agreement,  and let $\vec x$ be the flow induced by $w$. 
Then by construction $\forall i \in I_l. \sum_{t_j \in T} a^i_{t_j} x_{t_j} \geq 0$, hence  $\text{max } \gamma \geq 0$.

($\Leftarrow$) 
Follows from Lemma~\ref{lem:flowtrace} and the hypothesis.
\end{proof}

\theweaklyliableflow*
\begin{proof}

($\Rightarrow$) By hypothesis  $\exists w_1$ such that $\forall w_3.\, w_1\vec{a}w_3 \in \mathscr{L}(\mathcal{A}) \setminus \mathfrak W$ and 
$\exists w_2.\,w_1 w_2  \in \mathscr{L}(\mathcal{A}) \cap \mathfrak W$.
Let $\overline t =(\vec q_s, \vec a, \vec q_d)$ be the transition such that 
$(w_1\vec a, \vec q_0) \rightarrow^* (\vec a, \vec q_s) \rightarrow (\epsilon, \vec q_d)$, i.e. the principal $i$ in $\vec a$ is weakly liable. 
We show that $\gamma_{\overline{t}} <0$.

Let $w_1$ from $\vec{q_0}$ to $\vec{q_s}$ induce the flow $\vec x$, while  $w_2$ from $\vec{q_s}$ to $\vec q_f$ induce $\vec y$.
Since  $w_1w_2$ is in weak agreement, $\forall i \in I_l. \sum_{t_j \in T} a^i_{t_j}(x_{t_j} + y_{t_j})\geq 0 $. 


Since by hypothesis the i-th principal is liable, the flow $\vec x$ corresponding to the trace $w_1$ is such that $g (\vec x) <0$. 
Otherwise if $g (\vec x)\geq 0$ we can choose a trace, say, $w_3$ such that $w_1 \vec a w_3 \in \mathscr L (\mathcal A) \cap \mathfrak W$, obtaining a contradiction.
Therefore, $\gamma_{\overline t}\leq g(\vec x)<0$.

($\Leftarrow$) 
 by hypothesis $\gamma_{\overline t}<0$ and by 
Lemma~\ref{lem:flowtrace} $\vec x$ corresponds to a run $w$ from the initial state to  $\vec q_s$ such that (by hypothesis again)  $\forall w_3.w_1\vec{a}w_3 \not \in \mathscr{L}(\mathcal{A})\cap \mathfrak W$ and $\exists w_4.w_1w_4 \in \mathscr{L}(\mathcal{A})\cap \mathfrak W$, that is $\overline t$ is a weakly liable transition. 
\end{proof}

\subsection{Automata and Horn Propositional Contract Logic}

For completeness, we first define the grammar for the full PCL, while the rules for its sequent calculus are in Figure~\ref{fig:fullgentzen}.
Unless stated differently, in what follows we only consider proofs without the rules $(weakR)$ and $(cut)$, which are proved to be redundant in~\cite{BartolettiZ09}. 

\begin{defi}[\bf{PCL}]
The formulae of PCL are inductively defined by the following grammar.\\

\bigskip
\vbox{
\noindent\textbf{} \\
\begin{tabular}{lp{4cm}l}
\hspace{16pt} $p \;\; ::=$ & $\bot$ & false \\
& $\top$ & true \\
& $a$ & prime \\
& $\neg p$ & negation  \\
& $p \vee p$ & disjunction \\
& $p \wedge p$ & conjunction \\
& $p \rightarrow p$ & implication \\
& $p \twoheadrightarrow p$ & contractual implication \\
\end{tabular} \vspace{4pt} \\
}
\end{defi}

\begin{figure}[tb]

\[
\irule{}{\Gamma, p \vdash p} id
\qquad
\irule{\Gamma,p \wedge q, p \vdash r}{\Gamma, p \wedge q \vdash r} \wedge L1
\qquad
\irule{\Gamma,p \wedge q, q \vdash r}{\Gamma, p \wedge q \vdash r} \wedge L2
\qquad
\irule{\Gamma \vdash p \quad \Gamma \vdash q}{\Gamma \vdash p \wedge q} \wedge R
\]
\medskip
\[
\irule{\Gamma, p \vee q, p \vdash r \quad \Gamma, p \vee q, q \vdash r}{\Gamma,p \vee q \vdash r} \vee L
\qquad
\irule{\Gamma \vdash p}{\Gamma \vdash p \vee q} \vee R1
\qquad
\irule{\Gamma \vdash q}{\Gamma \vdash p \vee q} \vee R2
\]
\medskip
\[
\irule{\Gamma \vdash p \quad \Gamma,p \vdash q}{\Gamma \vdash q} \ \ cut
\qquad
\irule{\Gamma, p \rightarrow q \vdash p \quad \Gamma, p \rightarrow q, q \vdash r}{\Gamma, p \rightarrow q \vdash r} \rightarrow L
\qquad
\irule{\Gamma, p \vdash q}{\Gamma \vdash p \rightarrow q} \rightarrow R
\]
\medskip
\[
\irule{\Gamma, \neg p \vdash p}{\Gamma,\neg p \vdash r} \neg L
\qquad
\irule{\Gamma, p \vdash \bot}{\Gamma \vdash \neg p}\neg R
\qquad
\irule{}{\Gamma, \bot \vdash p} \bot L
\]
\medskip
\[
\irule{}{\Gamma \vdash \top} \top R
\qquad
\irule{\Gamma \vdash \bot}{\Gamma \vdash p} \ \ weakR
\qquad
\irule{\Gamma \vdash q}{ \Gamma \vdash p \twoheadrightarrow q} \ \ Zero
\]
%
\medskip
\[
\irule{\Gamma, p \twoheadrightarrow q,r \vdash p \quad \Gamma, p \twoheadrightarrow q,q \vdash r}{\Gamma, p \twoheadrightarrow q \vdash r} \ \  Fix
\qquad
\irule{\Gamma, p \twoheadrightarrow q,a \vdash p \quad  \Gamma, p \twoheadrightarrow q,q \vdash b}{\Gamma, p \twoheadrightarrow q \vdash a \twoheadrightarrow b}  \ \ PrePost
\]
\caption{The rules of the sequent calculus for PCL. The contractual implication rules  are  $Zero,Fix \text{ and }Prepost$ while the others are the standards for Intuitionistic logic.}
\label{fig:fullgentzen}
\end{figure}

%
%
%
%

The following proposition will be helpful later on.


\propuno*
\begin{proof}
The first item follows immediately from Definition~\ref{def:translationPCL}.
\\
For the second item, we first consider the translation of the clauses in the formula.
By construction, for each of them two cases are possible when considering request actions: either 
$\llbracket ( \bigwedge_{j \in J_i}a_j) \rightarrow b  \rrbracket$ or
$\llbracket (\bigwedge_{j \in J_i}a_j) \twoheadrightarrow b \rrbracket$.
In both cases we have outgoing request transitions of the form
$\{  (J' \cup \{j\}, a_j, J') \mid J' \cup \{j\} \in  2^{J_i}, j \in J \}$.
Finally by applying the associative composition $\boxtimes$ (Definition~\ref{def:aproduct}),  some requests
may be matched with corresponding offers, but no new request can be originated.
\\
The third item follows immediately by the translation and by the condition in Definition~\ref{def:hornpcl}, 
that all the atoms are different.
\end{proof}

The following lemma shows that if an atom $a$ is entailed by a formula $p$ then
there is a trace recognised by the contract automaton $\llbracket p \rrbracket$ 
where the request corresponding to 
the atom $a$, if any,  is always matched.

\begin{lem}
\label{lem:if}
Given a H-PCL formula $p$  and an atom $a$ in $p$ we have:
\[
p \vdash a \text{ is provable implies }\exists w \in \mathscr{L}(\llbracket p \rrbracket)
\text{ such that no } \vec{a}\text{ request on }a \text{ occurs in } w
\]
\end{lem}
\begin{proof}
Consider each of the conjuncts $\alpha$ of $p$.
If $a$ does not appear in $\alpha$ as the premise of an implication/contractual implication, 
then the statement follows trivially by Definition~\ref{def:translationPCL} and by hypothesis, 
since the translation of $a$ is an offer action. 
Otherwise $a$ also occurs in $\alpha$ within:
\begin{enumerate}[label=\arabic*.]
\item a conjunction, or
\item the conclusion of a contractual implication, or
\item the conclusion of an implication. 
\end{enumerate}

For the first two cases, by Definition~\ref{def:translationPCL}, a transition labelled by the relevant offer $\overline a$ is available in all states, so preventing a request $a$ to appear in $ \llbracket p \rrbracket$, i.e.\ after the product of the principals (Definition~\ref{def:aproduct}). 

For proving case 3, $\alpha = \bigwedge_{j \in J}a_j \rightarrow a$ and we proceed by induction on the depth of the proof of $p \vdash a$. 
It must be the case then that $\forall j$ it holds $p \vdash a_j$.
We can now either re-use the proof for cases 1 and 2 (that act as base cases), or the induction hypothesis if $a_j$ occurs in the conclusion of an implication.  By Definition~\ref{def:translationPCL} after all $a_j$ are matched, the offer $a$ will be 
always available, preventing a request $a$ to appear.
%
%
%
\end{proof} 

In order to keep the following definition compact, we use  $\circ$ for either $\rightarrow$ or $\twoheadrightarrow$.
In addition, by abuse of the notation we also use $\land$ to operate between formulas, we write $p'$ for an empty formula or  with a single clause, and we allow the indexing sets $J$ and $K$ in clauses to be empty.
Finally, we let $(\bigwedge_{j \in \emptyset}a_j) \circ b$ stand for $b$.

\begin{defi}
\label{def:meno-b}
Given a formula $p$, if from the initial state of $\llbracket p \rrbracket$ there is an outgoing offer or an outgoing match transition with label $\vec a$, we define
\[
  p \slash \vec{a}  = \left\{ 
    \begin{array}{ll} 
     p & \text{ if $\vec{a}$ is an offer}      
      \\
      p' \wedge  (\bigwedge_{z \in Z}c_z \twoheadrightarrow b)  \wedge  ( \bigwedge_{j \in J}a_j ) \circ b' \ & 
       \text{ if $\vec{a}$ is a match with $\ithel{\vec{a}} i=b$ and }\ \\
      & p=p' \wedge (\bigwedge_{z \in Z}c_z \twoheadrightarrow b) \wedge  ( \bigwedge_{j \in J}a_j \wedge b) \circ b' 
      \\
      p' \wedge  (\bigwedge_{k \in K} a_k \wedge b) \wedge  ( \bigwedge_{j \in J}a_j ) \circ b' \ & 
      \text{ if $\vec{a}$ is a match  with $\ithel{\vec{a}} i=b$ and }\ \\
    &  p' \wedge  (\bigwedge_{k \in K} a_k \wedge b) \wedge  ( \bigwedge_{j \in J}a_j  \wedge b) \circ b'
    \end{array} 
  \right.
\]
\end{defi}

%
%

We now establish a relation between $\llbracket p \slash \vec{a} \rrbracket$,
and the contract automaton
obtained by changing the initial state $\vec q_0$ of $\llbracket p \rrbracket$ 
to $\vec q$, for  the transition $(\vec q_0, \vec a, \vec q)$ of $\llbracket p \rrbracket$.  
The main idea is to relate  the formula $p \slash \vec{a}$ to the residual of the automaton 
$\llbracket p \rrbracket$ after the execution of an initial transition labelled by $\vec a$, that is 
$\llbracket p \slash \vec a \rrbracket$.
Recall that the translation given in Definition~\ref{def:translationPCL} yields deterministic automata.

\begin{lem}
\label{lem:automata_traduction}
Given a H-PCL formula  p and the contract automaton
  $\llbracket p \rrbracket=\langle Q, \vec{q_0},A^r,A^o,T,F \rangle$, 
if \mbox{$t=(\vec{q_0},\vec{a},\vec{q}) \in T$} is an offer or a match transition,  
then $\mathscr{L}(\mathcal{A}) =\mathscr{L}(\llbracket p \slash \vec{a} \rrbracket)$ 
where \linebreak \mbox{$\mathcal{A}=\langle Q, \vec{q},A^r,A^o,T,F \rangle$}.

\end{lem}
\begin{proof}
The proof is by cases of $\vec a$. If $\vec{a}$ is an offer, then  by  
Definition~\ref{def:translationPCL} it must be  $\vec{q}=\vec{q_0}$
 and trivially $\mathcal{A}=\llbracket p \slash \vec{a} \rrbracket$.

Otherwise, since $\vec{a}$ is a match action, say on atom $b$, it contains a request from, say, the $i$-th principal and a corresponding offer from another.
Therefore, $p = \bigwedge_{k \in K} \alpha_k$ contains within a clause $\alpha_j$ the atom $b$, originating the offer, as a conjunction or as a conclusion of a contractual implication (note that it cannot be an implication because we are in the initial state), and $\alpha_i$ also contains $b$ originating this time the request.
%
%
We now prove that  the automata $\mathcal A$ and $\llbracket p \slash \vec{a} \rrbracket$ have the same initial state.
Let $\vec{q_0}=\langle J_1, \ldots, J_n \rangle$, then, since $\ithel{\vec{a}} i = b$, 
the states $\vec{q_0}$ and $\vec{q}$ only differ in the $i$-th element, 
where in $\ithel{\vec{a}} i$ the request action $b$ is not available anymore; formally,
$\forall j \neq i$ it must be 
$\ithel{\vec{q}} j = \ithel{\vec{q_0}} j = J_j$, and 
$\ithel{\vec{q}} i=\ithel{\vec{q_0}} i  \setminus \{i\}$.  
By Definition~\ref{def:meno-b} $p$ and $p \slash \vec{a}$
differ because of the single atom $b$ has been removed from $\alpha_i$.
By these facts and by item 2 of Proposition~\ref{prop:1} the language equivalence follows.
Indeed, $\llbracket p \slash \vec{a} \rrbracket$ is the product of the same $\llbracket \alpha_k \rrbracket, k \neq i$ used for $\llbracket p \rrbracket$, and the match on $b$ of $\mathcal{A}$ leaves $\vec q_0$, that is not reachable from 
$\vec q$.
\end{proof}

\begin{figure}[tb]
\centering
\begin{tikzpicture}[->,>=stealth',shorten >=1pt,auto,node distance=2.3cm,
                    semithick, every node/.style={scale=0.6}]
  \tikzstyle{every state}=[fill=white,draw=black,text=black]

  \node[initial,state] 				 (A)                    {$\vec{q_1}$};
  \node[state] 				         (B) [right of=A]   {$\vec{q_2}$};
  \node[state]						 (C) [right of=B]   {$\vec{q_3}$};
  \node[state]         				 (D) [below of=A] {$\vec{q_4}$};	
  \node[state] 				         (E) [right of=D]   {$\vec{q_5}$};
  \node[state,accepting]			 (F) [right of=E] {$\vec{q_6}$};
  \node[state]				(G) [right of=F] {$\vec{q_8}$};
  \node[state]				(H) [above of=C] {$\vec{q_7}$};	
 
  \path (A)  edge             				   node{$(b,\overline b,\blk)$} (B)
			    edge             				   node[left]{$(\blk,c,\overline{c})$} (D)
			  edge	[bend left]		node {$(\blk,a,\blk)$} (H)
		  (B)  edge   [loop above]           node {$(\blk, \overline b, \blk)$} (B)
				edge         				       node {$(\overline{a},a,\blk)$} (C)
				edge         				       node[left] {$(\blk,c,\overline{c})$} (E)
		  (C)  edge                             node {$(\blk,c,\overline{c})$} (F)
				edge [loop right]               node {$*$} (C)
		  (D)  edge [loop left]             node {$(\blk, \blk, \overline c)$} (D)
				edge                             node{$(b,\overline b,\blk)$} (E)
		  (E)  edge[loop below]	               node {$ (\blk, \overline b, \blk), (\blk, \blk, \overline c)$} (E)
				edge              	     	   node {$(\overline{a},a,\blk)$} (F)
		   (F)  edge [loop below]             node {$**$} (F)
		  (H)  edge[bend left]	   node{$(\blk, c ,\overline c)$}(G)
			edge 			 node{$(b,\overline b, \blk)$}(C)
		 (G)  edge 			node{$(b,\overline b, \blk)$}(F)
		      edge[loop below]	node{$(\blk,\blk,\overline c)$}(G);
\end{tikzpicture} 
\caption{The contract automaton $\llbracket Alice \wedge Bob \wedge Charlie \rrbracket$discussed in Example~\ref{ex:hpcltranslation} is displayed here, where the principals are those of Figure~\ref{fig:pclautomata}, and \\ $*= (\overline a, \blk, \blk), (\blk, \overline b, \blk)$, $** =(\overline a, \blk, \blk), (\blk, \overline b, \blk), (\blk, \blk, \overline c)$.}
\label{fig:alicebobcharlieproduct}
\end{figure}

\begin{exa}
\label{ex:hpcltranslation}
 Let $\llbracket p \rrbracket$ be the automaton shown in Figure~\ref{fig:alicebobcharlieproduct},
where $p = Alice \wedge Bob \wedge Charlie$ and the principals are those of Figure~\ref{fig:pclautomata}.
Consider now $p' = p \slash (b,\overline b, \blk)=
(\,a \wedge ( (a \wedge c) \twoheadrightarrow b) \wedge c\,) $
and build
$\llbracket p' \rrbracket = \{ \langle \{ \vec q_2,\vec q_3,\vec q_5, \vec q_6 \}, \vec{q_2},A^r,A^o,T,\vec{q_6} \rangle  \}$
(transitions, alphabets and states are taken from $\llbracket p \rrbracket$).
It is immediate to verify that the language of $\llbracket p' \rrbracket$ is the same of $\llbracket p \rrbracket$, when the initial state is $\vec q_2$ instead of $\vec q_1$.
\end{exa}

The following lemma is auxiliary for proving the next theorem.
Its second item is similar to Lemma 1 in~\cite{Pfenning2000}.
\begin{lem}\label{lem:PCLaux}
Let $a,b$ be atoms, $p,q$ be conjunction of atoms, with $q$ possibly empty, $p_1, \ldots, p_n$ be formulae,
and $\circ \in \{\rightarrow, \twoheadrightarrow \}$,
 then
\begin{enumerate}[label=(\roman*)]
\item
$
\text{if }\quad
\dfrac{
 \Delta
}
{
\Gamma, q \circ b \vdash p
}
\quad
\text{ then}
\quad \exists \Delta' :
\dfrac{ \Delta'
}
{
\Gamma, (q \wedge a) \circ b, a \vdash p
}
$
\item 
if \quad $\Gamma \vdash p$ \quad then \quad $\forall \Gamma'.\,\,\Gamma,\Gamma' \vdash p$

\item
if \quad $\bigwedge_{i \in 1 \ldots n} p_i \vdash q$  \quad then \quad $p_1, \ldots, p_n \vdash q$

\item
if \quad $\Gamma \vdash \bigwedge_{i \in 1 \ldots n} p_i$  \quad then \quad $\forall i. \, \Gamma \vdash p_i$

\end{enumerate}
\end{lem}
\begin{proof}
To prove the first item, we  proceed by induction on the depth of $\Delta$ and by case analysis on the last rule applied. 
In the base case $\Delta$ is empty, we have two cases
\begin{enumerate}
\item $q \text{ non-empty or } p\neq b $: then it must be that $\Gamma = p, \Gamma'$  for some $\Gamma'$
and the last rule
applied is $id$. Trivially, $\Delta'$ will be empty and we have
\[
\dfrac{ 
}
{
\Gamma',p, (q \wedge a) \circ b, a \vdash p
}{id}
\]
\item $q \text{ empty and }p=b$: our hypothesis reads as $\,\,\dfrac{}{\Gamma, b \vdash b}{id}$, and we build the following deduction
\[
\dfrac{
	\dfrac{}{\Gamma', a \circ b, a \vdash a}{id}
	\quad
	\dfrac{}{\Gamma, a \circ b, a, b \vdash b}{id}
}
{
\Gamma, a \circ b, a \vdash b
}{\diamondsuit}
\]
where if $\circ = \rightarrow$ then $\Gamma'  = \Gamma$ and  $\diamondsuit= \rightarrow L$, otherwise if
$\circ = \twoheadrightarrow$ then $\Gamma' = \Gamma,b$ and $\diamondsuit = Fix$.
\end{enumerate}
For the inductive step, we distinguish two cases:
\begin{enumerate}
\item the last rule applied to deduce the hypothesis does not involve $q \circ b$.  
Hence the rule must be applied on $p$ or on a formula in $\Gamma$. 
We can apply the same rule to 
$\Gamma, (q \wedge a) \circ b, a \vdash p$ and use the inductive hypothesis.
\item  the last rule applied to deduce the hypothesis involves $q \circ b$.  
There are two exhaustive cases

\begin{enumerate}
\item $\circ = \rightarrow$, 
then the last rule applied is $\rightarrow L$ and the deduction tree has the following form:
\[
\dfrac{
	\dfrac{\Delta_1}{
	\Gamma,  q \rightarrow b \vdash q
	}
	\quad 
	\dfrac{\Delta_2}{\Gamma, q \rightarrow b, b \vdash p}
}
{\Gamma,  q \rightarrow b \vdash p } \rightarrow L
\]
Then by induction hypothesis we have
\[ 
	\dfrac{ \Delta_1'}
		{\Gamma, ( q\wedge a) \rightarrow b, a \vdash q}
	\qquad\qquad 
	\dfrac{\Delta_2'}{\Gamma, ( q\wedge a) \rightarrow b, a, b \vdash p}
\]

From the right one and a derivation tree $\Delta_3$ detailed below, we build 

\[
\dfrac{
	\Delta_3
	\quad 
	\dfrac{\Delta_2'}{\Gamma, ( q\wedge a) \rightarrow b, a, b \vdash p}
}
{\Gamma, ( q\wedge a) \rightarrow b, a \vdash p } \rightarrow L
\]
$\Delta_3$ is the derivation tree:
\[
\dfrac{
		\dfrac{ \Delta_1'}
		{\Gamma, ( q\wedge a) \rightarrow b, a \vdash q} 
	\quad
	\dfrac{}
		{\Gamma, (q\wedge a) \rightarrow b, a \vdash  a }id		
	}
	{
	\Gamma, (q\wedge a) \rightarrow b,a \vdash  q\wedge a
	}\wedge R
\]

\item $\circ = \twoheadrightarrow$, 
then the last rule applied is $Fix$ and the deduction tree has the following form:
\[
\dfrac
{
	\dfrac
	{
		\Delta_1
	}
	{
	\Gamma, q \twoheadrightarrow b, p  \vdash  q
	} \quad
	\dfrac
	{
		\Delta_2
	}
	{
	\Gamma,q \twoheadrightarrow b, b  \vdash p
	}
}
{\Gamma,q \twoheadrightarrow b  \vdash p } Fix
\]
Then by the induction hypothesis we have
\[
\dfrac
	{\Delta_1'} {\Gamma, ( q \wedge a) \twoheadrightarrow b,a ,p  \vdash   q}
	\qquad \qquad
\dfrac
	{\Delta_2'}{\Gamma, ( q \wedge a) \twoheadrightarrow b,a ,  b  \vdash p}	
\]
From the above, we build the following
\end{enumerate}
\[
\dfrac
{
	\dfrac
	{
		\dfrac
		{
			\Delta_1'
		}
		{
		\Gamma, ( q \wedge a) \twoheadrightarrow b,a ,p  \vdash   q
		}\quad
		\dfrac
		{
		}
		{
		\Gamma, ( q \wedge a) \twoheadrightarrow b,a ,p  \vdash   a
		} id
	}
	{
	\Gamma, ( q \wedge a) \twoheadrightarrow b,a ,p  \vdash   q \wedge a
	} \wedge R
	\quad
	\dfrac
	{
		\Delta_2'
	}
	{
	\Gamma, ( q \wedge a) \twoheadrightarrow b,a ,  b  \vdash p
	}
}
{\Gamma, ( q \wedge a) \twoheadrightarrow b,a  \vdash p } Fix
\] 
\end{enumerate}

\vspace{5pt}

\noindent For the second item, we prove a stronger fact: the last rule used to deduce $\Gamma, \Gamma' \vdash p$ is the same used for proving $\Gamma \vdash p$.
We proceed by induction on the depth of the derivation for $\Gamma \vdash p$ and then by case analysis on the last rule applied.

The base case is when the axiom $id$ is applied, and the proof is immediate.

For the inductive case, we assume that for some rule $\diamondsuit$
\[
\dfrac{\Delta}{\Gamma \vdash p} \,\diamondsuit
\qquad \text{implies} \qquad
\dfrac{\overline{\Delta}}{\Gamma, \Gamma' \vdash p} \,\diamondsuit
\]
Rather than considering each rule at a time, we group them in two classes: those with two premises, and those with one premise. 
Below, we discuss the first case, and the second follows simply erasing one premise in what follows.
The deduction tree in the premise above has the following form
\[
\dfrac
   { 
   \dfrac{\Delta'}{\overline{\Gamma} \vdash q}
    \qquad
    \dfrac{\Delta''}{\overline{\Gamma'} \vdash q'} 
    }
   {
   \Gamma \vdash p
   }\, \diamondsuit
\]
and by applying the induction hypothesis to both the premises we conclude
\[
\dfrac
   { 
   \dfrac{\overline{\Delta'}}{\overline{\Gamma}, \Gamma' \vdash q}
    \qquad
    \dfrac{\overline{\Delta''}}{\overline{\Gamma'}, \Gamma' \vdash q'} 
    }
   {
   \Gamma, \Gamma' \vdash p
   }\, \diamondsuit
\]
%
%
%
%

\smallskip
%
%
%
%
%
\noindent
Moreover note that in this fragment no contradictions can be introduced. 

For the third item, we have a derivation tree $\Delta$ for the sequent $\bigwedge_{i \in 1 \ldots n} p_i \vdash q$.
To build a derivation tree $\Delta'$ for $p_1, \ldots, p_n \vdash q$ apply the following two steps.
The first step removes from $\Delta$ all the rules $\wedge L_i$ applied to (each sub-term of)
$\bigwedge_{i \in 1 \ldots n} p_i $, obtaining $\Delta''$. 
Then, replace all applications of the axiom $(id)$ in $\Delta''$ of the form
\[
\dfrac{}{\Gamma, \bigwedge_{j \in J} p_j \vdash \bigwedge_{j \in J} p_j}{id}
\]
with a derivation tree with $k = |J|$ leaves of the form
\[
\dfrac{}{\Gamma, p_1, p_2, ...,p_k \vdash p_j}\,{id}
\]
and by repeatedly applying the rule $(\land R)$ until we obtain the relevant judgement
\[\Gamma, p_1, p_2, ...,p_k \vdash \bigwedge_{i \in 1 \ldots n} p_i\,.\]

\smallskip
\noindent
For the fourth item, we have a derivation tree $\Delta$ for the sequent $\Gamma \vdash \bigwedge_{i \in \{ 1 \ldots n\}} p_i$.

\noindent
For each sequent $\Gamma \vdash p_j$, $j \in \{ 1 \ldots n\}$, the derivation tree is then:
\[
\dfrac
{
	\dfrac{\Delta}{\Gamma \vdash p_j \wedge \bigwedge_{i \in \{ 1 \ldots n\} \setminus \{j\}} p_i}
	\quad
	\dfrac{
		\dfrac{}{ p_j , \bigwedge_{i \in \{ 1 \ldots n\} \setminus \{j\}} p_i \vdash p_j }{id}
	}
	{ p_j \wedge \bigwedge_{i \in \{ 1 \ldots n\} \setminus \{j\}} p_i \vdash p_j }{\wedge L1}
}
{
\Gamma \vdash p_j
}{cut}\vspace{-24 pt}
\]
\end{proof}\medskip

\theNPCLagreement*
\proof

($\Rightarrow$) Since $p\vdash \lambda(p)$ by Lemma~\ref{lem:PCLaux}$(iv)$ (where $\Gamma=p$) we have $p \vdash a$ for all atoms $a$ in $p$. 
It suffices to apply Lemma~\ref{lem:if} to each of these atoms, and by Definition~\ref{def:translationPCL} the offers are never consumed, there must be a trace $w \in \mathscr{L}(\llbracket p \rrbracket)$ where all the requests are matched.

($\Leftarrow$)
Let $\vec{q_0}$ be the initial state of $\llbracket p \rrbracket$ and $\vec{f}$ be the final state. 
We proceed by induction on the length of $w$. 

In the base case $w$ is empty, hence the initial state of $\llbracket p \rrbracket$ is also final. 
This situation only arises when the second rule of Definition~\ref{def:translationPCL} has been applied for all conjuncts $\alpha_i$ corresponding to principals.
Therefore it must be that  $p$ is a conjunction of atoms, so $p = \lambda(p) $ and the thesis holds immediately. 

For the inductive step we have $w=\vec{a}w_2$, and $(\vec{a}w_2,\vec{q_0}) \rightarrow (w_2,\vec{q}) \rightarrow^+ (\epsilon,\vec{f})$.
By inductive hypothesis and Lemma~$\ref{lem:automata_traduction}$ we have $p \slash \vec{a} \vdash \lambda(p \slash \vec{a})$. 
If $\vec{a}$ is an offer by Definition~\ref{def:meno-b} we have $p = p \slash \vec{a}$ and the thesis holds directly. 
Note that $\lambda(p)=\lambda(p \slash \vec{a})$ because $\vec{a}$ labels a match or an offer transition outgoing from $\vec q_0$ and the offer comes from the conclusion of a contractual implication or a conjunction of atoms, that is unmodified in $ p \slash \vec{a}$ .
Hence since by inductive hypothesis
$p \slash \vec{a} \vdash \lambda(p \slash \vec{a})$ and since $\lambda(p)=\lambda(p \slash \vec{a})$,  
proving $p \vdash p \slash \vec{a}$ entails $p \vdash \lambda(p)$. 
This is because of the following proof (note that there exists a longer one, cut-free)
 and  Lemma~\ref{lem:PCLaux} $(ii)$
\[
\dfrac{p \vdash p \slash \vec{a} \quad p, p \slash \vec{a} \vdash \lambda(p)}
{p \vdash \lambda(p)} \,cut
\]

To prove  $p \vdash p \slash \vec{a}$ we proceed by cases according to the structure of $p$, (omitting the cases for $J = \emptyset$ for which the proof is trivial)


\begin{itemize}
\item
if $p = p' \wedge (\bigwedge_{z \in Z}c_z \twoheadrightarrow b) \wedge  ( \bigwedge_{j \in J}a_j \wedge b \rightarrow b')$
 we have to prove the sequent $p \vdash p \slash \vec{a}$ that reads as
\[(\, p' \wedge (\bigwedge_{z \in Z}c_z \twoheadrightarrow b) \wedge  ( \bigwedge_{j \in J}a_j \wedge b \rightarrow b' )\,)\vdash  (\, p' \wedge  (\bigwedge_{z \in Z}c_z \twoheadrightarrow b)  \wedge  ( \bigwedge_{j \in J}a_j  \rightarrow b') \,)\]
For readability, we first determine the sequent $\Gamma \vdash ( \bigwedge_{j \in J}a_j ) \rightarrow b'$ where 
\[\Gamma = p' , (\bigwedge_{z \in Z}c_z \twoheadrightarrow b),  ( \bigwedge_{j \in J}a_j \wedge b) \rightarrow b'\]
from $p$, by applying the rule $\wedge R$, and Lemma~\ref{lem:PCLaux}$(iii)$.
Then we build the following derivation, where * is detailed below:
 \[
 \dfrac{
	\dfrac{ 
			\dfrac{
				\dfrac{}{\Gamma,\bigwedge_{j \in J}a_j  \vdash \bigwedge_{j \in J}a_j } id \quad
				\dfrac{*}{\Gamma,  \bigwedge_{j \in J}a_j \vdash b } \diamondsuit
			}{\Gamma,  \bigwedge_{j \in J}a_j \vdash \bigwedge_{j \in J}a_j \wedge b} \wedge R \quad
			\dfrac{}{\Gamma,  \bigwedge_{j \in J}a_j , b' \vdash b'} id
	}
	{\Gamma,  \bigwedge_{j \in J}a_j \vdash b'} \rightarrow L}
{\Gamma \vdash   ( \bigwedge_{j \in J}a_j ) \rightarrow b'} \rightarrow R
\]
The fragment * of the proof can have two different forms, depending on the set $Z$:
\begin{itemize}
\item if $Z=\emptyset$, then * is empty and the rule $\diamondsuit$ is $id$

\item otherwise the fragment * consists of the two sub-derivations below, and the rule $\diamondsuit$ applied to them is $Fix$ 
\begin{equation}
\label{delta3}
\dfrac{ \Delta_3 }{\Gamma, \bigwedge_{j \in J}a_j, b  \vdash \bigwedge_{z \in Z}c_z}
\end{equation}
\[
		\dfrac{}{\Gamma, \bigwedge_{j \in J}a_j, b  \vdash b} id
\]\\
\end{itemize}

We now show how to obtain $\Delta_3$. 
Let $\Delta$ be the derivation tree for $p \slash \vec{a} \vdash \lambda(p \slash \vec{a})$, that exists by the inductive  hypothesis. 
Note that since $\lambda(p \slash \vec{a})$ is a conjunction where $\bigwedge_{z \in Z}c_z$ occurs, the following proof
can be obtained by applying Lemma~\ref{lem:PCLaux}$(iv)$ for all $c_z$ and by combining them with rule $\wedge R$:
\begin{equation}
\label{delta2}
\dfrac{ \Delta_2}
{
(p' , (\bigwedge_{z \in Z}c_z \twoheadrightarrow b) ,  ( \bigwedge_{j \in J}a_j  \rightarrow b' )) \vdash \bigwedge_{z \in Z}c_z
}
\end{equation}

Now, in order to obtain the following from the proof~(\ref{delta2}), i.e.
\begin{equation} \label{delta3}
\dfrac
{ \Delta_3}
{ 
(p' , (\bigwedge_{z \in Z}c_z \twoheadrightarrow b),  ( \bigwedge_{j \in J}a_j \wedge b) \rightarrow b', \bigwedge_{j \in J}a_j, b) \vdash \bigwedge_{z \in Z}c_z
}
\end{equation}
we apply Lemma~\ref{lem:PCLaux}$(ii)$: the left hand-side of the sequent 
\[
(p' , (\bigwedge_{z \in Z}c_z \twoheadrightarrow b) ,  ( \bigwedge_{j \in J}a_j  \rightarrow b' ) \,)\vdash \bigwedge_{z \in Z}c_z
\]
is augmented with  $\bigwedge_{j \in J}a_j$. 
Finally by applying Lemma~\ref{lem:PCLaux}$(i)$, the formula 
\mbox{$ \bigwedge_{j \in J}a_j  \rightarrow b' $} above becomes
$  ( \bigwedge_{j \in J}a_j \wedge b) \rightarrow b',b$, obtaining~(\ref{delta3}).\\


\item
if $p = p' \wedge (\bigwedge_{k \in K} a_k \wedge b) \wedge ( ( \bigwedge_{j \in J}a_j \wedge b) \rightarrow b')$
we have to prove the sequent $p \vdash p \slash \vec{a}$ that reads as
\[(\, p' \wedge (\bigwedge_{k \in K} a_k \wedge b) \wedge  ( \bigwedge_{j \in J}a_j \wedge b) \rightarrow b' \,)\vdash  (\, p' \wedge   (\bigwedge_{k \in K} a_k \wedge b)   \wedge  ( \bigwedge_{j \in J}a_j  \rightarrow b') \,)\]
For readability, we first determine the sequent $\Gamma \vdash ( \bigwedge_{j \in J}a_j ) \rightarrow b'$, where 
\[\Gamma =  p' ,  (\bigwedge_{k \in K} a_k \wedge b) ,  (( \bigwedge_{j \in J}a_j \wedge b) \rightarrow b')\]
from $p$, by applying the rule $\wedge R$ and Lemma~\ref{lem:PCLaux}$(iii)$.
Then we build the following derivation, where * is detailed below:
 \[
 \dfrac{
	\dfrac{ 
			\dfrac{
				\dfrac{}{\Gamma,\bigwedge_{j \in J}a_j  \vdash \bigwedge_{j \in J}a_j } id \quad
				\dfrac{*}{\Gamma,  \bigwedge_{j \in J}a_j \vdash b } \diamondsuit
			}{\Gamma,  \bigwedge_{j \in J}a_j \vdash \bigwedge_{j \in J}a_j \wedge b} \wedge R \quad
			\dfrac{}{\Gamma,  \bigwedge_{j \in J}a_j , b' \vdash b'} id
	}
	{\Gamma,  \bigwedge_{j \in J}a_j \vdash b'} \rightarrow L}
{\Gamma \vdash   ( \bigwedge_{j \in J}a_j ) \rightarrow b'} \rightarrow R
\]

The fragment * of the proof can have two different forms, depending on the set $K$:
\begin{itemize}
\item if $K=\emptyset$, then * is empty and the rule $\diamondsuit$ is $id$

\item otherwise the rule $\diamondsuit$ is $\wedge L2$ applied to the fragment * below

\[
\dfrac{}
{p' ,  (\bigwedge_{k \in K} a_k \land b) , b, (( \bigwedge_{j \in J}a_j \wedge b) \rightarrow b'),  \bigwedge_{j \in J}a_j \vdash b}{id}
\]
\end{itemize}

\item 
if $p =  p' \wedge (\bigwedge_{z \in Z}c_z \twoheadrightarrow b) \wedge  (( \bigwedge_{j \in J}a_j \wedge b) \twoheadrightarrow b')$
we have to prove the sequent $p \vdash p \slash \vec{a}$ that reads as
\[(\, p' \wedge (\bigwedge_{z \in Z}c_z \twoheadrightarrow b) \wedge  ( \bigwedge_{j \in J}a_j \wedge b) \twoheadrightarrow b' \,)
\vdash  
(\, p' \wedge  (\bigwedge_{z \in Z}c_z \twoheadrightarrow b)  \wedge  ( \bigwedge_{j \in J}a_j  \twoheadrightarrow b') \,)\]
For readability, we first determine the sequent $\Gamma \vdash  \bigwedge_{j \in J}a_j \twoheadrightarrow b'$ where
\[\Gamma = p' , (\bigwedge_{z \in Z}c_z \twoheadrightarrow b),  ( \bigwedge_{j \in J}a_j \wedge b \twoheadrightarrow b')\,,\] by applying the rule $\wedge R$ and Lemma~\ref{lem:PCLaux}$(iii)$.
Then we build the following derivation, where * is detailed afterwards:
 \[
 \dfrac{
	\dfrac{ 
			\dfrac{
				*
			}{\Gamma, b' \vdash \bigwedge_{j \in J}a_j \wedge b} \, Fix \quad
			\dfrac{}{\Gamma, b' \vdash b'} id
	}
	{\Gamma \vdash b'} Fix}
{\Gamma \vdash    \bigwedge_{j \in J}a_j \twoheadrightarrow b'} Zero
\]

The fragment * of the proof can have two different forms, depending on the set $Z$:
\begin{itemize}
\item if $Z=\emptyset$, we have that $\Gamma = p',b,  ( \bigwedge_{j \in J}a_j \wedge b) \twoheadrightarrow b' $ and

\[
\dfrac
{
	\dfrac{}
	{\Gamma, b',  \bigwedge_{j \in J}a_j \wedge b \vdash  \bigwedge_{j \in J}a_j \wedge b} id
	\quad
	\dfrac{
		\dfrac{
			\Delta'_3
		}
		{
		\Gamma, b' \vdash \bigwedge_{j \in J}a_j
		} \quad
		\dfrac{
		}
		{
		\Gamma, b' \vdash  b
		}  id
	}
	{
	\Gamma, b' \vdash \bigwedge_{j \in J}a_j \wedge b
	} \wedge R
}
{
\Gamma , b' \vdash  \bigwedge_{j \in J}a_j \wedge b
} Fix
\]
Since the inductive  hypothesis guarantees that $p \slash \vec{a} \vdash \lambda(p \slash \vec{a})$ holds and 
\mbox{$\lambda(p \slash \vec{a})$} is a conjunction where $\bigwedge_{j \in J}a_j $ occurs, by applying the reasoning
of the previous case we have a derivation tree $\Delta'_2$ for the sequent
\[
(\, p' , b ,( \bigwedge_{j \in J}a_j \twoheadrightarrow b' )\,) \vdash \bigwedge_{j \in J}a_j \]
As done above, by applying Lemma~\ref{lem:PCLaux} we obtain the derivation tree $\Delta_3'$ for 
\[
(\, \Gamma'', p' ,b,  ( \bigwedge_{j \in J}a_j \wedge b) \twoheadrightarrow b' , b'\,) \vdash  \bigwedge_{j \in J}a_j 
\]

\item if $Z\neq \emptyset$ we obtain:\\

$\dfrac{
			\dfrac{(**)}{\Gamma,b',\bigwedge_{j \in J}a_j \wedge b   \vdash \bigwedge_{z \in Z}c_z  }  \quad
			\dfrac{ 
				\dfrac{ 
					(***)
				}{\Gamma,  b',b \vdash \bigwedge_{j \in J}a_j  }  \quad
				\dfrac{ 
				}{\Gamma,  b',b \vdash  b } id
			}{\Gamma,  b',b \vdash \bigwedge_{j \in J}a_j \wedge b } \wedge R
		}{\Gamma, b' \vdash \bigwedge_{j \in J}a_j \wedge b} Fix$\\
\end{itemize}

\noindent From the induction hypothesis, with the argument used in the previous cases, we prove the following sequent
\[ (p' ,  (\bigwedge_{z \in Z}c_z \twoheadrightarrow b) ,  ( \bigwedge_{j \in J}a_j  \twoheadrightarrow b')) \vdash \bigwedge_{z \in Z}c_z\]
Now, we apply Lemma~\ref{lem:PCLaux} to it, we determine the deduction $(**)$ and a proof for the leftmost sequent above
\[
(p' , b' , \bigwedge_{j \in J}a_j \wedge b , (\bigwedge_{z \in Z}c_z \twoheadrightarrow b) , ( \bigwedge_{j \in J}a_j \wedge b \twoheadrightarrow b') )\vdash \bigwedge_{z \in Z}c_z
\]

Just as done above, from the induction hypothesis we prove the sequent
\[ (p' ,  (\bigwedge_{z \in Z}c_z \twoheadrightarrow b) ,  ( \bigwedge_{j \in J}a_j ) \twoheadrightarrow b') \vdash \bigwedge_{j \in J}a_j \]
from which we obtain the right most sequent above (***), by applying Lemma~\ref{lem:PCLaux}
\[
(p' , b' ,  b , (\bigwedge_{z \in Z}c_z \twoheadrightarrow b) , ( \bigwedge_{j \in J}a_j \wedge b) \twoheadrightarrow b') \vdash \bigwedge_{j \in J}a_j 
\]
\\
\item 

if $p =  p' \wedge (\bigwedge_{k \in K} a_k \wedge b) \wedge  ( \bigwedge_{j \in J}a_j \wedge b \twoheadrightarrow b') $
we have to prove the sequent $p \vdash p \slash \vec{a}$ that reads as
\[(\, p' \wedge (\bigwedge_{k \in K} a_k \wedge b) \wedge  ( \bigwedge_{j \in J}a_j \wedge b \twoheadrightarrow b') \,)
\vdash  
(\, p' \wedge  (\bigwedge_{k \in K} a_k \wedge b)  \wedge  ( \bigwedge_{j \in J}a_j  \twoheadrightarrow b') \,)\]
For readability, we first determine the sequent $\Gamma \vdash  \bigwedge_{j \in J}a_j \twoheadrightarrow b'$ where 
\[\Gamma =  p' , (\bigwedge_{k \in K} a_k \wedge b),  ( \bigwedge_{j \in J}a_j \wedge b \twoheadrightarrow b')\,,\] by applying the rule $\wedge R$ and Lemma~\ref{lem:PCLaux}$(iii)$.
Then we build the following derivation, where * is detailed afterwards:
 \[
 \dfrac{
	\dfrac{ 
			\dfrac{
				*
			}{\Gamma, b' \vdash \bigwedge_{j \in J}a_j \wedge b} \, Fix \quad
			\dfrac{}{\Gamma, b' \vdash b'} id
	}
	{\Gamma \vdash b'} Fix}
{\Gamma \vdash    \bigwedge_{j \in J}a_j \twoheadrightarrow b'} Zero
\]

The fragment * of the proof can have two different forms, depending on the set $K$:
\begin{itemize}
\item if $K=\emptyset$, we have $\Gamma = p',b,  ( \bigwedge_{j \in J}a_j \wedge b) \twoheadrightarrow b' $ and

\[
\dfrac
{
	\dfrac{}
	{\Gamma, b',  \bigwedge_{j \in J}a_j \wedge b \vdash  \bigwedge_{j \in J}a_j \wedge b} id
	\quad
	\dfrac{
		\dfrac{
			\Delta'_3
		}
		{
		\Gamma, b' \vdash \bigwedge_{j \in J}a_j
		} \quad
		\dfrac{
		}
		{
		\Gamma, b' \vdash  b
		}  id
	}
	{
	\Gamma, b' \vdash \bigwedge_{j \in J}a_j \wedge b
	} \wedge R
}
{
\Gamma , b' \vdash  \bigwedge_{j \in J}a_j \wedge b
} Fix
\]
Since the inductive  hypothesis guarantees that $p \slash \vec{a} \vdash \lambda(p \slash \vec{a})$ holds and 
\mbox{$\lambda(p \slash \vec{a})$} is a conjunction where $\bigwedge_{j \in J}a_j $ occurs, by applying the reasoning
of the previous case we have a derivation tree $\Delta'_2$ for the sequent
\[
(\, p' , b ,( \bigwedge_{j \in J}a_j \twoheadrightarrow b' )\,) \vdash \bigwedge_{j \in J}a_j \]
As done above, by applying Lemma~\ref{lem:PCLaux} we obtain the derivation tree $\Delta_3'$ for 
\[
(\, p' ,b,  ( \bigwedge_{j \in J}a_j \wedge b) \twoheadrightarrow b' , b'\,) \vdash  \bigwedge_{j \in J}a_j 
\]

\item if $K \neq \emptyset$ we have that $\Gamma = p',(\bigwedge_{k \in K} a_k \wedge b),  ( \bigwedge_{j \in J}a_j \wedge b) \twoheadrightarrow b' $ and\\

\[
\dfrac
{
	\dfrac{}
	{\Gamma, b',  \bigwedge_{j \in J}a_j \wedge b \vdash  \bigwedge_{j \in J}a_j \wedge b} id
	\quad
	\dfrac{
		\dfrac{
			\Delta'_3
		}
		{
		\Gamma, b' \vdash \bigwedge_{j \in J}a_j
		} \quad
		\dfrac{
			\dfrac{
			}
			{\Gamma, b', b \vdash b}{id}
		}
		{
		   \Gamma, b' \vdash  b
		}  	\wedge L2
	}
	{
	\Gamma, b' \vdash \bigwedge_{j \in J}a_j \wedge b
	} \wedge R
}
{
\Gamma , b' \vdash  \bigwedge_{j \in J}a_j \wedge b
} Fix
\]
Since the inductive  hypothesis guarantees that $p \slash \vec{a} \vdash \lambda(p \slash \vec{a})$ holds and 
\mbox{$\lambda(p \slash \vec{a})$} is a conjunction where $\bigwedge_{j \in J}a_j $ occurs, by applying the reasoning
of the previous case we have a derivation tree $\Delta'_2$ for the sequent
\[
(\, p' , (\bigwedge_{k \in K} a_k \wedge b) ,( \bigwedge_{j \in J}a_j \twoheadrightarrow b' )\,) \vdash \bigwedge_{j \in J}a_j \]
As done above, by applying Lemma~\ref{lem:PCLaux} we obtain the derivation tree $\Delta_3'$ for 
\[
(\, p' ,(\bigwedge_{k \in K} a_k \wedge b),  ( \bigwedge_{j \in J}a_j
\wedge b) \twoheadrightarrow b' , b'\,) \vdash  \bigwedge_{j \in
  J}a_j\rlap{\hbox to 117 pt{\hfill\qEd}} 
\]
\end{itemize}
\end{itemize}

\thepclweakautomata*
\begin{proof}
($\Rightarrow$)
Straightforward from Theorem~\ref{the:1NPCLagreement} and from
$\mathfrak A \subset \mathfrak W$

($\Leftarrow$)
Since $\llbracket p  \rrbracket$  admits weak agreement 
there exists a trace $w \in \mathscr L (\llbracket p \rrbracket)$ where each request
is combined with a corresponding offer.  
For proving $p \vdash \lambda(p)$ we will prove $ p \vdash a$ for all the atoms $a$ in $\lambda(p)$ 
and the thesis follows by repeatedly applying the rule $\wedge R$.
If $a$ occurs within:
\begin{enumerate}
\item $\bigwedge_{j \in J} a_j$: it suffices to apply the rules $\wedge L_1, \wedge L_2, id$;

\item $\bigwedge_{j \in J} a_j \twoheadrightarrow a$:
$p\vdash a$ holds if we prove the sequent
\[
\Gamma,  ( \bigwedge_{j \in J}a_j  \twoheadrightarrow a ) \vdash a
\]
that is obtained from $ p \vdash a$ by repeatedly applying  the rules  $\wedge L_i$, for some $\Gamma$ containing $p$ and sub-formulas of $p$.
The proof of this sequent has the following form:
\[
	\dfrac{  
		\dfrac{*}{\Gamma,  ( \bigwedge_{j \in J}a_j  \twoheadrightarrow a),a \vdash \bigwedge_{j \in J}a_j} \qquad
		\dfrac{}{\Gamma,  ( \bigwedge_{j \in J}a_j  \twoheadrightarrow a),a \vdash a} \,id
	}{\Gamma,  ( \bigwedge_{j \in J}a_j  \twoheadrightarrow a) \vdash a} \,Fix
\]

We prove the sequent in the left premise, it suffices to establish the sequents 
\linebreak
\hbox{$\Gamma,  ( \bigwedge_{j \in J}a_j  \twoheadrightarrow a),a \vdash a_j$}, for all the atoms $a_j$ 
of $ \bigwedge_{j \in J} a_j$.
Then, the derivation proceeds by repeatedly applying the rule $\wedge R$.
We are left to prove \hbox{$\Gamma,  ( \bigwedge_{j \in J}a_j  \twoheadrightarrow a), a \vdash a_j$}, which is done by recursively applying the construction of cases (1) and (2).
This procedure will eventually terminate. 
Indeed, at each iteration $a_j$ is either a conjunct in $\bigwedge_{k \in K} a_k$ (case 1) and the proof is closed  by rule 
$(id)$, or $a_j$ is the conclusion of the contractual implication $\bigwedge_{k \in K} a_k \twoheadrightarrow a_j$ and the proof proceeds as in case (2) by applying the rule 
$(Fix)$.
In the last case, the premise in the left hand-side becomes
$\Gamma',  ( \bigwedge_{k \in K} a_k  \twoheadrightarrow a_j), a, a_j \vdash \bigwedge_{k \in K} a_k$, so adding $a_j$ in the left part of the sequent.
The number of iterations is therefore bound by the number of atoms in $p$.

\item $\bigwedge_{j \in J} a_j \twoheadrightarrow b$ where $a \neq b$.
This case reduces to one of the above two, because if $\exists j \in J$ such that \mbox{$a_j = a$}, then $a$ must also appear in another conjunct $\bigwedge_{z \in Z} a_z$ or  in another contractual implication $\bigwedge_{z \in Z} a_z \twoheadrightarrow a$, otherwise all the traces of $\llbracket p \rrbracket$ would have an unmatched request on $a$, against the hypothesis that it admits weak agreement.\qedhere
\end{enumerate}
\end{proof}

\subsection{Automata and Intuitionistic Linear Logic with Mix}
We recall for completeness the full grammar of \illmix.
\begin{defi}
The formulas $A,B, \ldots$ of \illmix are defined as follows:
\[
A::= a \mid A^\bot \mid A \otimes A \mid A \multimap A\mid A \& A \mid A \oplus A \mid !A \mid 1 \mid 0 \mid \top \mid \bot
\]
\end{defi}
The full sequent calculus for \illmix is displayed in Figure~\ref{fig:fullgentzenlinear}.
We will only consider proofs without the rule $Cut$, which is redundant by~\cite{benton1995mixed}, 
Theorem 24.

The following definition and lemmata are auxiliary. 
\begin{figure}[tb]
\[
\irule{}{A \vdash A}{  \ \ Ax}
\qquad
\irule{\Gamma \vdash \ \ \Gamma' \vdash \gamma}{ \Gamma, \Gamma' \vdash \gamma}{ \ \ Mix}
\qquad 
\irule{\Gamma \vdash A}{\Gamma,A^\bot \vdash}{ \ \ Neg L }
\qquad
\irule{\Gamma, A, B \vdash \gamma}{\Gamma, A \otimes B \vdash \gamma}{ \ \ \otimes L }
\] \\
\[
\irule{\Gamma \vdash A \quad \Gamma' \vdash B}{\Gamma, \Gamma' \vdash A \otimes B}{ \ \ \otimes R }
\qquad
\irule{\Gamma \vdash A \quad \Gamma',B \vdash \gamma}{ \Gamma, \Gamma', A \multimap B \vdash \gamma}{ \ \ \multimap  L }
\qquad
\irule{\Gamma, A \vdash B}{ \Gamma \vdash A \multimap B}{ \ \ \multimap R }
\]
\[
\irule{\Gamma \vdash A \quad \Gamma', A \vdash \gamma}{ \Gamma, \Gamma' \vdash \gamma}{ Cut}
\qquad
\irule{\Gamma, A \vdash}{\Gamma \vdash A^\bot}{NegR}
\qquad
\irule{\Gamma \vdash}{\Gamma \vdash \bot}{\bot R}
\qquad
\irule{}{\bot \vdash }{\bot L}
\]
\[
\irule{}{\vdash 1}{1R}
\qquad
\irule{\Gamma \vdash \gamma}{\Gamma,1 \vdash \gamma}{1L}
\qquad
\irule{}{\Gamma \vdash \top}{\top}
\qquad
\irule{}{\Gamma,0 \vdash A}{0L}
\]
\[
\irule{\Gamma,A \vdash \gamma \quad \Gamma,B\vdash \gamma}{\Gamma, A \oplus B \vdash \gamma}{ \oplus L}
\qquad
\irule{\Gamma \vdash A}{\Gamma \vdash A \oplus B}{ \oplus R1}
\qquad
\irule{\Gamma\vdash B}{\Gamma \vdash A \oplus B}{\oplus R2}
\]
\[
\irule{\Gamma \vdash A \quad \Gamma \vdash B}{\Gamma \vdash A \& B}{ \& R}
\qquad
\irule{\Gamma,A \vdash \gamma}{\Gamma, A \& B \vdash \gamma }{\& L1}
\qquad
\irule{\Gamma,B \vdash \gamma}{\Gamma, A \& B \vdash \gamma}{ \& L2}
\]
\[
\irule{\Gamma, A \vdash \gamma}{\Gamma,!A \vdash \gamma}{!L}
\qquad
\irule{!\Gamma \vdash A}{!\Gamma \vdash !A}{!R}
\qquad
\irule{\Gamma \vdash \gamma}{\Gamma,!A \vdash \gamma}{weakL}
\qquad
\irule{\Gamma,!A,!A \vdash \gamma}{\Gamma,!A \vdash \gamma}{coL}
\]
\caption{The sequent calculus for \illmix}
\label{fig:fullgentzenlinear}
\end{figure}

\begin{lem}\label{lem:honouredrules}
If $\Gamma \vdash Z$ is an honoured sequent, there exists a derivation tree for $\Gamma \vdash Z$ such that:
\begin{itemize}
\item it only uses the rules $Ax,Mix,NegL,\otimes L,\otimes R$ and $\multimap L$ of Figure~\ref{fig:fullgentzenlinear};
\item it is only made of honoured sequents.
\end{itemize}
\end{lem}
\begin{proof}
Recall that we are in the Horn fragment and we only consider cut-free proofs.
Since $Z$ is a positive tensor product (or empty), a simple inspection on the rules in Figure~\ref{fig:fullgentzenlinear} 
suffices to prove the first statement.
The second statement is proved because $Ax,Mix,NegL,\otimes L,\otimes R$ and $\multimap L$ introduce no sequents with negative literals on their right hand-side.
\end{proof}

%
\begin{restatable}{lem}{lemhonoredagreement}\label{lem:honored_agreement}
Let $\Gamma \vdash Z$ be an honoured sequent, then:
\[
\Gamma \vdash Z \text{ implies } \llbracket \Gamma \rrbracket \text{ admits agreement on } Z.
\]
\end{restatable}

\begin{proof}
We will prove that there exists a trace $w \in \mathscr L (\llbracket \Gamma \rrbracket)$
made of matches and as many offers as the literals in $Z=\bigotimes_{a \in Y} a$ (recall that they all are positive), 
or it is made by only matches if $Z$ is empty.
Also, note that the sequents in the proof of $\Gamma \vdash Z$ are all honoured, by hypothesis and Lemma~\ref{lem:honouredrules}.
We proceed by induction on the depth of the proof of  $\Gamma \vdash Z$.


In the base case, the proof consists of a single application of the rule $Ax$. 
By Definition~\ref{def:translation} one first has an offer transition for each $a$ in $Z$, and then interleaves them in any possible order.
Hence the thesis holds trivially.

For the inductive case we proceed by cases on the last rule applied. 
We assume that all clauses (i.e. principals) in $\Gamma$ are divided by commas, which can be easily
obtained by repeatedly applying  the rule $\otimes L$. 
In the following, let $\overline a$ be offers in correspondence with the literals $a$ in $Z$, we will consider only
the relevant rules as stated by Lemma~\ref{lem:honouredrules}.
\begin{itemize}
\item $\irule{\Gamma \vdash  \quad \Gamma' \vdash Z}{\Gamma,\Gamma' \vdash Z}{ \ \  Mix}$
By induction hypothesis there exists $w \in \mathscr L (\llbracket \Gamma \rrbracket)$ with match actions, only,  and
 $w_1 \in \mathscr L ( \llbracket \Gamma' \rrbracket)$ with match actions and offers in correspondence with the literals in 
 $Z$ (if non-empty).   
By Definition~\ref{def:aproduct}, there exists 
$w_2 \in \mathscr L (\llbracket \Gamma \rrbracket \boxtimes \llbracket \Gamma' \rrbracket)$ in agreement.\\

\item $\irule{\Gamma \vdash A }{\Gamma, A^\bot \vdash }{ \ \ NegL}$
By induction hypothesis there exists $w \in \mathscr L (\llbracket \Gamma \rrbracket)$ with match actions, and with offers in correspondence with the literals in $A$.  
By Definition~\ref{def:translation} the traces of the automaton $\llbracket A^\bot \rrbracket$ are all the possible
permutations of the requests in correspondence with the literals in $A^\bot$. 
The thesis follows, because there is an offer for each request, and by Definition~\ref{def:aproduct}.\\

\item $\irule{\Gamma, A, B \vdash Z}{\Gamma, A \otimes B \vdash Z}{ \ \  \otimes L}$ \ 
By the induction hypothesis there exists $w \in \mathscr{L}(\llbracket \Gamma,A,B \rrbracket) = $ \linebreak
\mbox{$\mathscr{L}(\llbracket \Gamma \rrbracket \boxtimes \llbracket A \rrbracket \boxtimes \llbracket B \rrbracket)$} with offers in correspondence with the literals in $Z$ (if non-empty). 
No atom and its negation can occur in $A \otimes B$ by Definition~\ref{def:hornillmix}, because it is a principal.
Hence $\llbracket A \otimes B \rrbracket$ and $\llbracket A, B \rrbracket$ are the same automaton (with a different rank), and the statement follows immediately.
\\


\item $\irule{\Gamma \vdash A \quad \Gamma' \vdash B}{\Gamma, \Gamma' \vdash A \otimes B}{ \ \ \otimes R}$
By the induction hypothesis there exist $w \in \mathscr L (\llbracket \Gamma \rrbracket)$ and 
$w' \in \mathscr L ( \llbracket \Gamma' \rrbracket)$ with only match actions and offers in correspondence with the literals in $A$ and in $B$, respectively.
Now Definition~\ref{def:aproduct} guarantees that there exists a trace in 
$\mathscr{L} ( \llbracket \Gamma \rrbracket \boxtimes  \llbracket \Gamma' \rrbracket ) $ in agreement.
\\

\item $\irule{\Gamma \vdash A \quad \Gamma',B \vdash Z}{ \Gamma, \Gamma', A \multimap B \vdash Z}{ \multimap L}$
By the induction hypothesis there exists $w \in \mathscr L (\llbracket \Gamma \rrbracket)$ and \linebreak
$w' \in \mathscr L (\llbracket \Gamma',B \rrbracket)$ with only match actions and offers in correspondence with the literals in $A$ and in $Z$ (if non-empty), respectively.
By Definition~\ref{def:translation} the literals occurring in $A$ become requests in $\llbracket A \multimap B\rrbracket$, in all possible ordering.
The trace $w$ contains exactly the needed matching offers.
We conclude by noting that no other request is possible in 
$\mathscr L (\llbracket \Gamma, \Gamma', A \multimap B \rrbracket)$.\qedhere
\end{itemize}
\end{proof}


\noindent In order to keep the following definition compact, with a slight abuse of the notation we use $\otimes$ to operate between formulas; we remove the constraints of Definition~\ref{def:hornillmix} on the indexing sets $I$ in formulas and $X_1,X_2$ and $Y$ in clauses; and we let  \mbox{$\bigotimes_{b \in \emptyset} b \multimap  \bigotimes_{a \in X_2} a$} to stand for 
$\bigotimes_{a \in X_2} a$.

\begin{defi}\label{def:meno-alinear}
Given a Horn formula $p$ and an offer or match transition leaving 
the initial state of  $\llbracket p \rrbracket$ with label $\vec a$, then define the formula
$p \slash \vec{a}$ as:
\[
  p \slash \vec{a}  = \left\{ 
    \begin{array}{ll} 
     p' \otimes \bigotimes_{a_1\in X_1 } a_1 & \text{ if $\vec{a}$ is an offer on $c$ and }      \\
	& p =  p' \otimes \bigotimes_{a_1 \in X_1 \cup \{c\}} a_1 \\
       p' \otimes \bigotimes_{a_1 \in X_1 } a_1 \otimes   \bigotimes_{a_2 \in X_2} a_2 & \text{ if $\vec{a}$ is a match on $c$ and }      \\
	 & p =  p' \otimes \bigotimes_{a_1 \in X_1 \cup \{ c\} } a_1 \otimes  \\
	 & \qquad \bigotimes_{a_2 \in X_2 \cup \{c^\bot \}} a_2\\
     p' \otimes \bigotimes_{a_1 \in X_1 } a_1 \otimes  \bigotimes_{b \in Y} b \multimap \bigotimes_{a_2 \in X_2} a_2 & \text{ if $\vec{a}$ is a match on $c$ and }      \\
	& p =  p' \otimes \bigotimes_{a_1 \in X_1 \cup \{c\}} a_1 \otimes   \\
	 & \qquad \bigotimes_{b \in Y \cup \{ c \}} b \multimap \bigotimes_{a_2 \in X_2}a_2\\
    \end{array} 
  \right.
\]
\end{defi}

We now establish a relation between $\llbracket p \slash \vec{a} \rrbracket$,
and the contract automaton
obtained by changing the initial state $\vec q_0$ of $\llbracket p \rrbracket$ 
to $\vec q$, for  the transition $(\vec q_0, \vec a, \vec q)$ of $\llbracket p \rrbracket$. Without
loss of generality we assume that the automaton obtained from Definition~\ref{def:translation} is 
deterministic. 
If not, we first transform the non deterministic automaton to a
deterministic one.

\begin{restatable}{lem}{lemautomataformula}\label{lem:automata2formula}
Given a Horn formula  $p$ and the contract automaton
  $\llbracket p \rrbracket=\langle Q, \vec{q_0},A^r,A^o,T,F \rangle$, 
if $t=(\vec{q_0},\vec{a},\vec{q}) \in T$ is an offer or a match transition,  
then $\mathscr{L}(\mathcal{A}) =\mathscr{L}(\llbracket p \slash \vec{a} \rrbracket)$, where  
$\mathcal{A}=\langle Q, \vec{q},A^r,A^o,T,F \rangle$.
\end{restatable}
\begin{proof}
The proof is similar to the one of Lemma~\ref{lem:automata_traduction}.
The statement follows by noting that in Definition~\ref{def:translation} a tensor product is translated in all the possible
permutations of actions corresponding to the literals, and noting that in $p \slash \vec{a}$ we remove exactly the actions fired in $\vec a$, that are therefore not 
available any more in the state $\vec q$.
%
%
%
\end{proof}

The following lemma suggests that we can safely substitute a multi-set of Horn formulae and clauses $\Gamma$ with a
single Horn formula, without affecting the corresponding automaton.

\begin{lem}\label{lem:gammaformula}
Let $\Gamma$ be a non-empty multi-set of Horn formulae, then there exists a Horn formula $p$ such that:
\[
\llbracket \Gamma \rrbracket = \llbracket p \rrbracket \quad 
\]
\end{lem}
\begin{proof}
Immediate from Definition~\ref{def:translation} (recall that we abuse the notation).
\end{proof}

We now prove the following lemma.

\begin{restatable}{lem}{lemagreementhonoured}
\label{lem:agreement_honoured}
Let $\Gamma \neq \emptyset$ be a multi-set of Horn formulae and $Z$ be a positive tensor product or empty.
Then 
\[
\llbracket \Gamma \rrbracket \text{ admits agreement on } Z \text{ implies } \Gamma \vdash Z \text{ is an honoured sequent}
\] 
\end{restatable}
\proof
By hypothesis $w \in \mathscr L (\llbracket \Gamma \rrbracket)$ is a trace only composed of match and offer actions on $Z$.
We proceed by induction on the length of $w$. 
In the base case $w$ has length one. 
Note that it is not possible to have $w=\epsilon$ by the hypothesis $\Gamma \neq \emptyset$ and 
Definition~\ref{def:hornillmix}. Moreover by Definition~\ref{def:hornillmix} it must be that 
$w=\vec a$ where $\vec a$ is a match on action $a$ (a Horn formula must contain at least two principals).  Hence
by Definition~\ref{def:translation} it must be that $Z=\emptyset$ and
 $\Gamma = \{ \alpha \otimes \alpha'\}$ where $\alpha=a$ and $\alpha'=a^\bot$ for some literal $a$.  
Then we have:
\[
\dfrac
{
	\dfrac
	{
		\dfrac{}{a \vdash a} Ax
	}
	{a,a^\bot \vdash} Neg
}
{a \otimes a^\bot \vdash} \otimes L
\]
%

For the inductive step, let $w=\vec{a}w_2$, let $\vec{q_0}$ and $\vec{f}$ be the initial and the final states of $\llbracket \Gamma \rrbracket$, then $(\vec{a}w_2,\vec{q_0}) \rightarrow (w_2,\vec{q}) \rightarrow^+ (\epsilon,\vec{f})$.
Let $p$ be a Horn formula such that 
$\llbracket \Gamma \rrbracket = \llbracket p \rrbracket$ (Lemma~\ref{lem:gammaformula}), so it suffices to prove $p \vdash Z$.
 By the induction hypothesis and Lemma~$\ref{lem:automata2formula}$ we have that 
$\llbracket p \slash \vec{a} \rrbracket$ admits agreement on some $Z'$ implies 
$p \slash \vec{a} \vdash Z'$ is honoured.
To build $Z$ from $Z'$, we proceed by cases on $\vec a$:
\begin{itemize}
\item if $\vec a$ is an offer action on c we prove that   $p \vdash Z$ 
where $Z= Z' \otimes c$. 
We have the following
\[
\dfrac
{
	\dfrac{\Delta'}{p \vdash (p \slash \vec{a}) \otimes c}{}
	\quad
	\dfrac
	{
		\dfrac{\Delta}{ (p \slash\vec a) \vdash Z' }
		\quad
		\dfrac{}{c \vdash c}{Ax}
	}
	{p \slash \vec{a} \otimes c \vdash Z' \otimes c} \otimes R
}
{
p \vdash Z
}{cut}
\]
where $\Delta$ is obtained by the inductive hypothesis and for $\Delta'$ we have two cases depending on $p$:

\begin{itemize}
\item 
$p = p' \otimes \bigotimes_{a_1 \in X_1 \cup \{c\}}a_1$
then the derivation  $\dfrac{\Delta'}{p \vdash (p \slash \vec{a} ) \otimes c}$  becomes 
$
\dfrac
{
}
{
 p' \otimes c \vdash  p' \otimes c
} Ax
$\\
if $X_1 = \emptyset$ and the following otherwise

\[
\dfrac
{
	\dfrac
	{
		\dfrac{
			\dfrac{
				\dfrac{}{p' \vdash p' }{ Ax }
				\quad
				\dfrac{}{  \bigotimes_{a_1 \in X_1 } a \vdash \bigotimes_{a_1 \in X_1} a }{ Ax }
			}
			{p',   \bigotimes_{a_1 \in X_1 } a \vdash p' \otimes   \bigotimes_{a_1 \in X_1} a }{\otimes R}
			\quad				
			\dfrac{}
			{c \vdash  c}{Ax}
		}
		{p',   \bigotimes_{a_1 \in X_1} a, c \vdash p' \otimes   \bigotimes_{a_1 \in X_1} a \otimes c}{\otimes R}
	}
	{p', \bigotimes_{a_1 \in X_1\cup \{c\} } a \vdash p' \otimes   \bigotimes_{a_1 \in X_1} a \otimes c}{\otimes L}
}
{
	p' \otimes \bigotimes_{a_1 \in X_1 \cup \{c\} } a \vdash p' \otimes   \bigotimes_{a_1 \in X_1} a \otimes c
}{\otimes L}
\]

\end{itemize}

%
%
%
%
%
%

\item if $\vec a$ is a match action we prove that   $p \vdash Z$. 
We have the following
\[
\dfrac
{
	\dfrac{\Delta'}{p \vdash p \slash \vec{a}}
	\quad
	\dfrac
	{
		\Delta
	}
	{p \slash \vec{a} \vdash Z'} 
}
{
p \vdash Z'
}{cut}
\]
where $\Delta$ is obtained by the inductive hypothesis, $Z = Z'$ because $\vec a$ is a match, and for $\Delta'$ we have eight   cases depending on $p$:
\begin{itemize}
\item $p  =  p' \otimes c \otimes  c^\bot$; 
then the derivation $\dfrac{\Delta'}{p \vdash p \slash \vec{a}}$ becomes:
$
\dfrac
{
\dfrac{\Delta_{mix}}
{p' , c \otimes  c^\bot  \vdash p'} \otimes L
}
{p' \otimes c \otimes  c^\bot  \vdash p'} \otimes L$\\\\
%
Since the deduction tree $\Delta_{mix}$ will be also used later on, we keep it more general, by writing $q$ for $p'$:
\[
\Delta_{mix}\quad =\quad
\dfrac
{
	\dfrac{
		\dfrac{}{c \vdash c} Ax
	}
	{c,c^\bot \vdash} NegL
	\quad
	\dfrac{}{q\vdash q}id
}
{q , c ,  c^\bot  \vdash q} Mix
\]

\item $ p =  p' \otimes \bigotimes_{a_1 \in X_1 \cup \{ c\} } a_1 \otimes  c^\bot $\\ 
then, writing in $\Delta_{mix}$ $p' \otimes \bigotimes_{a_1 \in X_1  } a_1$ for $q$ the derivation 
$\dfrac{\Delta'}{p \vdash p \slash \vec{a}}$ becomes:
\[
\dfrac
{{
\dfrac{\Delta_{mix}}
{ p' \otimes \bigotimes_{a_1 \in X_1 \cup \{ c\} } a_1 ,  c^\bot  \vdash  p' \otimes \bigotimes_{a_1 \in X_1  } a_1} \otimes L
}
}
{
 p' \otimes \bigotimes_{a_1 \in X_1 \cup \{ c\} } a_1 \otimes  c^\bot  \vdash  p' \otimes \bigotimes_{a_1 \in X_1  } a_1 
} \otimes L 
\]

\item $ p =  p' \otimes \bigotimes_{a_2 \in X_2 \cup \{c^\bot \}} a_2 \otimes c$\\ 

then, writing in $\Delta_{mix}$ $p' \otimes \bigotimes_{a_2 \in X_2 } a_2$ for $q$ the derivation
$\dfrac{\Delta'}{p \vdash p \slash \vec{a}}$ becomes:
\[
\dfrac
{
{
\dfrac{
\Delta_{mix}
}
{p' \otimes  \bigotimes_{a_2 \in X_2 \cup \{c^\bot \}} a_2 , c  \vdash 
 p' \otimes \bigotimes_{a_2 \in X_2 } a_2} \otimes L
}
}
{
 p'  \otimes  \bigotimes_{a_2 \in X_2 \cup \{c^\bot \}} a_2  \otimes c\vdash 
 p' \otimes \bigotimes_{a_2 \in X_2 } a_2
} \otimes L 
\]

\item $ p =  p' \otimes \bigotimes_{a_1 \in X_1 \cup \{ c\} } a_1 \otimes  \bigotimes_{a_2 \in X_2 \cup \{c^\bot \}} a_2  $\\ 

then, writing in $\Delta_{mix}$ $p' \otimes \bigotimes_{a_1 \in X_1  } a_1 \otimes  \bigotimes_{a_2 \in X_2 } a_2$ for $q$ the derivation
$\dfrac{\Delta'}{p \vdash p \slash \vec{a}}$ becomes:
\[
\dfrac
{ \dfrac{\Delta_ {mix}}
  {p' \otimes \bigotimes_{a_1 \in X_1 \cup \{ c\} } a_1 \otimes  \bigotimes_{a_2 \in X_2} a_2 , c^\bot  \vdash
p' \otimes \bigotimes_{a_1 \in X_1  } a_1 \otimes  \bigotimes_{a_2 \in X_2 } a_2 } \otimes L
}
{p' \otimes \bigotimes_{a_1 \in X_1 \cup \{ c\} } a_1 \otimes  \bigotimes_{a_2 \in X_2 \cup \{c^\bot \}} a_2  \vdash
p' \otimes \bigotimes_{a_1 \in X_1  } a_1 \otimes  \bigotimes_{a_2 \in X_2 } a_2  } \otimes L 
\]

\item $ p =  p' \otimes c \otimes (c \multimap \bigotimes_{a_2 \in X_2}a_2 ) $\\ 

then the derivation $\dfrac{\Delta'}{p \vdash p \slash \vec{a}}$ becomes:
\[
\dfrac
{
\dfrac
{
		\Delta_{ax} 
		\quad
		\Delta_{\multimap}
	}
	{
	 p' , c , (c \multimap \bigotimes_{a_2 \in X_2}a_2 ) \vdash
	 p' \otimes \bigotimes_{a_2 \in X_2}a_2  
	} \otimes R 
}
{
	 p' \otimes c \otimes (c \multimap \bigotimes_{a_2 \in X_2}a_2 ) \vdash
	 p' \otimes \bigotimes_{a_2 \in X_2}a_2  
} \otimes L (x2)
\]
where letting $q=p'$ \\
\[
\Delta_{ax} \quad = \quad \dfrac{}{q \vdash q} \, Ax
\]

and $\Delta_\multimap$ is the following proof:
\[
\dfrac{ 
			\dfrac{}{c \vdash c} \,Ax
			\quad
			\dfrac{}{ \bigotimes_{a_2 \in X_2}a_2 \vdash  \bigotimes_{a_2 \in X_2}a_2}Ax
		}{ c , (c \multimap \bigotimes_{a_2 \in X_2}a_2 ) \vdash   \bigotimes_{a_2 \in X_2}a_2  } \multimap L
\]

\item $p =  p' \otimes \bigotimes_{a_1 \in X_1 \cup \{c\}} a_1 \otimes (c \multimap \bigotimes_{a_2 \in X_2}a_2 )$\\ 

then, letting in $\Delta_{ax}$ $q= p' \otimes \bigotimes_{a_1 \in X_1 } a_1$,  
the derivation $\dfrac{\Delta'}{p \vdash p \slash \vec{a}}$ becomes:
\[
\dfrac{
	\dfrac{ 
		\Delta_{ax}
		\quad
		\Delta_{\multimap}
	}
	{
		 p' \otimes \bigotimes_{a_1 \in X_1} a_1, c, (c \multimap \bigotimes_{a_2 \in X_2}a_2) \vdash 
 		p' \otimes \bigotimes_{a_1 \in X_1 } a_1 \otimes  \bigotimes_{a_2 \in X_2}a_2
	} \otimes R 
}
{
 p' \otimes \bigotimes_{a_1 \in X_1 \cup \{c\}} a_1 \otimes (c \multimap \bigotimes_{a_2 \in X_2}a_2) \vdash 
 p' \otimes \bigotimes_{a_1 \in X_1 } a_1 \otimes  \bigotimes_{a_2 \in X_2}a_2
} \otimes L \text{ --- twice}
\]

\item $ p =  p' \otimes \bigotimes_{a_1 \in X_1 \cup \{c\}} a_1 \otimes (\bigotimes_{b \in Y \cup \{ c \}} b \multimap \bigotimes_{a_2 \in X_2}a_2 )$\\ 

then writing 
$\hat q$ for $ p' \otimes \bigotimes_{a_1 \in X_1 } a_1 \otimes
(\bigotimes_{b \in Y } b \multimap \bigotimes_{a_2 \in X_2}a_2 )$ below, the derivation $\dfrac{\Delta'}{p \vdash p \slash \vec{a}}$ becomes:

\[
\dfrac
{
	\dfrac
	{
		\Delta_{ax}
		\quad
		\Delta_{\multimap 2}
	}
	{
	 p' \otimes \bigotimes_{a_1 \in X_1 }  a_1  , c , (\bigotimes_{b \in Y \cup \{ c \}} b \multimap \bigotimes_{a_2 \in X_2}a_2 )\vdash
	 \hat q
	} \otimes R
}
{
 p' \otimes \bigotimes_{a_1 \in X_1 \cup \{c\}} a_1 \otimes (\bigotimes_{b \in Y \cup \{ c \}} b \multimap \bigotimes_{a_2 \in X_2}a_2 )\vdash
 \hat q
}\otimes L (x2)
\]

where  $q = p' \otimes \bigotimes_{a_1 \in X_1 } a_1$  in $\Delta_{ax}$, and $\Delta_{\multimap 2}$ is the deduction tree below:
\[
\dfrac
{
	\dfrac
	{
		\dfrac
		{
			\dfrac{}{ c \vdash c} Ax
			\quad
			\dfrac{}{ \bigotimes_{b \in Y }  b  \vdash \bigotimes_{b \in Y }  b  }Ax
		}
		{
			 c ,  \bigotimes_{b \in Y }  b \vdash\bigotimes_{b \in Y \cup \{ c \}} b 
		} \otimes R
		\quad
		\dfrac
		{}{ \bigotimes_{a_2 \in X_2}a_2  \vdash  \bigotimes_{a_2 \in X_2}a_2 } Ax
	}
	{
		 c , (\bigotimes_{b \in Y \cup \{ c \}} b \multimap \bigotimes_{a_2 \in X_2}a_2 ), \bigotimes_{b \in Y }  b \vdash
		\bigotimes_{a_2 \in X_2}a_2 
	} \multimap L
}
{
	 c , (\bigotimes_{b \in Y \cup \{ c \}} b \multimap \bigotimes_{a_2 \in X_2}a_2 )\vdash
	 \bigotimes_{b \in Y } b \multimap \bigotimes_{a_2 \in X_2}a_2 
} \multimap R
\]

\item $p =  p' \otimes c \otimes (\bigotimes_{b \in Y \cup \{ c \}} b \multimap \bigotimes_{a_2 \in X_2}a_2 ) $\\ 

then, letting $q= p'$ in $\Delta_{ax}$, the derivation $\dfrac{\Delta'}{p \vdash p \slash \vec{a}}$ becomes:
\[
\dfrac
{
	\dfrac
	{
		\Delta_{ax}
		\quad
		\Delta_{\multimap 2}
	}
	{
	 p', c \otimes (\bigotimes_{b \in Y \cup \{ c \}} b \multimap \bigotimes_{a_2 \in X_2}a_2  )
	\vdash
	 p' \otimes (\bigotimes_{b \in Y} b \multimap \bigotimes_{a_2 \in X_2}a_2)
	} \otimes R
}
{
 p' \otimes c \otimes (\bigotimes_{b \in Y \cup \{ c \}} b \multimap \bigotimes_{a_2 \in X_2}a_2  )
\vdash
 p' \otimes (\bigotimes_{b \in Y} b \multimap \bigotimes_{a_2 \in X_2}a_2)
} \otimes L\rlap{\hbox to 55 pt{\hfill\qEd}}
\]
\end{itemize}
\end{itemize}\medskip

\noindent The main theorem of this sub-section has now an immediate proof.
\theautomaillmix*
\begin{proof}
By Lemmata~\ref{lem:honored_agreement} and~\ref{lem:agreement_honoured}. 
\end{proof}

\end{document}